%% file: acc.tex
\documentclass[11pt, DIV10,a4paper]{article}
\usepackage{color}
\usepackage{float}
\usepackage[bf,small]{caption2}
\usepackage{comment}
\usepackage{amsmath}
\usepackage{bigints}
\usepackage{rotating, booktabs}
\usepackage{amsopn}
\usepackage{amsfonts}
\usepackage{amsthm}
\usepackage{amssymb}
\usepackage{amsbsy}
\usepackage{multirow}
\usepackage{slashbox}
\usepackage{natbib}
\usepackage{tocbibind}
\usepackage[a4paper,colorlinks,breaklinks,bookmarksopen,bookmarksnumbered]{hyperref}
\usepackage[refpage]{nomencl}
\usepackage{makeidx}
\usepackage{pst-all}
\usepackage{epsfig,psfrag}
\usepackage{graphicx}
\usepackage{amsmath}
\usepackage{epstopdf}
\usepackage{tabularx}
\usepackage{subfigure}
\usepackage{setspace}
\usepackage[mathscr]{euscript}
\usepackage[margin=1in]{geometry}
\newcommand{\bzero}{\boldsymbol{0}}
\newcommand{\bone}{\boldsymbol{1}}
\newcommand{\bmu}{\boldsymbol{\mu}}
\newcommand{\bSigma}{\boldsymbol{\Sigma}}
\newcommand{\blambda}{\boldsymbol{\lambda}}
\newcommand{\bI}{\bf{I}}

\newcommand{\bx}{\bold{x}}
\newcommand{\bz}{\bold{z}}

\newcommand {\ctn}{\cite}

\usepackage{url}
\newtheorem{corollary}{Corollary}[section]
\newtheorem{theorem}{Theorem}[section]

\normalsize
\begin{document}

\title{\textbf{A Brief Tutorial on Transformation Based Markov Chain Monte Carlo and
Optimal Scaling of the Additive Transformation}}
\author{ Kushal Kr. Dey$^{\dag}$ , Sourabh Bhattacharya$^{\ddag +}$ }
\date{}
\maketitle
\begin{center}
$^{\dag}$   Department of Statistics, University of Chicago, USA \\
$^{\ddag}$   Interdisciplinary Statistical Research Unit, Indian Statistical Institute, Kolkata \\
$+$ Corresponding author:  \href{mailto: bhsourabh@gmail.com}{bhsourabh@gmail.com}
\end{center}

\begin{abstract}

We consider the recently introduced Transformation-based Markov Chain
Monte Carlo (TMCMC) (\ctn{Dutta11}), a methodology that is designed to update all the parameters 
simultaneously using some simple deterministic transformation of a one-dimensional random variable
drawn from some arbitrary distribution on a relevant support. The additive transformation based 
TMCMC is similar in spirit to random walk Metropolis, except the fact that unlike the latter, additive
TMCMC uses a single draw from a one-dimensional proposal distribution to update the high-dimensional
parameter.
In this paper, we first provide a brief tutorial on TMCMC, exploring its connections and contrasts
with various available MCMC methods. 

Then we study the diffusion limits of additive TMCMC under 
various set-ups ranging from the product structure of the target density
to the case where the target is absolutely continuous with respect to a Gaussian measure; we also consider
the additive TMCMC within Gibbs approach for all the above set-ups.  
These investigations lead to appropriate scaling of the 
one-dimensional proposal density. We also show that 
the optimal acceptance rate of additive TMCMC is 0.439 under all the aforementioned set-ups, 
in contrast with the well-established 0.234 
acceptance rate associated with optimal random walk Metropolis algorithms under the same set-ups.
We also elucidate the ramifications of our results and clear advantages of additive TMCMC over 
random walk Metropolis with ample simulation studies and Bayesian analysis of a real, spatial dataset with which $160$
unknowns are associated. 
\\[2mm]

\end{abstract}

\section{Introduction}

Markov Chain Monte Carlo (MCMC), particularly, the Metropolis-Hastings (MH) methods, 
have revolutionized Bayesian computation -- this pleasing truth, however, is often hard
to appreciate in the face of the challenges posed by computational complexities and convergence issues of traditional MCMC.
Indeed, exploration of very high-dimensional posterior distributions using MCMC can be both computationally very expensive
and troublesome convergence-wise. 
Thus, there seems to be trade-off between the great flexibility of MCMC algorithms (see, for example,
\ctn{Storvik11}, \ctn{Martino13}) and choice of the right MCMC algorithm that ensures good convergence properties
and reasonable computational complexity. Investigation of connections between varieties of available MCMC methods,
as provided in the aforementioned papers, seems to be important to decide upon a suitable MCMC algorithm,
given any paricular problem at hand.

The random walk Metropolis (RWM) algorithm is a popular MH algorithm because
of its simplicity and ease in implementation, but unless great care is taken to properly scale the proposal distribution
the algorithm can have poor convergence properties. For instance, if the variance of the proposal density is small, then the
jumps will be small in magnitude, implying that the Markov chain will require a large number of iterations to explore the entire
state-space. On the other hand, large variance of the proposal density causes too many rejections of the proposed moves, again
considerably slowing down convergence of the underlying Markov chain.
The need for an optimal choice of the proposal variance is thus inherent in the RWM algorithms.
The pioneering approach towards obtaining an optimal scaling of the RWM proposal is due to \ctn{Roberts97a} in the case
of target densities associated with independent and identical ($i.i.d.$) random variables; generalization of
this work to more general set-ups are provided by \ctn{Bedard2007} (target density associated with independent but
non-identical random variables) and \ctn{Pillai2011} (target density absolutely continuous with respect to a Gaussian measure).
The approach used in all these works is to study the diffusion approximation of the high-dimensional RWM algorithm, and
maximization of the speed of convergence of the limiting diffusion. The optimal scaling, the optimal
acceptance rate and the optimal speed of convergence of the limiting diffusion, along with the complexity of the algorithm
are all obtained from this powerful approach.

In practice, a serious drawback of the RWM algorithm in high dimensions is that there is always a positive 
probability that a particular
coordinate of the high-dimensional random variable is ill-proposed; in that case the acceptance ratio will tend to be
extremely small, prompting rejection of the entire high-dimensional move. In general, unless the high-dimensional 
proposal distribution, which need not necessarily be a random walk proposal distribution, is designed with extreme care, 
such problem usually persists. Unfortunately, such carefully designed proposal density is rare in high dimensions.
To combat these difficulties \ctn{Dutta11} proposed an approach where the entire block of parameters can be updated 
simultaneously
using some simple deterministic transformation of a scalar random variable sampled from some arbitrary distribution
defined on some suitable support. The strategy effectively reduces the high-dimensional
proposal distribution to a one-dimensional proposal, greatly improving the acceptance rate 
and computational speed in the process. 
This methodology is no longer Metropolis-Hastings 
for dimensions greater than one; 
the proposal density in more than one dimension becomes singular because it is induced by a 
one-dimensional random variable. However,
in one-dimensional cases this coincides with Metropolis-Hastings with a specialized mixture proposal density;
in particular, the additive transformation based TMCMC coincides with RWM in one-dimensional situations. 
\ctn{Dutta11} refer to this new general methodology as Transformation-based MCMC (TMCMC). 
In their work the authors point out several advantages of the additive transformation in comparison with
the other valid transformations. For instance, they show that 
additive TMCMC requires less number of `move-types' compared to other valid transformations; 
moreover, the acceptance rate has
a simple form for additive transformations since the Jacobian of additive transformations is 1. 

The contribution of this paper is two-fold. First, we provide a brief tutorial on TMCMC,
attempting to convey the key ideas and the properties in simple terms and from several perspectives.
We explore various connections and contrasts with existing MCMC algorithms.

Second, we investigate the diffusion limits of additive TMCMC in 
high-dimensional situations under various forms 
of the target density when the one-dimensional random variable used for the additive transformation is drawn
from a left truncated zero-mean normal density. In particular, we consider situations when the target density 
corresponds to $i.i.d.$ random variables, 
independent but non-identically distributed random variables; we also study the diffusion limit 
of additive TMCMC when
the target is absolutely continuous with respect to a Gaussian measure. Since all these forms are considered 
in the MCMC literature
related to diffusion limits and optimal scaling of RWM, comparisons of our additive TMCMC-based approaches 
can be made with the
respective RWM-based approaches. Furthermore, in each of the aforementioned set-ups, 
we also consider additive TMCMC within Gibbs approach, where one or multiple 
components of the high-dimensional random variable are updated by additive TMCMC, conditioning 
on the remaining components.
This we compare with the corresponding RWM within Gibbs approach 
under the same settings of the target densities. 

Briefly, our scaling investigations show that the optimal additive TMCMC acceptance rate in all the set-ups is 
0.439, as opposed to
0.234 associated with RWM. Moreover, we point out that even though the optimal diffusion speed of RWM 
is slightly greater
than that of additive TMCMC, the diffusion speed associated with additive TMCMC is more robust with respect to 
the choice of the scaling constant. In other words, if the optimal scaling constant for RWM is somewhat altered, 
this triggers
a sharp fall in the diffusion speed, but in the case of additive TMCMC the rate of decrease of diffusion speed is much slower.
Investigation of the consequences of this phenomenon with simulation studies reveal severe decline in the 
performance 
of RWM in comparison with additive TMCMC.

This non-robustness of RWM with respect to scale choices other than
the optimal, presents quite important consequences for applied MCMC practitioners, which we elaborate
with a real, spatial data analysis problem. In a nutshell, in the context of the spatial problem,
we have provided a method, which appears to be generally applicable, for approximately
achieving 44\% and 23\% acceptance rates for additive TMCMC and RWM;
however, achieving the desired acceptance rates in general problems where optimal scaling theories are yet
lacking, does not guarantee that the achieved acceptance rates 
correspond to optimal scales, as there are usually very many scale choices corresponding to the same
acceptance rate. Because of such sub-optimality, in the real spatial problem, 
RWM faces very serious performance problems. On the other hand, additive TMCMC, because of its robustness
with respect to the scales, performs quite reasonably.

Our paper is structured as follows.
In Section \ref{sec:overview_tmcmc},
we provide a brief tutorial on TMCMC.  
We develop the theory for optimal additive TMCMC scaling in the $i.i.d.$ set-up in
Section \ref{sec:iid_prod}; in the same section
(Section \ref{subsec:tmcmc_within_gibbs_iid})
we also develop the corresponding theory for
additive TMCMC within Gibbs in the $i.i.d.$ situation. 
In Section \ref{sec:non_iid}, we extend the
additive TMCMC-based optimal scaling theory to the independent but non-identical set-up;
in Section \ref{subsec:tmcmc_within_gibbs_non_iid} we outline the corresponding TMCMC within Gibbs case. 
We then further extend our additive TMCMC based optimal scaling theory 
to the aforementioned dependent set-up in Section \ref{sec:generalization}, presenting
the formal result in Section \ref{subsec:formal_statement}; the
corresponding TMCMC within Gibbs case is considered in Section \ref{subsec:tmcmc_within_gibbs_dependent}.
In Section \ref{sec:comparison}, we provide numerical comparisons between additive TMCMC and RWM 
in terms of optimal acceptance rates and
diffusion speeds; in Section \ref{sec:simulation}, we illustrate our theoretical results and compare the
performances of additive TMCMC and RWM
using simulation studies, illustrating that the former is a far more effective algorithm in comparison
with the latter. 
In Section \ref{sec:comparison_real}, we compare additive TMCMC with RWM with respect to a $160$-dimensional 
posterior density associated with a real, spatial dataset, vividly demonstrating the clear superiority
of additive TMCMC over RWM.
Finally, we make concluding remarks in Section \ref{sec:conclusion}.

Apart from the main developments provided in this article, we provide additional details in our 
supplementary material \ctn{Dey15b}, whose sections and figures have the prefix
``S-" when referred to in this article.
Briefly, in Section S-1, we provide details on computational efficiency of TMCMC. Specifically, we demonstrate 
with an experiment the superior computational speed of additive TMCMC in comparison with RWM, particularly
in high dimensions. In Section S-2 we discuss, with appropriate experiments, the necessity of optimal scaling 
in additive TMCMC, while
in Sections S-3 and S-4 we delve into the robustness issues associated with the scale choices of additive TMCMC and RWM.
In Section S-5, we include brief discussions of adaptive versions of RWM and TMCMC. 
Moreover, the proofs of all our technical results are provided in Sections S-6 and S-7 of the supplement.

\section{A brief overview of TMCMC}
\label{sec:overview_tmcmc}

Suppose that we are simulating from a $d$ dimensional space (usually $\mathbb{R}^{d}$, where $\mathbb R$ is the
real line), and suppose we are currently at a point 
$x= (x_{1}, \ldots, x_{d})$.
Let us define the $d$-dimensional random vector $b=(b_{1}, \ldots, b_{d})$, such that, for $i=1,\ldots,d$, 
\begin{equation}
b_{i} =\left\{\begin{array}{ccc} +1 & \mbox{with probability} & p_i; \\
  0 & \mbox{with probability} & 1-p_i-q_i;\\
 -1 & \mbox{with probability} & q_i,
 \end{array}\right.
 \label{eq:b}
\end{equation}
where, for each $i$, $0<p_i,q_i<1$ such that $p_i+q_i\leq 1$.
Let $\epsilon\sim \varrho(\epsilon)=\tilde\varrho(\epsilon)I_{\mathbb S}(\epsilon)$, 
where $\tilde\varrho (\cdot)$ is any arbitrary density
supported on some suitable space $\mathbb S$; here $I_{\mathbb S}(\cdot)$ denotes the indicator function of $\mathbb S$.

TMCMC uses moves of the following type:
\begin{equation}
(x_{1}, \ldots, x_{d}) \rightarrow (T^{b_1}(x_{1},\epsilon), \ldots, T^{b_d}(x_{d},\epsilon)), 
\label{eq:tmcmc_move}
\end{equation}
where $T^{+1}(x_i,\epsilon)$, the forward transformation to coordinate $x_i$, and $T^{-1}(x_i,\epsilon)$, the backward
transformation to $x_i$, are bijective for fixed $\epsilon$ and injective  
for fixed $x_i$, satisfying 
\begin{equation}
T^{+1}(T^{-1}(x_i,\epsilon),\epsilon)=T^{-1}(T^{+1}(x_i,\epsilon),\epsilon)=x_i.
\label{eq:transformation1}
\end{equation}
The transformation  
\begin{equation}
T^{0}(x_i,\epsilon)\equiv x_i,~\forall\epsilon\in\mathbb S, 
\label{eq:transformation2}
\end{equation}
indicates no change to 
the coordinate $x_i$ while updating the vector $x=(x_{1}, \ldots, x_{d})$ to $x^*=\mathcal T_{b}(x,\epsilon)$,
where $\mathcal T_{b}(x,\epsilon)$ denotes the updated vector $(T^{b_1}(x_{1},\epsilon), \ldots, T^{b_d}(x_{d},\epsilon))$.
Assuming for simplicity of illustration that $p_i=q_i$ for $i=1,\ldots,d$, 
move (\ref{eq:tmcmc_move}) is to be accepted with probability
\begin{equation}
\alpha=\min\left\{1,\frac{\pi(x^*)}{\pi(x)}J^b(x,\epsilon)\right\},
\label{eq:acc_tmcmc_general}
\end{equation}
where $J^b(x,\epsilon)=\left|\frac{\partial(\mathcal T^b(x,\epsilon),\epsilon)}{\partial(x,\epsilon)}\right|$
is the Jacobian of the transformation associated with $\mathcal T^b$.
For general $(p_1,\ldots,p_d)$ and $(q_1,\ldots,q_d)$, the acceptance ratio depends upon 
these probabilities; see \ctn{Dutta11}.

\subsection{Detailed balance}
\label{subsec:detailed_balance}
In the supplement to \ctn{Dutta11} the proof of detailed balance has been provided, but here we refurnish the
proof with more details and with more intuitive discussion. 
For the purpose of detailed balance, we need the following definition of ``conjugate" $b^c=(b^c_1,\ldots,b^c_d)$
of the random vector $d$: 
\begin{equation}
b^c_{i} =\left\{\begin{array}{ccc} +1 & \mbox{with probability} & q_i; \\
  0 & \mbox{with probability} & 1-p_i-q_i;\\
 -1 & \mbox{with probability} & p_i,
 \end{array}\right.
 \label{eq:b_conjugate}
\end{equation}
This definition is needed for returning to $x$ from $x^*$, so that moving from $x$ to $x^*$ using the transformation
$\mathcal T^b$ has, in essence, the same probability as returning from $x^*$ to $x$ using the transformation $\mathcal T^{b^c}$.
Details are provided below; at this point we note that for the $i$-th coordinate $x_i$,
the probability of making a forward move to $x^*_i$ using $T^{+1}(x_i,\epsilon)$ is $p_i$, which is also the probability
of returning from $x^*_i$ to $x_i$ using the backward move $T^{-1}(x^*_i,\epsilon)=T^{(+1)^c}(x^*_i,\epsilon)$.

Letting $K$ denote the Markov transition kernel associated with TMCMC, note that
for moving from $x$ to $x^*$, the kernel satisfies 
\begin{align}
\pi(x)K(x\rightarrow x^*)&=\pi(x)P(b)\varrho(\epsilon)\min\left\{1,\frac{\pi(x^*)}{\pi(x)}J^b(x,\epsilon)\right\}\notag\\
&=\min\left\{\pi(x)P(b)\varrho(\epsilon),P(b)\varrho(\epsilon)\pi(x^*)J^b(x,\epsilon)\right\},
\label{eq:detailed_balance1}
\end{align}
where $P(b)$ is the probability of $b$ responsible for the movement of the underlying Markov chain from $x$ to $x^*$.
For returning from $x^*$ to $x$, the kernel satisfies
\begin{align}
\pi(x^*)K(x^*\rightarrow x)&=\pi(x^*)P(b^c)\varrho(\epsilon)J^{b}(x,\epsilon)
\min\left\{1,\frac{\pi(x)}{\pi(x^*)}J^{b^c}(x^*,\epsilon)\right\}\notag\\
&=\min\left\{\pi(x^*)P(b^c)\varrho(\epsilon)J^{b}(x,\epsilon),\pi(x)P(b^c)\varrho(\epsilon)
J^{b}(x,\epsilon)\times J^{b^c}(x^*,\epsilon)\right\},\notag\\
&=\min\left\{\pi(x^*)P(b^c)\varrho(\epsilon)J^{b}(x,\epsilon),\pi(x)P(b^c)\varrho(\epsilon)
J^{b}(x,\epsilon)\times J^{b^c}(x^*,\epsilon)\right\}.\notag\\
\label{eq:detailed_balance2}
\end{align}
It follows from (\ref{eq:transformation1}) and (\ref{eq:transformation2}) that
$\mathcal T^{b^c}(\mathcal T^b(x,\epsilon))=x$, so that
\begin{align}
&J^{b}(x,\epsilon)\times J^{b^c}(x^*,\epsilon)\notag\\
&=\left|\frac{\partial(\mathcal T^{b}(x,\epsilon),\epsilon)}{\partial (x,\epsilon)}\right|
\times\left|\frac{\partial(\mathcal T^{b^c}(\mathcal T^b(x,\epsilon)),\epsilon)}
{\partial (\mathcal T^b(x,\epsilon),\epsilon)}\right|\notag\\
&=\left|\frac{\partial(\mathcal T^{b}(x,\epsilon),\epsilon)}{\partial (x,\epsilon)}\right|
\times \left|\frac{\partial (x,\epsilon)}{\partial(\mathcal T^{b}(x,\epsilon),\epsilon)}\right|\notag\\
&=1.
\label{eq:jacobians}
\end{align}
Substituting (\ref{eq:jacobians}) in (\ref{eq:detailed_balance2}) we obtain
\begin{equation}
\pi(x^*)K(x^*\rightarrow x)=
\min\left\{\pi(x^*)P(b^c)g(\epsilon)J^{b}(x,\epsilon),\pi(x)P(b^c)g(\epsilon)\right\}.
\label{eq:detailed_balance3}
\end{equation}
If, for simplicity of illustration we assume $p_i=q_i$ for $i=1,\ldots,d$, it follows that $P(b)=P(b^c)$, so that
(\ref{eq:detailed_balance3}) is equal to (\ref{eq:detailed_balance1}), proving detailed balance.
Detailed balance of course holds for general probabilities $(p_1,\ldots,p_d)$
and $(q_1,\ldots,q_d)$; see the supplement of \ctn{Dutta11}.

\subsection{Discussion on the independence of the acceptance probability of the proposal density $\varrho$}
\label{subsec:proposal_independence}
An important feature of the TMCMC acceptance probability distinguishing it from the acceptance
probability of the Metropolis-Hastings algorithms is its independence of the proposal density $\varrho$,
irrespective of whether or not it is symmetric, and for all valid transformations
$T^b$. The reason for this is implicit in the above detailed balance arguments -- the same $\epsilon$ is generated
from the proposal density $\varrho$ while moving forward from $x$ to $x^*$ as well as while returning
to $x$ from $x^*$. This is exactly the reason why $\varrho(\epsilon)$ features in both (\ref{eq:detailed_balance1})
and (\ref{eq:detailed_balance2}). Consequently, detailed balance is satisfied only if the acceptance ratio
of TMCMC is independent of $\varrho$.

\subsection{Relationship of TMCMC with the MH methodology}
\label{subsec:mh_relation}
Note that, although in Section \ref{sec:overview_tmcmc} we indicated
a single random variable $\epsilon$ to be used in the transformations, it is permissible to use
$k$ random variables $\left\{\epsilon_1,\ldots,\epsilon_k\right\}$, where $k\in\left\{1,\ldots,d\right\}$.
Here it is also important to remark that general TMCMC with
$k=d$ does not reduce to general MH method for the very reason that the acceptance
ratio of TMCMC is always independent of the proposal density, while in MH it is not.
For dimension $d=1$, however, TMCMC boils down to an MH algorithm with a specialized two-component 
mixture proposal density, where the mixture components correspond to the two available move types, forward
and backward. See \ctn{Dutta11} for the complete technical details.

As pointed out by a reviewer, TMCMC can be viewed as a MH within Gibbs methodology,
which updates the variables one at a time using general MH (\ctn{Geyer11} criticizes this terminology
since Gibbs is a special case of MH). 
This can be seen as follows. 
Suppose that we assign positive probabilities to the sets 
$\tilde b_i=(b_1=0,\ldots,b_{i-1}=0,b_i,b_{i+1}=0,\ldots,b_d=0)$, for $i=1,\ldots,d$,
where $b_i\in\{-1,0,1\}$. For TMCMC, at each iteration, we can select one of the sets $\tilde b_{i^*}$ 
by choosing $i^*\in\{1,\ldots,d\}$ at random, and 
generate $b_{i^*}\in\{-1,0,1\}$ 
with probabilities $q_{i^*}$, $1-p_{i^*}-q_{i^*}$ and $p_{i^*}$.  
%
%
The corresponding $x_{i^*}$, corresponding to $b_{i^*}$, is then updated according to the TMCMC mechanism.
This being a single-dimensional TMCMC step, coincides with an MH within Gibbs procedure
with a specialized proposal mechanism.


\subsection{Additive and multiplicative transformations in TMCMC}
\label{subsec:add_mult}
\subsubsection{Additive TMCMC}
\label{subsubsec:additive_tmcmc}
The additive TMCMC, which is of our interest in this work, uses moves of the following type: 
\begin{equation*}
(x_{1}, \ldots, x_{d}) \rightarrow (x_{1}+ b_{1}\epsilon, \ldots, x_{d}+b_{d}\epsilon), 
\end{equation*}
where $\epsilon\sim \varrho(\epsilon)$. Here $\varrho(\cdot)$ is an arbitrary density with support $\mathbb R_+$, the
positive part of the real line.
In other words, we assume that $T^{b_i}(x_i,\epsilon)=x_i+b_i\epsilon$. 
In this work,
we shall assume that $p_i=1/2$ and $q_i=1/2$ for $i=1,\ldots,d$, 
so that the probabilities of choosing $b_i=0$ and $b^c_i=0$ are zero. Indeed, as proved in \ctn{Dutta11}, this is a completely
valid and efficient choice for additive transformations.
For this work we set
$\varrho(\epsilon)I_{\{\epsilon>0\}}\equiv N(0,\frac{\ell^2}{d})I_{\{\epsilon>0\}}$.

Note that, for each $i$, $b_i\epsilon\sim N(0,\frac{\ell^2}{d})$, but even though $b_i\epsilon$ are pairwise uncorrelated
($E(b_i\epsilon\times b_j\epsilon)=0$ for $i\neq j$), they are not independent since all of them involve the same $\epsilon$. 
Also observe that
$b_i\epsilon+b_j\epsilon=0$ with probability $1/2$ for $i\neq j$, showing that the linear combinations of $b_i\epsilon$ need
not be normal. In other words, the joint distribution of $(b_1\epsilon,\ldots,b_d\epsilon)$ is not normal, even though the marginal
distributions are normal and the components are pairwise uncorrelated. This also shows that $b_i\epsilon$ are not independent,
because independence would imply joint normality of the components.

Thus, a single $\epsilon$ is simulated from a truncated normal distribution, which is then either added to, or subtracted
from each of the $d$ coordinates of $x$ with probability $1/2$. Assuming that the target distribution is proportional to 
$\pi$, the new move $x^*=(x_1+b_1\epsilon,\ldots,x_d+b_d\epsilon)$ is accepted
with probability
\begin{equation}
\alpha=\min\left\{1,\frac{\pi(x^*)}{\pi(x)}\right\}.
\label{eq:acc_tmcmc}
\end{equation}


The RWM algorithm, unlike additive TMCMC, proceeds by simulating $\epsilon_1,\ldots,\epsilon_d$ independently from
$N(0,\frac{\ell^2}{d})$, and then adding $\epsilon_i$ to the coordinate $x_i$, for each $i$. 
The new move is accepted
with probability having the same form as (\ref{eq:acc_tmcmc}). 
As discussed in Section \ref{subsec:mh_relation} 
for TMCMC it is permissible to use
$k$ random variables $\left\{\epsilon_1,\ldots,\epsilon_k\right\}$, where $k\in\left\{1,\ldots,d\right\}$.
Thus, with such a proposal, if $k=d$, additive TMCMC reduces to RWM, showing that the latter is a special case of the former.

Discussion of computational efficiency of TMCMC is already provided in \ctn{Dutta11}. In this article, we supplement  
the discussion by demonstrating with an experiment the substantial 
computational gains of additive TMCMC over RWM, particularly in high dimensions; see Section S-1.

\subsubsection{Multiplicative TMCMC}
\label{subsubsec:multTMCMC}
In contrast with additive TMCMC, multiplicative TMCMC proceeds by generating $\epsilon$ from some relevant   
distribution $\varrho(\epsilon)$ supported on $[-1,1]\backslash\{0\}$ and then either multiplying or dividing
the current states by $\epsilon$. In this case, it is necessary to assign positive probabilities to $b_i=0$ and $b^c_i=0$
to ensure irreducibility. For $i=1,\ldots,d$, the general multiplicative TMCMC algorithm applies the transformation
$x_i\epsilon^{b_i}$ to $x_i$, which necessitates multiplication of the Jacobian $\epsilon^{\sum_{i=1}^db_i}$
to the ratio of the target distribution evaluated at the new and current states, for computation of the acceptance
ratio; see \ctn{Dutta12}, \ctn{Dutta11} and \ctn{Dey15}. 

There is an alternative version of multiplicative TMCMC that proceeds by generating $\epsilon$ from $\varrho(\epsilon)$ 
that is supported on $(0,1]$, then making the transformation $x_i\mapsto c_ix_i\epsilon^{b_i}$, where
$c_i=1$ with probability $r_i$ and $-1$ with probability $1-r_i$, where $0<r_i<1$. Here, it is again permissible
to set the probabilities of $b_i=0$ and $b^c_i=0$ to zero. The Jacobian remains the same as in the previous 
multiplicative version. 

Note that the multiplicative proposal allows relatively large jumps, while  
the additive moves are local in nature. Hence, judicious combination of the two proposals 
may allow local moves as well as large jumps, which would result in a more efficient
algorithm compared to the individual proposals. Indeed,
\ctn{Dey15} demonstrate that a mixture of the additive and multiplicative transformations outperforms
individual additive TMCMC and multiplicative TMCMC.

\subsubsection{Additive-multiplicative TMCMC}
\label{subsubsec:addmultTMCMC}
Apart from these TMCMC mechanisms where all the coordinates are given either additive transformation or multiplicative
transformation, it is possible to give additive transformation to some coordinates and multiplicative to the rest;
this TMCMC mechanism has been referred to as additive-multiplicative TMCMC in \ctn{Dutta11} and \ctn{Dey15}.

\subsection{A manifold interpretation of the TMCMC proposal mechanism}
\label{subsec:manifold_interpretation}

A reviewer observed that for a $d$-dimensional target distribution, the TMCMC proposal can be viewed as generating
points in a manifold $\mathbb M$ such that $\mathbb M\subseteq\mathbb R^d$. For instance, 
for additive TMCMC this manifold $\mathbb M$ is formed by hyperplanes within $\mathbb R^d$. 
\begin{figure}
\centering
\includegraphics[width=8cm,height=8cm]{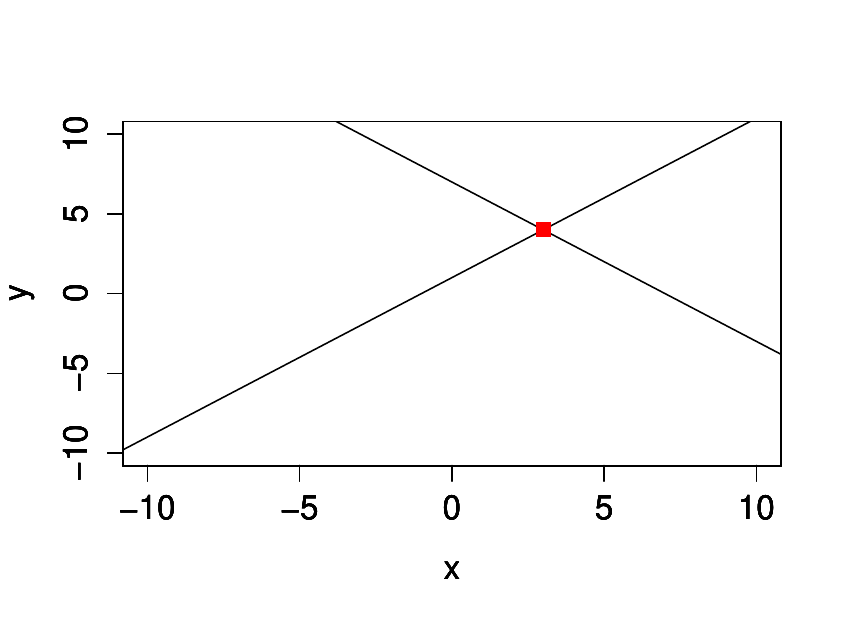}
\caption{Relevant manifold representing the one-step proposal mechanism for additive TMCMC for $d=2$ 
given the current state, denoted by the red patch.
}
\label{fig:addTMCMC}
\end{figure}
When $d=2$, Figure \ref{fig:addTMCMC}  depicts
the manifold as intersecting straight lines with the current state $(x_1,x_2)$ being the point of intersection.
For the two versions of multiplicative TMCMC, the manifolds are shown in Figures \ref{fig:multTMCMC1} and
\ref{fig:multTMCMC2} respectively.
\begin{figure}
\centering
\includegraphics[width=8cm,height=8cm]{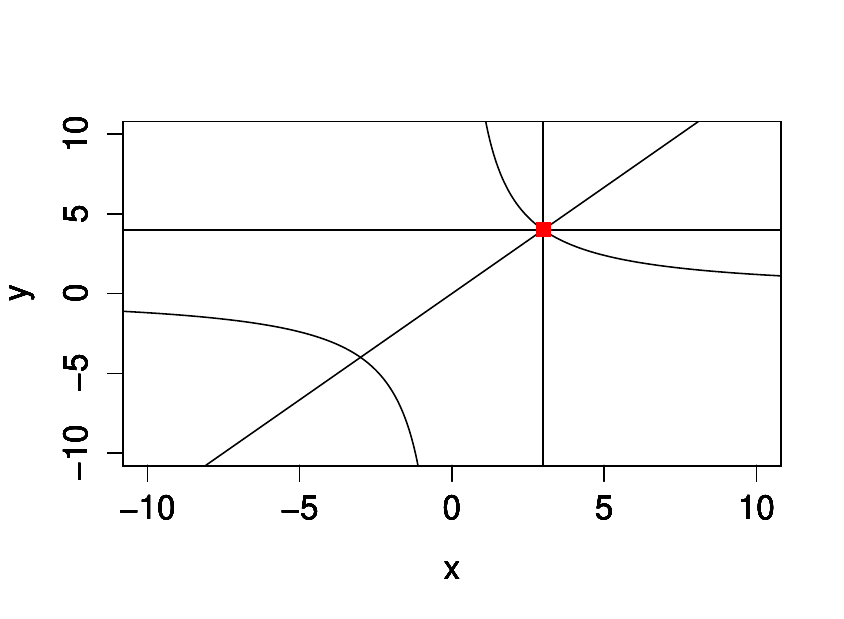}
\caption{Relevant manifold representing the one-step proposal mechanism for the first version of
multiplicative TMCMC for $d=2$ 
given the current state, denoted by the red patch.
}
\label{fig:multTMCMC1}
\end{figure}
\begin{figure}
\centering
\includegraphics[width=8cm,height=8cm]{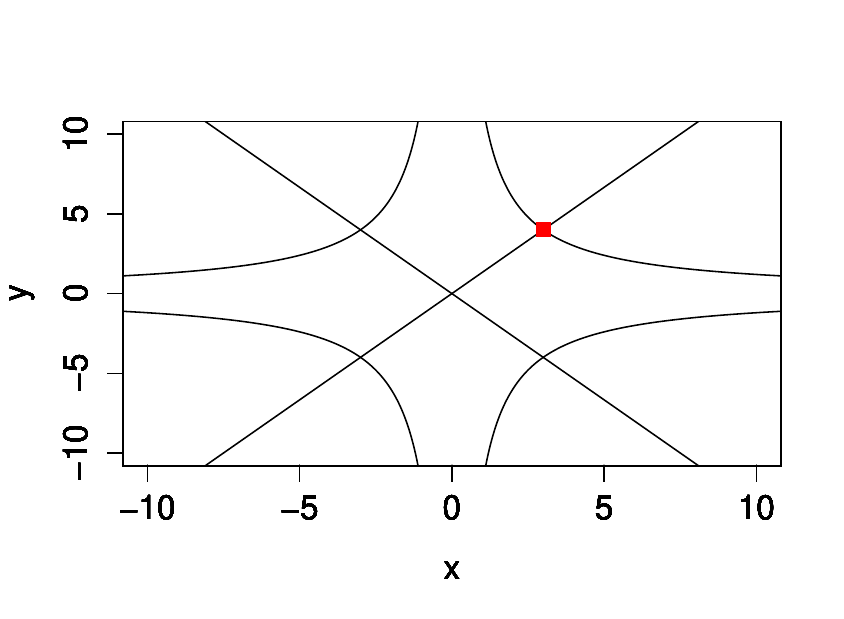}
\caption{Relevant manifold representing the one-step proposal mechanism for the second version of
multiplicative TMCMC for $d=2$ 
given the current state, denoted by the red patch.
}
\label{fig:multTMCMC2}
\end{figure}

The TMCMC proposal implicitly describes the manifold parametrically, where $\epsilon\in\mathbb R^k$ plays the role
of the parameter. If $k=1$, the proposed values must belong to appropriate curves embedded in $\mathbb R^d$.
It is important to clarify, as the reviewer also observed, that $\varrho(\epsilon)$ does not completely define
the TMCMC proposal; the complete TMCMC proposal consists of first generating $\epsilon$ from $\varrho$ and then 
transforming the current state deterministically with the help of $\epsilon$ in a way that ensures the transformed
state falls within the appropriate manifold.

Since, for $d>1$, $\mathbb M\subset\mathbb R^d$ even for $k=d$, the above manifold viewpoint 
also succeeds in explaining why, 
even for $k=d$, the TMCMC proposal does not reduce, in general, to the standard MH methodology. This perspective
also clearly shows the singularity of the proposal, which also explains the association of the Jacobian with the acceptance
probability of TMCMC. Note that for $d=1$ (so that $k=1$), however, we must have $\mathbb M=\mathbb R$.

\subsection{Contrast of TMCMC with deterministic transformation based generalized Gibbs/MH approaches}
\label{subsec:contrast_generalized_mcmc}
TMCMC uses deterministic transformations to update the variables; however, deterministic transformations
have also been considered by \ctn{Liu99}, \ctn{Liu00}, \ctn{Kou05} in an attempt to improve mixing behaviour 
of the underlying usual Gibbs or MH algorithm. In a nutshell, these authors apply suitable deterministic 
transformations, usually additive, as $(x_1,\ldots,x_d)\rightarrow (x_1+\epsilon,\ldots,x_d+\epsilon)$,
or multiplicative, $(x_1,\ldots,x_d)\rightarrow (\eta x_1,\ldots,\eta x_d)$, to the samples generated at each iteration of 
Gibbs or MH in a way that the stationarity
of the target distribution is preserved, while mixing properties of chain after the transformation
may perhaps be improved. To maintain stationarity, $\epsilon$ and $\eta$ must be simulated from some appropriate
distribution which is often impossible to generate from (see \ctn{Liu99}). To alleviate the problem
\ctn{Liu00} (see also \ctn{Kou05}) suggest MH methods.

As pointed out in the supplement of \ctn{Dutta11}, these generalized Gibbs/MH methods are mere attempts towards 
improving mixing of the underlying MCMC. 
These are not stand-alone methodologies which can converge to the target by themselves, unlike TMCMC.
In fact, as shown in the supplement of \ctn{Dutta11} the aforementioned additive and multiplicative transformations 
of generalized Gibbs/MH approaches are themselves not even irreducible, and are hence generally non-convergent.
See the supplement of \ctn{Dutta11} for detailed discussions regarding various issues pertaining to these approaches.

It is important to point out that the aforementioned generalized Gibbs/MH approaches are not the first to consider
deterministic transformations, the earlier methods being the hit-and-run algorithm (see, for example, 
\ctn{Berbee87}, \ctn{Belisle93}, \ctn{Romeijn94a}, \ctn{Smith96}) 
and the adaptive direction sampling algorithm (\ctn{Gilks94a}, \ctn{Gilks94b}). The hit-and-run algorithm proceeds by 
generating random directions from an available set of directions and then considers addition
of a scalar times the sampled direction to the current state as the updated state, where the scalar is drawn from an appropriate,
but non-trivial distribution associated with the target density. The adaptive direction sampler generalizes the
hit-and-run algorithm by selecting the directions and the current state adaptively, based on all the previous iterations.
Since sampling the scalar directly is usually difficult or impossible, this step can be replaced with an
appropriate MH step as in the case of generalized Gibbs/MH methods discussed above (see, for example, \ctn{Romeijn94a}).

The adaptive version of additive TMCMC (\ctn{Dey13}), where the scales are chosen adaptively, comes close
to the Metropolized versions of hit-and-run and adaptive direction sampling methods. The differences are
that, in the latter algorithms, the directions (and also the current state in adaptive direction method) 
are chosen randomly and require a Jacobian of transformations to be evaluated at both the numerator
and denominator of the acceptance ratio apart from the proposal distribution (if not symmetric), 
while in adaptive additive TMCMC the directions (and the current state) are chosen deterministically, 
and the acceptance ratio depends neither on any Jacobian, nor on proposal densities.

\subsection{Contrast of the TMCMC idea with reversible jump Markov chain Monte Carlo (RJMCMC)}
\label{subsec:contrast_rjmcmc}
Interestingly, although both TMCMC and RJMCMC are based on deterministic transformations, 
the philosophy of TMCMC is in sharp contrast with that of RJMCMC. First, TMCMC is designed
for simulation from fixed-dimensional distributions while RJMCMC is meant for generating from
variable-dimensional distributions. Second, the deterministic transformations in RJMCMC are not required for move-types 
that are not meant for changing dimensions, while in TMCMC none of the moves change dimensions, yet all of them
are based on deterministic transformations. Third, the acceptance rates of dimension changing moves of RJMCMC depend upon
the proposal density because for moving from lower to higher dimension simulation of
random variables is required but for returning to lower dimension from higher dimension no simulation is necesary.
For instance, for moving from $x_1\in\mathbb R$ to $(x^*_1,x^*_2)\in\mathbb R^2$, 
one may simulate $u\sim\varphi$, where $\varphi$ is some density supported on $\mathbb R$, and then make the
transformation $x^*_1=x_1-u$ and $x^*_2=x_1+u$. However, for returning from $(x^*_1,x^*_2)$ to $x_1$,
one simply sets $x_1=\frac{x^*_1+x^*_2}{2}$, without simulating any random variable. 
On the other hand, as discussed in Section \ref{subsec:proposal_independence}, since $\epsilon$ is simulated
from $\varrho$ while moving forward and also for moving back, the detailed balance condition dictates that
the acceptance ratio of TMCMC must always be independent of $\varrho$.

\subsection{Discussion of extension of TMCMC to variable dimensional problems} 
\label{subsec:ttmcmc}
\ctn{Das14} extend TMCMC to accommodate variable-dimensional problems; they refer to the new variable
dimensional methodology as Transdimensional Transformation based Markov chain Monte Carlo (TTMCMC).
TTMCMC is designed to update the entire set of parameters, both fixed and variable dimensional, 
as well as the number of parameters, in a single block using simple deterministic
transformations of some low-dimensional (often one-dimensional) random variable drawn from some fixed, but 
arbitrary distribution defined on some relevant support. Again, the acceptance probability is independent
of the proposal density. The advantages of TMCMC over Metropolis-Hastings algorithms are clearly
carried over to the advantages of TTMCMC over RJMCMC. In fact, since it is well-known that efficient implementation of
RJMCMC is generally infeasible, TTMCMC offers huge advantages in this regard; see \ctn{Das14} for details.

\section{Optimal scaling of additive TMCMC when the target density is a product based on $i.i.d.$ random variables}
\label{sec:iid_prod}

It must be emphasized that the proposal density for $\epsilon$ in TMCMC can be any distribution on 
the positive support. Similarly, the RWM algorithm also does not require the proposal to be normal. 
However, the optimal scaling 
results for RWM inherently assume normality, and for the sake of comparison, we have also restricted 
our focus on $\epsilon \sim N(0,\frac{\ell^2}{d})I_{\{\epsilon>0\}}$ in the subsequent sections. 

In this paper, we are primarily interested in choosing the parameters of the process judiciously so as 
to enhance the performance of the chain. Our method as stated above involves only a single parameter 
$\--$ the proposal variance, or to be more precise, the scaling factor $\ell$. 
Details on the need for optimal scaling of additive TMCMC with regard to the scaling factor $\ell$
are provided in Section S-2 of the supplement.

In this section we consider the problem of optimal scaling in the simplest case where the target 
density $\pi$ is a product of $i.i.d.$ marginals, given by
\begin{equation}
\pi(x) = \prod_{i=1}^{d}{f(x_{i})}. 
\end{equation}

Assuming that the TMCMC chain is started at stationarity,
we shall show that for each component of $X$, 
the corresponding one-dimensional process converges to a diffusion 
process which is analytically tractable and whose diffusion and drift speeds may be numerically evaluated. 
It is important to remark that it is possible to relax the assumption of stationarity; 
see \ctn{Jour12} in the context of RWM, although we do not pursue this in our current work. 

Let $X^d_t=(X_{t,1},\ldots,X_{t,d})$.
We define ${U_{t}}^{d} = {X_{[dt],1}}$ ($[\cdot]$ denotes the integer part), 
the sped up first component of the actual additive TMCMC-induced Markov chain. 
Note that this process proposes a jump every $\frac{1}{d}$ time units. As $d \rightarrow \infty$, 
that is, as the dimension grows to $\infty$, the process essentially becomes a continuous time diffusion process. 


Before proceeding first let us introduce the notion of Skorohod topology (\cite{Skorohod}). It is a topology generated by a class of functions 
from $[0,1] \rightarrow \mathbb{R}$ for which the right-hand side and the left-hand side limits are well defined at each point (even though they 
may not be the same). 
It is an important tool for formulating Poisson process, Levy process and other stochastic point processes. As considered in
\ctn{Roberts97a} here we also consider the metric 
separable topology 
on the above class of functions as defined in \cite{Skorohod}. In other words, whenever we mention convergence of 
discrete time stochastic processes to diffusion process in this paper, we mean convergence with respect to this 
topology.


In what follows, we assume the following:
\begin{align}
&E_f\left(\frac{f'(X)}{f(X)}\right)^4 <\infty,
\label{eq:assump1}\\
&E_f\left(\frac{f''(X)}{f(X)}\right)^4  <\infty,
\label{eq:assump2}\\
&E_f\left(\frac{f'''(X)}{f(X)}\right)^4  <\infty.
\label{eq:assump3}\\
&E_f\left|\frac{f''''(X)}{f(X)}\right|  <\infty.
\label{eq:assump4}
\end{align}
These assumptions can also be somewhat relaxed, depending upon the order of the Taylor's series
expansions used in the proofs.
Following \ctn{Roberts97a} let us denote weak convergence of processes in the Skorohod topology by $\Rightarrow$. 

We next present our formal result in the $i.i.d.$ situation, the proof of which is presented
in Section S-6.1. Our proof differs from the previous
approaches associated with RWM particularly because as already shown in Section \ref{sec:overview_tmcmc}, 
in additive TMCMC, the terms $b_i\epsilon$ are not
jointly normally distributed unlike the RWM-based approaches. 
Thus, unlike the RWM-based approaches, in our case obtaining appropriate normal approximation to
relevant quantities are not assured. To handle the difficulty, we
had to apply Lyapunov's central limit theorem on sums associated with the discrete random variables 
$\{b_i;i=2,\ldots,d\}$, {\it conditional
on} $\epsilon$ (and $b_1$).
This required us to verify Lyapunov's condition (see, for example, \ctn{Koralov07}) before
applying the central limit theorem. We then integrated
over $\epsilon$ and $b_1$. These issues make our proof substantially different from the previous 
approaches associated
with RWM. It is important to remark that, not only in this $i.i.d.$ scenario, but in all the set-ups that we
consider in this paper, application of Lyapunov's central limit theorem, conditionally on $\epsilon$ 
(and often $b_1$), is crucial, before finally integrating over the conditioned variables to obtain our
results.

\begin{theorem}
\label{theorem:theorem1}
Assume that $f$ is positive with at least three continuous derivatives and that the fourth derivative exists almost everywhere. 
Also assume that $(\log f)'$ is Lipschitz continuous, 
and that (\ref{eq:assump1}) -- (\ref{eq:assump4}) hold. 
Let $X^d_0\sim\pi$, that is, the $d$-dimensional additive TMCMC chain is started at stationarity, and
let the transition be given by $(x_1,\ldots,x_d)\rightarrow (x_1+b_1\epsilon,\ldots,b_d\epsilon)$, where
for $i=1,\ldots,d$, $b_i=\pm 1$ with equal probability and 
$\epsilon\equiv\frac{\ell}{\sqrt{d}}\epsilon^*$, where $\epsilon^*\sim N(0,1)I_{\{\epsilon^*>0\}}$.
We then have
\[
\{U^d_t;~t\geq 0\}\Rightarrow \{U_t;~t\geq 0\},
\]
where $U_0\sim f$ and $\{U_t;~t\geq 0\}$ satisfies the Langevin stochastic differential equation (SDE)
\begin{equation}
dU_t=g(\ell)^{1/2}dB_t+\frac{1}{2}g(\ell)\left(\log f(U_t)\right)'dt,
\label{eq:sde1}
\end{equation}
with $B_t$ denoting standard Brownian motion at time $t$,
\begin{equation}
g(\ell)=4\ell^2\int_0^{\infty} u^2\Phi\left (- \frac{u\ell\sqrt{\mathbb{I}}}{2}\right)\phi(u)du;
\label{eq:diff_speed_iid}
\end{equation}
$\Phi(\cdot)$ and $\phi(\cdot)$ being the standard normal cumulative distribution function (cdf) and density, respectively, and
\begin{equation}
\mathbb I=E_f\left(\frac{f'(X)}{f(X)}\right)^2.
\label{eq:information}
\end{equation}
\end{theorem}
%
%
In connection with our diffusion equation (see Equation (18) 
in connection with the proof of
Theorem \ref{theorem:theorem1} in Section S-6.1), we note that our SDE is also 
Langevin like the usual RWM approach. But, we have a different \emph{speed} and it is interesting to compare how the 
two \emph{speed} functions of our method is related to that of RWM and also, how it alters the optimal 
expected acceptance rate of the process. In what follows, we use the terms \emph{speed} and 
\emph{diffusion speed} of the process, given by  $g(\ell)$ as in (\ref{eq:diff_speed_iid}) interchangeably.

\begin{corollary}
\label{cor:cor1}
The diffusion speed $g(\ell)$ is maximized by
\begin{equation}
\ell_{opt}=\frac{2.426}{\sqrt{\mathbb I}},
\label{eq:l1}
\end{equation}
and the optimal acceptance rate is given by
\begin{equation}
\alpha_{opt}=4\int_0^{\infty} \Phi\left (- \frac{u\ell_{opt}\sqrt{\mathbb{I}}}{2}\right)\phi(u)du
=0.439\ \ \mbox{(up to three decimal places)}.
\label{eq:tmcmc_opt_acc}
\end{equation}
\end{corollary}

\subsection{TMCMC within Gibbs for $i.i.d.$ product densities}
\label{subsec:tmcmc_within_gibbs_iid}

The main notion of Gibbs sampling is to update one or multiple components of a 
multidimensional random vector conditional on the remaining components. 
In TMCMC within Gibbs, we update only a fixed proportion $c_{d}$ of the $d$ coordinates, 
where $c_{d}$ is a function of $d$ and we assume that as $d \rightarrow \infty$, 
then $c_{d} \rightarrow c$, for some $0<c\leq 1$. In order to explain the transitions in this process analytically, 
we define an indicator function $\chi_{i}$ for $ i= 1,\ldots,d$. For fixed $d$, 
\begin{eqnarray}
\mathbb{\chi}_{i} &=& 1 \hspace{0.5 cm} \mbox{if transition takes place in the} \ \  i^{th} \ \ \mbox{coordinate} 
\nonumber \\
&=& 0 \hspace{0.5 cm} \mbox{if no transition takes place in the}\ \  i^{th} \ \ \mbox{coordinate}. 
\label{eq:chi_definition}
\end{eqnarray}
Our assumptions imply that
\begin{equation}
P (\mathbb{\chi}_{i}= 1) = c_{d}; \ \ i =1,\ldots,d.
\label{eq:chi_distribution}
\end{equation}
Then a feasible transition with respect to additive TMCMC can be analytically expressed as 
\begin{equation}
(x_{1},\ldots,x_{d}) \rightarrow (x_{1}+\mathbb{\chi}_{1}b_{1}\epsilon,\ldots,x_{d}+ \mathbb{\chi}_{d}b_{d}\epsilon),
\label{eq:transition}
\end{equation}
where $\epsilon\equiv\frac{\ell}{\sqrt{d}}\epsilon^*$, where $\epsilon^*\sim N(0,1)I_{\{\epsilon^*>0\}}$.
We then have the following theorem, the proof of which is presented in Section S-6.2 of the supplement.

\begin{theorem}
\label{theorem:theorem2}
Assume that $f$ is positive with at least three continuous derivatives and that the fourth derivative exists almost everywhere. 
Also assume that $(\log f)'$ is Lipschitz continuous, 
and that (\ref{eq:assump1}) -- (\ref{eq:assump4}) hold. 
Suppose also that the transition is given by (\ref{eq:transition}) and that 
as $d \rightarrow \infty$, $c_{d} \rightarrow c$, for some $0<c\leq 1$.
Let $X^d_0\sim\pi$, that is, the $d$-dimensional additive TMCMC chain is started at stationarity. 
We then have
\[
\{U^d_t;~t\geq 0\}\Rightarrow \{U_t;~t\geq 0\},
\]
where $U_0\sim f$ and $\{U_t;~t\geq 0\}$ satisfies the Langevin SDE
\begin{equation}
dU_t=g_c(\ell)^{1/2}dB_t+\frac{1}{2}g_c(\ell)\left(\log f(U_t)\right)'dt,
\label{eq:sde2}
\end{equation}
where
\begin{equation}
g_c(\ell)=4c\ell^2\int_0^{\infty} u^2\Phi\left (- \frac{u\ell\sqrt{c\mathbb{I}}}{2}\right)\phi(u)du,
\label{eq:diff_speed_iid2}
\end{equation}
and
$\mathbb I$ is given by (\ref{eq:information}).
\end{theorem}

\begin{corollary}
\label{cor:cor2}
The diffusion speed $g_c(\ell)$ is maximized by
\begin{equation}
\ell_{opt}=\frac{2.426}{\sqrt{c\mathbb I}},
\label{eq:l2}
\end{equation}
and the optimal acceptance rate is given by
\begin{equation}
\alpha_{opt}=4\int_0^{\infty} \Phi\left (- \frac{u\ell_{opt}\sqrt{c\mathbb{I}}}{2}\right)\phi(u)du
=0.439\ \ \mbox{(up to three decimal places)}.
\label{eq:acc_tmcmc_gibbs}
\end{equation}
\end{corollary}

\section{Diffusion approximation for independent but non-identical random variables}
\label{sec:non_iid}

So far we have considered only those target densities $\pi$ which correspond to $i.i.d.$ components of $x$. 
Now, we extend our investigation to those target densities that are associated with independent 
but not identically distributed random variables.  
That is, we now consider

\begin{equation}
\pi(x) = \prod_{i=1}^{d}{f_{i}(x_{i})}.
\end{equation}

We concentrate on a particular form of the target density involving some scaling constant parameters, 
as considered in \cite{Bedard2008}, \cite{BedRose}.
\begin{equation}
\pi(x) = \prod_{j=1}^{d} \theta_{j}(d) f(\theta_{j}(d)x_{j}). 
\label{eq:pi_non_iid}
\end{equation}
As before, we assume that $f$ is twice continuously differentiable
with existence of third derivative almost everywhere, and that $(\log f)'$ is Lipschitz continuous.
We define $\Theta(d) = \left \{\theta_{1}(d), \theta_{2}(d), \ldots, \theta_{d}(d) 
\right\}$ and we shall focus on the case where $d \rightarrow \infty$. Some of the scaling terms are 
allowed to appear multiple times. We assume that the first $k$ terms of the parameter vector may or may not 
be identical, but the remaining $d-k$ terms can be split into $m$ subgroups of independent scaling terms.
In other words, 
\begin{eqnarray}
 \Theta(d) &=& \left (\vphantom{\underbrace{\theta_{k+2}(d), \ldots, \theta_{k+2}(d)}_{r(2,d)-1}}\theta_{1}(d), 
 \theta_{2}(d), \ldots, \theta_{k}(d), \theta_{k+1}(d),\ldots, \theta_{k+m}(d),\nonumber \right.\\
&& \qquad \left.  \underbrace{\theta_{k+1}(d), \ldots, \theta_{k+1}(d)}_{r(1,d)-1},  
\underbrace{\theta_{k+2}(d), \ldots, \theta_{k+2}(d)}_{r(2,d)-1}, \ldots, 
\underbrace{ \theta_{k+m}(d), \ldots, \theta_{k+m}(d)}_{r(m,d)-1}\right),\label{eq:bs1}
\end{eqnarray}
where $ r(1,d), r(2,d), \ldots, r(m,d)$ are the number of occurrences of the parameters in each of the 
$m$ distinct classes. We assume that for any $i$, 
\begin{equation}
\underset{d\rightarrow \infty}{\lim} r(i,d) = \infty.
\label{eq:bs2}
\end{equation}
Also, we assume a particular form of each scaling parameter $ \theta_{i}(d)$:
\begin{equation}
\frac{1}{\{\theta_{i}(d)\}^{2}} = \frac{K_{i}}{d^{\lambda_{i}}};\ \ i=1,\ldots,k, \ \ \ \ 
\mbox{and}\ \ \ \ \frac{1}{\{\theta_{i}(d)\}^{2}} = \frac{K_{i}}{d^{\gamma_{i}}};\ \  i=k+1,\ldots,k+m. 
\label{eq:bs3}
\end{equation}

Assume that ${\theta_{i}}^{-2}(d)$ are so arranged that $\gamma_{i}$ are in a decreasing sequence 
for $i=k+1,\ldots,k+m$ and also let $\lambda_{i}$ form a decreasing sequence from $i=1,\ldots,k$. 
According to \cite{Bedard2007}, the optimal form of the scaling variance $\sigma^{2}(d)$ 
should be of the form $ \sigma^{2}(d) =\frac{\ell^{2}}{d^{\alpha}}$, where $\ell^{2}$ is some 
constant and $\alpha$ satisfies

\begin{equation}
\underset{d\rightarrow \infty}{\lim}\frac{d^{\lambda_{1}}}{d^{\alpha}} < \infty, \ \ \ \ \mbox{and} 
\ \ \ \  \underset{d\rightarrow \infty}{\lim}\frac{d^{\gamma_{i}}r(i,d)}{d^{\alpha}} < \infty; 
\ \  i= 1,\ldots,m.
\label{eq:Bedardassumptions}
\end{equation}
Here, let $U^d_{t}$ be the process at time $t$ sped up by a factor of $d^{\alpha}$. That is,
$U^d_{t} = \left ( X_{1}([d^{\alpha}t]), \ldots, X_{d}([d^{\alpha}t]) \right)$. 
We then have the following theorem, the proof of which is provided in Section S-6.3 of the supplement.

\begin{theorem}
\label{theorem:theorem4}
Assume that the target distribution is of the form (\ref{eq:pi_non_iid}), where 
$f$ is positive with at least three continuous derivatives and that the fourth derivative exists almost everywhere. 
Also assume that $(\log f)'$ is Lipschitz continuous, 
and that (\ref{eq:assump1}) -- (\ref{eq:assump4}), (\ref{eq:bs1}), (\ref{eq:bs2}),
(\ref{eq:bs3}) and (\ref{eq:Bedardassumptions}) hold. 
Let $X^d_0\sim\pi$, that is, the $d$-dimensional additive TMCMC chain is started at stationarity. 
Let the transition be given by $(x_1,\ldots,x_d)\rightarrow (x_1+b_1\epsilon,\ldots,x_d+b_d\epsilon)$,
where, for $i=1,\ldots,d$, $b_i=\pm 1$ with equal probability and $\epsilon\equiv\frac{\ell}{d^{\frac{\alpha}{2}}}\epsilon^*$,
with $\epsilon^*\sim N(0,1)I_{\{\epsilon^*>0\}}$.
We then have
\[
\{U^d_t;~t\geq 0\}\Rightarrow \{U_t;~t\geq 0\},
\]
where $U_0\sim f$ and $\{U_t;~t\geq 0\}$ satisfies the Langevin SDE
\begin{equation}
dU_t=g_{\xi}(\ell)^{1/2}dB_t+\frac{1}{2}g_{\xi}(\ell)\left(\log f(U_t)\right)'dt,
\label{eq:sde4}
\end{equation}
where
\begin{equation}
g_{\xi}(\ell)=4\ell^2\int_0^{\infty} u^2\Phi\left (- \frac{u\ell\xi\sqrt{\mathbb{I}}}{2}\right)\phi(u)du.
\label{eq:diff_speed_bedard}
\end{equation}
\end{theorem}


\begin{corollary}
\label{cor:cor3}
The diffusion speed $g_c(\ell)$ is maximized by
\begin{equation}
\ell_{opt}=\frac{2.426}{\xi\sqrt{\mathbb I}},
\label{eq:l3}
\end{equation}
and the optimal acceptance rate is given by
\begin{equation}
\alpha_{opt}=4\int_0^{\infty} \Phi\left (- \frac{u\ell_{opt}\xi\sqrt{\mathbb{I}}}{2}\right)\phi(u)du
=0.439\ \ \mbox{(up to three decimal places)}.
\label{eq:opt_acc_non_iid}
\end{equation}
\end{corollary}

\subsection{TMCMC within Gibbs for independent but non-identical random variables}
\label{subsec:tmcmc_within_gibbs_non_iid}

As in Section \ref{subsec:tmcmc_within_gibbs_iid}, here also we define
transitions of the form (\ref{eq:transition}), where $\chi_i$, having the same definitions
as (\ref{eq:chi_definition}) and (\ref{eq:chi_distribution}), indicates whether or not the $i$-th
coordinate $x_i$ will be updated.

The rest of the proof is a simple modification of the proof for independent but non-identical
random variables provided in Section S-6.3. There we must replace 
\begin{equation*}
\frac{1}{r(i,d)}\sum_{j=1}^{r(i,d)}\left(\frac{f'(u_j)}{f(u_j)}\right)^2 
\rightarrow E\left [\left\{\frac{f^{'}(U)}{f(U)}\right \}^{2} \right]=\mathbb{I}.
\label{eq:convergence_information}
\end{equation*}
with
\begin{equation}
\frac{c_d}{c_dr(i,d)}\sum_{j=1}^{c_dr(i,d)}\left(\frac{f'(u_j)}{f(u_j)}\right)^2 
\rightarrow cE\left [\left\{\frac{f^{'}(U)}{f(U)}\right \}^{2} \right]=c\mathbb{I}.
\label{eq:convergence_information2}
\end{equation}
With the above modification the diffusion speed can be calculated as 
\begin{equation}
g_{c,\xi}(\ell)= 4c\ell^{2}\int_0^{\infty} \left \{{u}^{2}\Phi 
\left (- \frac{u\ell\xi\sqrt{c\mathbb{I}}}{2}\right) \right \}\phi(u)du.
\label{eq:diff_speed_non_iid2}
\end{equation}

Formally, we have the following theorem:
\begin{theorem}
\label{theorem:theorem5}
Assume that the target distribution $\pi$ is of the form (\ref{eq:pi_non_iid}), where
$f$ is positive with at least three continuous derivatives and that the fourth derivative exists almost everywhere. 
Also assume that $(\log f)'$ is Lipschitz continuous, 
and that (\ref{eq:assump1}) -- (\ref{eq:assump4}), (\ref{eq:bs1}), (\ref{eq:bs2}),
(\ref{eq:bs3}) and (\ref{eq:Bedardassumptions}) hold. 
Let $X^d_0\sim\pi$, that is, the $d$-dimensional additive TMCMC chain is started at stationarity. 
Let the transition be $(x_1,\ldots,x_d)\rightarrow (x_1+\chi_1b_1\epsilon,\ldots,\chi_db_d\epsilon)$, where
for $i=1,\ldots,d$, $P(\chi_i=1)=c_d$, $b_i=\pm 1$ with equal probability, and $\epsilon\equiv\frac{\ell}{d^{\frac{\alpha}{2}}}\epsilon^*$,
with $\epsilon^*\sim N(0,1)I_{\{\epsilon^*>0\}}$.
We then have
\[
\{U^d_t;~t\geq 0\}\Rightarrow \{U_t;~t\geq 0\},
\]
where $U_0\sim f$ and $\{U_t;~t\geq 0\}$ satisfies the Langevin SDE
\begin{equation}
dU_t=g_{c,\xi}(\ell)^{1/2}dB_t+\frac{1}{2}g_{c,\xi}(\ell)\left(\log f(U_t)\right)'dt,
\label{eq:sde5}
\end{equation}
where $g_{x,\xi}(\ell)$ is given by (\ref{eq:diff_speed_non_iid2}).
\end{theorem}

\begin{corollary}
\label{cor:cor4}
The diffusion speed $g_{c,\xi}(\ell)$ is maximized by
\begin{equation}
\ell_{opt}=\frac{2.426}{\xi\sqrt{c\mathbb I}},
\label{eq:l4}
\end{equation}
and the optimal acceptance rate is given by
\begin{equation}
\alpha_{opt}=4\int_0^{\infty} \Phi\left (- \frac{u\ell_{opt}\xi\sqrt{c\mathbb{I}}}{2}\right)\phi(u)du
=0.439\ \ \mbox{(up to three decimal places)}.
\label{eq:opt_acc_non_iid2}
\end{equation}
\end{corollary}

\section{Diffusion approximation for a more general dependent family of distributions}
\label{sec:generalization}

So far, we assumed that the target density $\pi$ is associated with either $i.i.d.$ or mutually 
independent random variables, with a special structure. Now, we extend our notion to a much wider 
class of distributions where there is a particular form of dependence structure between the 
components of the distribution. In determining these non-product measures, we adopted the framework 
of \ctn{Pillai2011}, \cite{Beskos2009}, \ctn{Beskos2008}, \cite{Bedard2009}. 
For clarity, we first discuss this in the case of finite dimension $d$, and then discuss
the generalization in infinite dimensions. 

Let $x^d\in \mathbb R^d$ denote the first $d$ coordinates of 
$x\in\mathbb R^{\infty}$.
Let us assume that the $d$-dimensional target density $\pi^d$ satisfies
\begin{equation}
\frac{d\pi^d}{d\pi^d_{0}}(x^d) = M_{\Psi^d} \exp ( -\Psi^d(x^d)),
\label{eq:dep1}
\end{equation}
where $\Psi^d$ is measurable with respect to the Borel $\sigma$-field on $\mathbb R^d$, 
$M_{\Psi^d}$ is an appropriate normalizing constant depending upon $\Psi^d$,
and
$\pi^d$ has the density
\begin{equation}
\pi^d_0(x^d) = \prod_{j=1}^d\frac{1}{\lambda_j}\phi\left(\frac{x_j}{\lambda_j}\right)
\label{eq:dep2}
\end{equation}
with respect to the Lebesgue measure. In other words, under $\pi^d_0$,
$x_j\sim N(0,\lambda^2_j);~j=1,2,\ldots,d$.

Then, with respect to Lebesgue measure, $\pi^d$ has the following density:
\begin{equation}
\pi^d(x^d)=M_{\Psi^d} \exp ( -\Psi^d(x^d))\prod_{i=1}^d\frac{1}{\lambda_i}\phi\left(\frac{x_i}{\lambda_i}\right).
\label{eq:dep3}
\end{equation}

The above finite dimensional structure can be represented in terms of projection
onto the first $d$ eigenfunctions of an appropriate covariance operator associated
with a Hilbert space. Indeed, let 
$(\mathbb H,\langle\cdot\rangle,\|\cdot\|)$ denote a real, separable Hilbert space.
Consider a covariance operator $\Sigma:\mathbb H\rightarrow\mathbb H$, which is self-adjoint,
positive, and trace class operator 
on $\mathbb H$ with a complete orthonormal eigenbasis
$\{\lambda^2_j,\phi_j\}_{j=1}^{\infty}$ such that
\begin{equation}
\centering
\Sigma\phi_{j}= \lambda^2_{j}\phi_{j}; \ \ j = 1,2,\ldots.
\label{eq:eigen}
\end{equation}
As in \ctn{Pillai2011} we assume that the eigenvalues are arranged in decreasing order and $\lambda_j>0$.

Now note that any function $x$ in $\mathbb{R}^{\infty}$ can be uniquely represented as 
\begin{equation}
x= \sum_{j=1}^{\infty} { x_{j}\phi_{j}}, \ \ \mbox{where} \ \  x_{j}= \langle x, \phi_{j} \rangle.
\label{eq:fourier_expansion}
\end{equation}
The function $x$ can be identified with its coordinates $\{x_j\}_{j=1}^{\infty}$ which belongs 
to the space of square-summable sequences.
Note that $\Sigma$ is diagonal with respect to the coordinates of this eigenbasis, and 
if $x_j\sim N(0,\lambda^2_j)$; $j=1,2,\ldots$ independently, then by the
Karhunen-Lo\'{e}ve expansion (see, for example, \ctn{Prato92}), $x$ follows the Gaussian measure $\pi_0$, 
which is an infinite dimensional generalization of (\ref{eq:dep2}).
In particular, we assume that $\pi_{0}$ is a 
Gaussian measure with mean 0 and covariance $\Sigma$.

Now, let $\Psi^d(\cdot)=\Psi(P^d\cdot)$, where $P^d$ denotes projection (in $\mathbb H$) onto the first
$d$ eigenfunctions of $\Sigma$, 
and $\Psi$ is a real $\pi_{0}$-measurable function on 
$\mathbb{R}^{\infty}$.  
Then $\pi^d(x^d)$ given by (\ref{eq:dep3}) can be represented as
\begin{equation}
\pi^{d}(x) =  M_{\Psi^{d}}\exp \left ( \vphantom{\frac{1}{2}}-\Psi^{d}(x) -\frac{1}{2} 
\langle x, (\Sigma^{d})^{-1}x \rangle \right ),  
\label{eq:piN}
\end{equation}
where $\Sigma^{d} = P^{d} \Sigma P^{d}$. 
As $d\rightarrow\infty$, (\ref{eq:piN}) approximates the target density $\pi(x)$, where
the Radon Nikodym derivative of the target 
$\pi$ with respect to the Gaussian measure $\pi_{0}$ is given by 
\begin{equation}
\frac{d\pi}{d\pi_{0}}(x) = M_{\Psi} \exp ( -\Psi(x)).
\label{eq:target_pi}
\end{equation}
Hence, for our purpose we shall work with the finite-dimensional approximation (\ref{eq:piN});
as $d\rightarrow\infty$, the appropriate piecewise linear, continuous interpolant (to be defined
subsequently in Section \ref{subsec:formal_statement}) 
that is described by our additive TMCMC algorithm and associated with $\pi^d$ will converge
to the correct diffusion equation associated with the infinite dimensional 
distribution $\pi$ represented by (\ref{eq:target_pi}).

\subsection{Representation of the additive TMCMC algorithm in the dependent set-up}
\label{subsec:tmcmc_algo_dependent}

Under the TMCMC set up, the move at the 
$(k+1)$-th time point can be explicitly stated in terms of the position at $k$-th time point as follows 

\begin{equation}\label{eq:move}
x^{k+1} = { \gamma^{k+1}{y}^{k+1}} + \left ( 1- { \gamma^{k+1}} \right) x^{k}, 
\end{equation}
where 
$$ {\gamma}^{k+1} \sim \mbox{Bernoulli} \left ( \min \left \{1, \frac{\pi^d({y}^{k+1})}{\pi^d(x^{k})}\right \}\right ). $$

We define the move ${y}^{k+1}$ as 
\begin{equation}\label{eq:ynew}
{y}^{k+1} = x^{k} + \sqrt{\frac{2\ell^{2}}{d}}\Sigma^{\frac{1}{2}}\xi^{k+1},
\end{equation}
 where $ \xi^{k+1} = \left ({b_{1}}^{k+1}{\epsilon}^{k+1}, \ldots, {b_{d}}^{k+1}{\epsilon}^{k+1} \right )$ 
 with $b_{i}=\pm 1$ with probability 1/2 each, and $\epsilon\sim N(0,1)I_{\{\epsilon>0\}}$. From (\ref{eq:piN}) it follows that 
 $ \min \left \{1, \frac{\pi^d({y}^{k+1})}{\pi^d(x^{k})}\right \}$ can be written as 
 $ \min \left \{1, e^{\mathbb{Q}(x^k,\xi^{k+1})} \right \}$ where $\mathbb{Q}(x,\xi)$ is given by 
\begin{equation}
\mathbb{Q}(x,\xi)= \frac{1}{2}\left \| \Sigma^{-\frac{1}{2}} \left ( P^{d}x \right) \right \|^{2} -  
\frac{1}{2}\left \| \Sigma^{-\frac{1}{2}} \left ( P^{d}{y} \right) \right \|^{2}+  
\Psi^{d}(x) - \Psi^{d}({y}).
\end{equation}
\newcommand{\y}{{y}^{k+1}}
\newcommand{\x}{x^{k}}
Using (\ref{eq:ynew}), one obtains 
\begin{equation}
\mathbb{Q}(x,\xi) = -\sqrt{\frac{2\ell^{2}}{d}} \langle \eta, \xi \rangle - \frac{\ell^{2}}{d}\|\xi\|^{2}-r(x,\xi),
\end{equation}
where 
\begin{equation}\label{eq:defineeta}
\eta = \Sigma^{-\frac{1}{2}}(P^{d}x) + \Sigma^{\frac{1}{2}}\nabla\Psi^{d}(x),
\end{equation}
and
\begin{equation}
r(x,\xi)= \Psi^{d}(y) - \Psi^{d}(x) - \langle \nabla \Psi^{d}(x), P^{d}y-P^{d}x \rangle.
\end{equation}

We further define

\begin{equation}\label{eq:R}
R(x, \xi) = -\sqrt{\frac{2\ell^{2}}{d}} \sum_{j=1}^{d} \eta_{j}\xi_{j} 
-\frac{\ell^{2}}{d}\sum_{j=1}^{d}{\xi_{j}}^{2},
\end{equation}
and
\begin{equation}
R_{i}(x, \xi) = -\sqrt{\frac{2\ell^{2}}{d}} \sum_{j=1, j\neq i}^{d} \eta_{j}\xi_{j} 
-\frac{\ell^{2}}{d}\sum_{j=1, j \neq i}^{d}{\xi_{j}}^{2}
\end{equation}

Using Lemma 5.5 of \cite{Pillai2011}, for large $d$ one can show that 
\begin{equation}\label{eq:Q}
\mathbb{Q}(x,\xi) \hspace{0.1 cm} = \hspace{0.1 cm} R(x,\xi) -r(x,\xi) 
\hspace{0.1 cm} \approx \hspace{0.1 cm} R_{i}(x,\xi) -\sqrt{\frac{2\ell^{2}}{d}}\eta_{i}\xi_{i}.
\end{equation}

Using (\ref{eq:R}) and (\ref{eq:Q}) it can be seen that $\mathbb{Q}(x,\xi)$ is approximately equal to 
$R(x,\xi)$ as $d$ goes to $\infty$, where $R(x,\xi)$ in our case is given by
\begin{equation}
R(x, \xi) = -\epsilon\sqrt{\frac{2\ell^{2}}{d}} \sum_{j=1}^{d} \eta_{j} b_{j} -\ell^{2}\epsilon^{2}.
\end{equation}
Note that in the case of \ctn{Pillai2011}, conditional on $x$, $R_i(x,\xi)$ was independent
of $\xi_i$, which enabled them to compute 
$E_0\left(\min\left\{1,e^{\mathbb Q(x,\xi)}\right\}\xi_i\right)$
by first computing it over $\xi_i$ and then over $\xi\backslash\xi_i$. However, such independence
does not hold in our case since all the components of $\xi$ involve $\epsilon$.

To obtain  $E_0\left(\min\left\{1,e^{\mathbb Q(x,\xi)}\right\}\xi_i\right)$ in our case,
we need to obtain the asymptotic distribution of $\mathbb Q(x,\xi)$ for large $d$. 
Since our TMCMC based proposal is not $i.i.d.$, we verify Lyapunov's central limit theorem;
see Section S-7.1. 
For obtaining the diffusion approximation in this dependent set-up we need to obtain the expected
drift and the expected diffusion coefficient.
In 
Section S-7.2 we calculate the expected drift
and in Section S-7.3, 
we obtain the expected diffusion coefficient.

\subsection{Formal statement of our main result in the general dependent set-up}
\label{subsec:formal_statement}

Before formally stating our result in the dependent set-up, we need to provide the explicit
form of a continuous interpolant which converges to the solution of the appropriate SDE.

Note that we can construct, following \ctn{Pillai2011}, the following continuous interpolant
\begin{equation}
z^d(t)=(dt-k)x^{k+1}+(k+1-dt)x^k,\hspace{2mm} k\leq dt<k+1.
\label{eq:continuous_interpolant}
\end{equation}
Observe that $z^d(t)$ admits the following representation 
\begin{equation}
z^d(t)=z^0+\int_0^t\vartheta^d({\bar z}^d(s))ds+\sqrt{2g(\ell)}W^d(t),
\label{eq:continuous_interpolant2}
\end{equation}
where $z^0\sim\pi$, $g(\ell)=\ell^2\beta$,
$\vartheta^d(x)=dE_0(x^1-x)$, ${\bar z}^{d}(t) = x^{k};\ \ t \in [t^{k}, t^{k+1} ]$ is a piecewise
constant interpolant of $x^k$,
where 
\begin{align}
&\quad t^k=k\Delta t, \quad \eta^{k,d}=\sqrt{\Delta t}\sum_{j=1}^k\Gamma^{j,d},\label{eq:def1}\\
& W^d(t) = \eta^{[dt],d}+\frac{dt-[dt]}{\sqrt{d}}\Gamma^{[dt]+1,d};\ \ t\in [0,T],\label{eq:def2}
\end{align}
where $T>0$ is fixed.

In fact as $d\rightarrow\infty$, there exists ${\widehat W}^d\Rightarrow W$ such that 
$z^d(t)$ admits the following representation: 
\begin{equation}
z^d(t)=z^0-g(\ell)\int_0^t\left(z^d(s)+\Sigma\nabla\Psi\left(z^d(s)\right)\right)ds+
\sqrt{2g(\ell)}\widehat W^d(t).
\label{eq:approx4}
\end{equation}

It can be shown, proceeding in the same way, and using the same assumptions on the covariance operator 
and $\Psi$ as \ctn{Pillai2011}, that $z^d(t)$ converges weakly to $z$ (see \ctn{Pillai2011} for the rigorous definition), 
where $z$ satisfies the SDE given by 
\begin{equation}
\frac{dz}{dt} = - g(\ell) \left (z + \Sigma\nabla \Psi(z) \right ) + \sqrt{2g(\ell)}\frac{dW}{dt}, \ \ z(0)=z^0,
\label{eq:diffusion_equation}
\end{equation}
where $z^0\sim\pi$, $W$ is a Brownian motion in a relevant Hilbert space with covariance operator $\Sigma$, and 
\begin{equation}
g(\ell)=\ell^2\beta,
\label{eq:general_hl}
\end{equation}
is the diffusion speed. 


Our result, which we state as Theorem \ref{theorem:theorem6}, requires the same assumptions 
on the decay of eigenvalues $\lambda^2_j$ of $\Sigma$ and properties of $\Psi$ that
were also required by \ctn{Pillai2011}. For the sake of completeness we present these
assumptions below. But before that we need to define some new notation, as follows.

Using the expansion (\ref{eq:fourier_expansion}), following \ctn{Pillai2011} 
we define the Sobolev spaces $\mathbb H^r$; $r\in\mathbb R$, where the inner products
and norms are defined by
\[
\langle x,y\rangle_r=\sum_{j=1}^{\infty}j^{2r}x_jy_j,\quad\quad \|x\|^2_r=\sum_{j=1}^{\infty}j^{2r}x^2_j.
\]
For an operator $L:\mathbb H^r\rightarrow\mathbb H^l$, we denote, following \ctn{Pillai2011},
the operator norm on $\mathbb H$ by $\|L\|_{\mathcal L(\mathbb H^r,\mathbb H^l)}$ defined by
\[
\|L\|_{\mathcal L(\mathbb H^r,\mathbb H^l)} = \underset{\|x\|_r=1}{\sup}\|L x\|_l.
\]

\subsubsection{Assumptions}
\label{subsubsec:assumptions}

(1) {\it Decay of eigenvalues $\lambda^2_j$ of $\Sigma$}:
There exist $M_{-},M_{+}\in (0,\infty)$ and $\kappa>\frac{1}{2}$ such that
\begin{equation}
M_{-}\leq j^{\kappa}\lambda_j\leq M_{+},\quad\forall~ j\in\mathbb Z_{+}=\{1,2,3,\ldots\}.
\label{eq:as1}
\end{equation}
(2) {\it Assumptions on $\Psi$}: There exist constants $M_i\in\mathbb R,~i\leq 4$ and 
$s\in [0,\kappa-\frac{1}{2})$ such that
\begin{align}
M_1\leq\Psi(x) & \leq M_2\left(1+\|x\|^2_s\right)~\forall x\in\mathbb H^s \label{eq:as2}\\
\|\nabla\Psi(x)\|_{-s} &\leq M_3\left(1+\|x\|_s\right)~\forall x\in\mathbb H^s \label{eq:as3}\\
\|\partial^2\Psi(x)\|_{\mathcal L(\mathbb H^r,\mathbb H^l)}&\leq M_4\quad\forall x\in\mathbb H^s.
\label{eq:as4}
\end{align}
(3) {\it Assumptions on $\Psi^d$}:
The functions $\Psi^d$ satisfy the same conditions imposed on $\Psi$ given 
by (\ref{eq:as2}), (\ref{eq:as3}) and (\ref{eq:as4}) with the same constants
uniformly across $d$.

\begin{theorem}
\label{theorem:theorem6}
Let assumptions (1) -- (3), as stated in Section \ref{subsubsec:assumptions}, hold. 
Let $x^0\sim\pi^d$, where $\pi^d$ is given by
(\ref{eq:piN}) and let $z^d(t)$ be given by 
(\ref{eq:continuous_interpolant}). Then $z^d$ converges weakly to the diffusion process $z$ given by
(\ref{eq:diffusion_equation}) with $z(0)\sim\pi$.
\end{theorem}


\begin{corollary}
\label{cor:cor5}
The diffusion speed $g(\ell)$ is maximized by
\begin{equation}
\ell_{opt}=\frac{2.426}{\sqrt{2}}=1.715,
\label{eq:l5}
\end{equation}
and the optimal acceptance rate is given by
\begin{equation}
\alpha_{opt}=4\int_0^{\infty}\Phi\left(-\frac{\ell_{opt}u}{\sqrt{2}}\right)\phi(u)du
=0.439\ \ \mbox{(up to three decimal places)}.
\label{eq:alpha_opt_general}
\end{equation}
\end{corollary}

\subsection{TMCMC within Gibbs for the dependent family of distributions}
\label{subsec:tmcmc_within_gibbs_dependent}

As before, here we define
transitions of the form (\ref{eq:transition}), where the random variable $\chi_i;~i=1,\ldots,d$
indicates whether or not the $i$-th
coordinate of $x$ will be updated. 
Formally,
\begin{equation}\label{eq:move2}
x^{k+1} = { \gamma^{k+1}{y}^{k+1}} + \left ( 1- { \gamma^{k+1}} \right) x^{k}, 
\end{equation}
where 
$$ {\gamma}^{k+1} \sim \mbox{Bernoulli} \left ( \min \left \{1, \frac{\pi^d({y}^{k+1})}{\pi^d(x^{k})}\right \}\right ). $$

We define the new move ${y}^{k+1}$ of the same form as (\ref{eq:ynew}), but with the indicator variables
$\chi_i$ incorporated appropriately. In other words,
\begin{equation}\label{eq:ynew2}
{y}^{k+1} = x^{k} + \sqrt{\frac{2\ell^{2}}{d}}\Sigma^{\frac{1}{2}}\xi^{k+1},
\end{equation}
 where $ \xi^{k+1} = \left ({\chi^{k+1}_1b_{1}}^{k+1}{\epsilon}^{k+1}, \ldots, \chi^{k+1}_d{b_{d}}^{k+1}{\epsilon}^{k+1} \right )$; 
$b_{i}=\pm 1$ with probability 1/2 each, $\epsilon\sim N(0,1)I_{\{\epsilon>0\}}$, and
 for any $k>0$ and for $i=1,\ldots,d$, $P(\chi^{k+1}_i=1)=c_d$. As before, we assume that $c_d\rightarrow c$ as $d\rightarrow\infty$,
 where $0<c\leq 1$. 

The proof again required only minor modification to the above proof provided
in the case of this dependent family of distributions. Here, additionally, we only need to 
take expectations with respect to $\chi^{k+1}_i;i=1,\ldots,d$,
so that we now have
\[\left[\mathbb Q(x,\xi)\vert b_i,\epsilon\right]\approx 
d\left(-\ell^2\epsilon^2-c\epsilon\sqrt{\frac{2\ell^2}{d}}\eta_ib_i,2\ell^2\epsilon^2c^2\right).\]
Proceeding in the same manner as in the above proof, we obtain a stochastic differential equation of the same
form as (\ref{eq:diffusion_equation}), but with $g(\ell)$ replaced with
\begin{equation}
g_c(\ell)=c\ell^2\beta_c,
\label{eq:gibbs_general_hl2}
\end{equation}
where
\[\beta_c=4\int_0^{\infty}u^2\Phi\left(-\frac{\ell u}{c\sqrt{2}}\right)\phi(u)du.\]

The result can be stated formally as follows:
\begin{theorem}
\label{theorem:theorem7}
Let assumptions (1) -- (3), as stated in Section \ref{subsubsec:assumptions}, hold. 
Let $x^0\sim\pi^d$, where $\pi^d$ is given by
(\ref{eq:piN}) and let $z^d(t)$ be given by 
(\ref{eq:continuous_interpolant}), where $z^d(t)$ depends upon $x^k$ and $x^{k+1}$ through 
$ \xi^{k+1} = \left ({\chi^{k+1}_1b_{1}}^{k+1}{\epsilon}^{k+1}, \ldots, \chi^{k+1}_d{b_{d}}^{k+1}{\epsilon}^{k+1} \right )$,
where for any $k>0$ and for $i=1,\ldots,d$, $P(\chi^{k+1}_i=1)=c_d$, other definitions remaining the same as before.
Then $z^d$ converges weakly to the diffusion process $z$ having the same form as
(\ref{eq:diffusion_equation}), but $g(\ell)$ replaced with $g_c(\ell)$ given by (\ref{eq:gibbs_general_hl2}), and as before,
$z(0)\sim\pi$.
\end{theorem}


\begin{corollary}
\label{cor:cor6}
The diffusion speed $g_c(\ell)$ is maximized by
\begin{equation}
\ell_{opt}=\frac{2.426c}{\sqrt{2}}=1.715c,
\label{eq:l6}
\end{equation}
and the optimal acceptance rate is given by
\begin{equation}
\alpha_{opt}=4\int_0^{\infty}\Phi\left(-\frac{\ell_{opt}u}{c\sqrt{2}}\right)\phi(u)du
=0.439\ \ \mbox{(up to three decimal places)}.
\label{eq:gibbs_alpha_opt_general2}
\end{equation}
\end{corollary}

\section{Comparison with RWM}
\label{sec:comparison}

\subsection{Comparison in the $i.i.d.$ set-up}
\label{subsec:compare_iid}

Note that for both the standard RWM algorithm and our additive TMCMC algorithm, 
the diffusion process reduces to the Langevin diffusion having the same form, but different
diffusion speeds.
For the RWM algorithm, the diffusion speed $h(\ell)$ is given by
$h(\ell) = 2 \ell^{2} \Phi \left (-\frac{\ell \sqrt{\mathbb I}}{2}\right )$,
and the optimal acceptance rate is $2 \Phi \left (-\frac{\ell_{opt} \sqrt{\mathbb I}}{2}\right )$,
where $\ell_{opt}$ maximizes $h(\ell)$.
A comparison between (\ref{eq:diff_speed_iid}) and 
and the above diffusion speed reveals that if, instead of the standard
normal distribution, $z^*_1$ associated with equation (13) of the supplement, 
corresponding to the proof of Theorem 3.1, 
had a distribution that assigned probability $1/2$ to each of $+1$ and $-1$, then
the additive TMCMC-based diffusion speed would reduce to the RWM-based diffusion speed.

Note 
that the optimum value of $\ell$ in RWM 
is $\ell_{opt}= \frac{2.381}{\sqrt{\mathbb{I}}}$ and the corresponding expected acceptance rate is $0.234$. 
However, in TMCMC it is observed on maximizing (\ref{eq:diff_speed_iid}) 
that $\ell_{opt}= \frac{2.426}{\sqrt{\mathbb{I}}}$ and the corresponding expected acceptance rate is 
$0.439$; see Corollary \ref{cor:cor1}. Hence, although the values of the optimizer $\ell_{opt}$ are close
for RWMH and additive TMCMC, the optimal acceptance rate of the latter is significantly higher.
This much higher acceptance rate for TMCMC is to be expected because effectively just a one-dimensional
proposal distribution is used to update the entire high-dimensional random vector $x$. 

Figure~\ref{fig:fig1} compares the diffusion speeds of TMCMC and RWM in the $i.i.d.$ case.
\begin{figure}
\centering
\includegraphics[height=8cm,width=10cm]{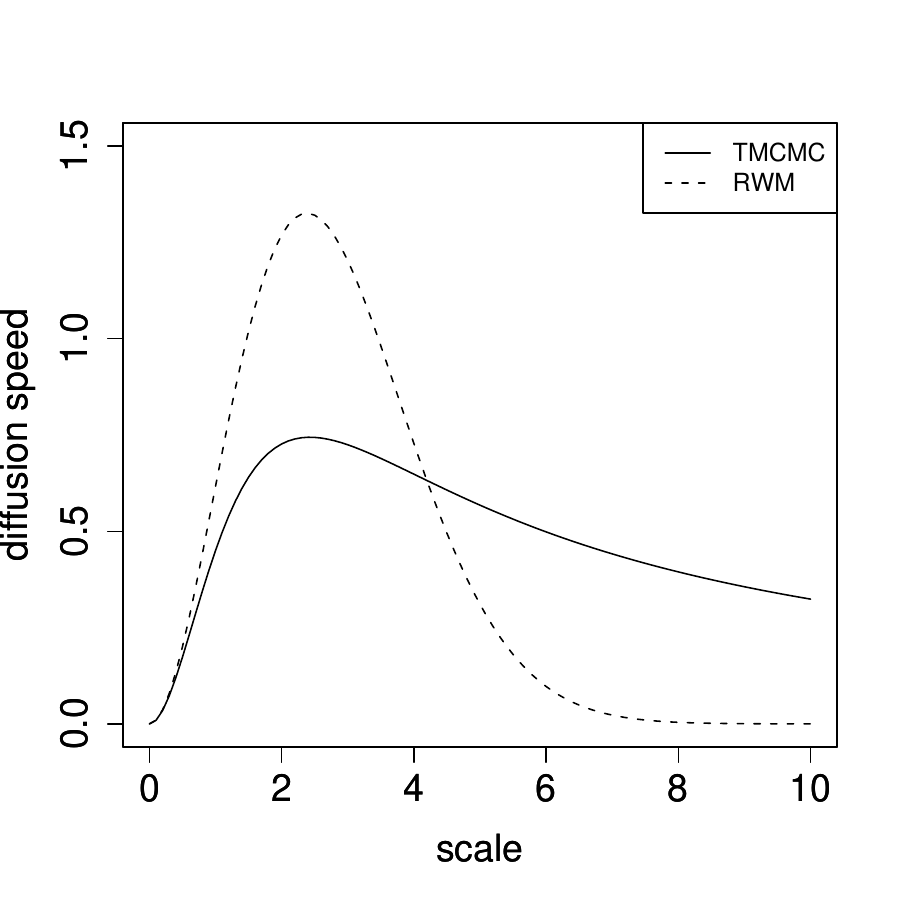}
\caption{Comparison of diffusion speeds 
of TMCMC and RWM in the $i.i.d.$ case.}
\label{fig:fig1}
\end{figure}
Observe that the maximum diffusion speed for RWM is greater than that of TMCMC.
However, the graph for RWM falls much more steeply compared to TMCMC for large $\ell$, showing that
the diffusion speed is quite sensitive towards misspecification of the scaling constant, and that
scaling constants other than the maximizer can substantially decrease the diffusion speed.
On the other hand, the graph for TMCMC is much more flat, indicating relatively more robustness
with respect to the choice of $\ell$. 

As we will see, the same phenomenon holds for all the other
set-ups, such as the target distributions with non-identical and dependent components. 
This is an important issue in practice for general high-dimensional target distributions, particularly
with non-identical and dependent components since, as discussed in Sections S-3 
and S-4, 
in practice, tuning the
scaling constants of the proposal distributions to approximately achieve the optimal acceptance rate is
generally infeasible in high dimensions, which in turn makes the maximum diffusion speed infeasible to achieve.  
For the RWM algorithm any such misspecification entails a sharp fall in the diffusion speed.
Since in high dimensions misspecifications are very much likely, RWM is quite generally prone to sub-optimal
performances. 
From the discussion presented in Section S-5 
it can be anticipated that 
in very high dimensions, it may not be practically feasible to achieve the optimal acceptance rate using 
adaptive algorithms based on RWM. 
On the other hand, additive TMCMC remains far more robust even in the face of
such misspecifications, thus significantly cutting down the risk of poor performance in high dimensions.
Adaptive algorithms based on additive TMCMC are also demonstrated by \ctn{Dey13} to be much more 
efficient compared to the adaptive RWM algorithms.

\subsubsection{Within Gibbs comparison in the $i.i.d.$ set-up}
\label{subsubsec:within_gibbs_iid}

Now we compare TMCMC within Gibbs based diffusion speed and optimal acceptance rate given by 
\begin{equation}
g_c(\ell)=4c\ell^{2}\int_0^{\infty} u^{2}\Phi 
\left (- \frac{u\ell\sqrt{c\mathbb{I}}}{2}\right) \phi(u)du.
\label{eq:diff_speed_tmcmc_gibbs}
\end{equation}
(see (26) of the supplement, Section S-6.2)
and (\ref{eq:acc_tmcmc_gibbs})  
with those of RWM within Gibbs.
The diffusion speed for the RWM within Gibbs algorithm is
$h_{c}(\ell)= 2c\ell^{2} \Phi \left (- \frac{\ell\sqrt{c\mathbb{I}}}{2}\right)$, 
and the optimal acceptance rate is 
$2\Phi \left (- \frac{\ell_{opt}\sqrt{c\mathbb{I}}}{2}\right)$,
where $\ell_{opt}$ maximizes $h_c(\ell)$;
see \cite{NealRoberts}. 
It turns out that $\ell_{opt}$ for RWM within Gibbs is given by $ \frac{2.381}{\sqrt{c\mathbb{I}}}$,
and the optimal acceptance rate is 0.234, as before.
Figure~\ref{fig:fig2} compares the diffusion speeds associated with TMCMC within Gibbs and RWM within Gibbs,
with $c=0.3$. Once again, we observe that the diffusion speed of TMCMC within Gibbs is more robust
with respect to misspecification of the scale.
\begin{figure}
\centering
\includegraphics[height=8cm,width=10cm]{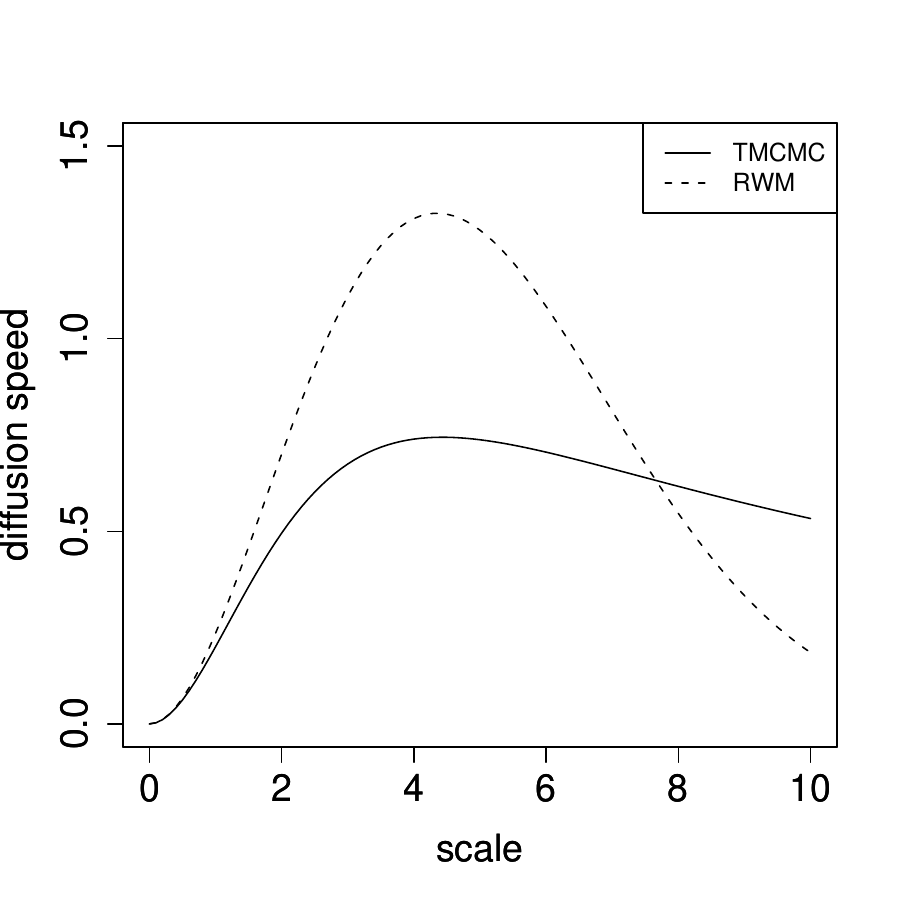}
\caption{Comparison of diffusion speeds 
of TMCMC within Gibbs and RWM
within Gibbs in the $i.i.d.$ case, with $c=0.3$.}
\label{fig:fig2}
\end{figure}

\subsection{Comparison in the independent but non-identical set-up}
\label{subsec:compare_non_iid}

The equations 
\begin{equation}
g_{\xi}(\ell)= 4\ell^{2}\int_0^{\infty} \left \{{u}^{2}\Phi 
\left (- \frac{u\ell\xi\sqrt{\mathbb{I}}}{2}\right) \right \}\phi(u)du.
\label{eq:diff_speed_non_iid}
\end{equation}
(see also (39) of the supplement) 
and (\ref{eq:opt_acc_non_iid}) 
provide the diffusion speed and the optimal acceptance rate for TMCMC in the independent
but non-identical set-up. The corresponding quantities for RWM are given by
$2 \ell^{2} \Phi \left (-\frac{\ell\xi \sqrt{\mathbb I}}{2}\right )$
and $2 \Phi \left (-\frac{\ell_{opt}\xi \sqrt{\mathbb I}}{2}\right )$, respectively.
As before, the optimal acceptance rates remain 0.234 and 0.439 for RWM and TMCMC, respectively.
Figure~\ref{fig:fig3} compares the diffusion speeds associated with TMCMC and RWM,
with $\xi=10$. 
Here both the graphs are steep, but that for RWM is much more steeper, leading to the
same observations regarding robustness with respect to misspecification of scale.
\begin{figure}
\centering
\includegraphics[height=8cm,width=10cm]{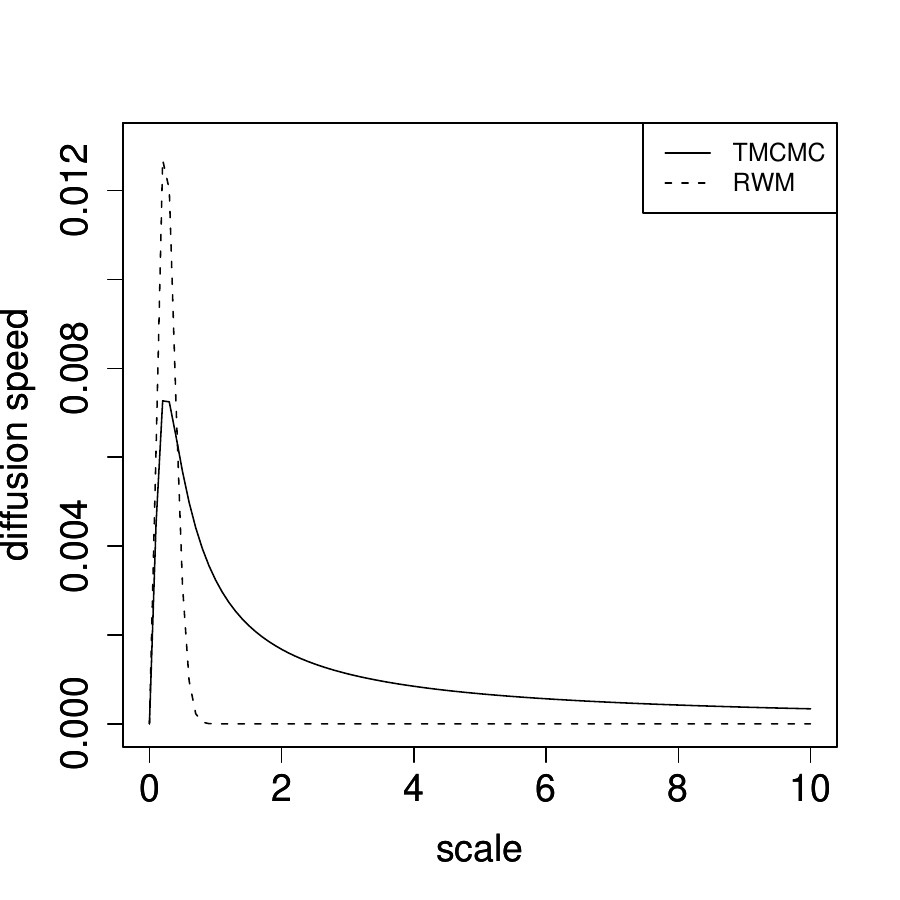}
\caption{Comparison of diffusion speeds 
of TMCMC within Gibbs and RWM
within Gibbs in the independent but non-identical case, with $\xi=10$.}
\label{fig:fig3}
\end{figure}

\subsubsection{Within Gibbs comparison in the independent but non-identical set-up}
\label{subsubsec:compare_gibbs_non_iid}

It can be easily shown that the RWM-based diffusion speed and the acceptance rate
in the independent but non-identical set-up
are
$2c \ell^{2} \Phi \left (-\frac{\ell\xi \sqrt{c\mathbb I}}{2}\right )$
and $2 \Phi \left (-\frac{\ell_{opt}\xi \sqrt{c\mathbb I}}{2}\right )$, respectively.
These are to be compared with the TMCMC-based quantities given by
(\ref{eq:diff_speed_non_iid2}) and (\ref{eq:opt_acc_non_iid2}).
The optimal acceptance rates for TMCMC and RWM, as before, are 0.234 and 0.439.
\begin{figure}
\centering
\includegraphics[height=8cm,width=10cm]{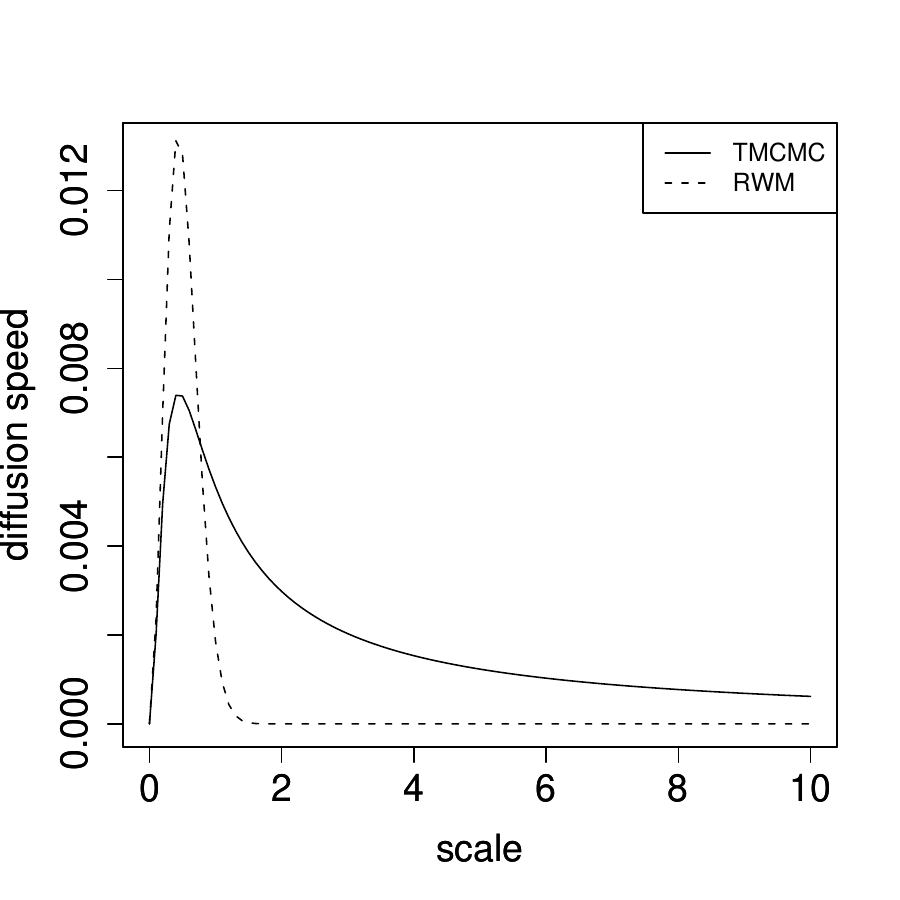}
\caption{Comparison of diffusion speeds 
of TMCMC within Gibbs and RWM
within Gibbs in the independent but non-identical case, with $\xi=10$, $c=0.3$.}
\label{fig:fig4}
\end{figure}
Conclusions similar as before are reached on observing Figure~\ref{fig:fig4} that compares
the diffusion speeds of TMCMC and RWM in this case.

\subsection{Dependent case}
\label{subsec:compare_dependent}

In the dependent case, the diffusion speed and the optimal acceptance rate of additive TMCMC are of the forms 
(\ref{eq:general_hl}) and (\ref{eq:alpha_opt_general}), respectively. As usual, the TMCMC-based
optimal acceptance rate turns out to be 0.439. The corresponding RWM-based optimal acceptance rate,
having the form
$2\Phi\left(-\frac{\ell_{opt}}{\sqrt{2}}\right)$, turns out to be 0.234 as before, where
$\ell_{opt}$ maximizes the corresponding diffusion speed 
$2\ell^2\Phi\left(-\frac{\ell}{\sqrt{2}}\right)$. 
\begin{figure}
\centering
\includegraphics[height=8cm,width=10cm]{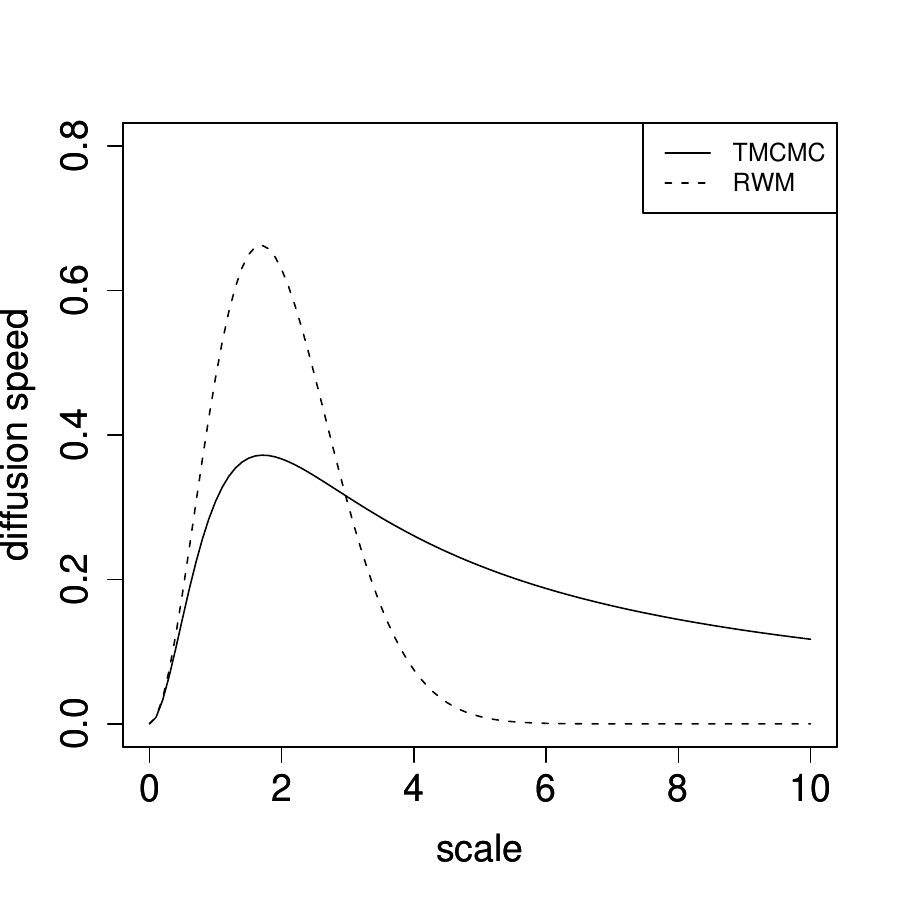}
\caption{Comparison of diffusion speeds 
of TMCMC and RWM
in the dependent case.}
\label{fig:fig5}
\end{figure}
Similar information as before are provided by Figure~\ref{fig:fig5}.

\subsubsection{Within Gibbs comparison in the dependent set-up}
\label{subsubsec:compare_gibbs_dependent}

In the dependent case, it is easily shown that the RWM-based diffusion speed and the acceptance
rate are, respectively, $2c \ell^{2} \Phi \left (-\frac{\ell}{c\sqrt{2}}\right )$
and $2 \Phi \left (-\frac{\ell_{opt}}{c\sqrt{2}}\right )$.
The corresponding TMCMC-based quantities are 
(\ref{eq:gibbs_general_hl2}) and (\ref{eq:gibbs_alpha_opt_general2}).
The optimal acceptance rates remain 0.234 and 0.439 for RWM and TMCMC.
\begin{figure}
\centering
\includegraphics[height=8cm,width=10cm]{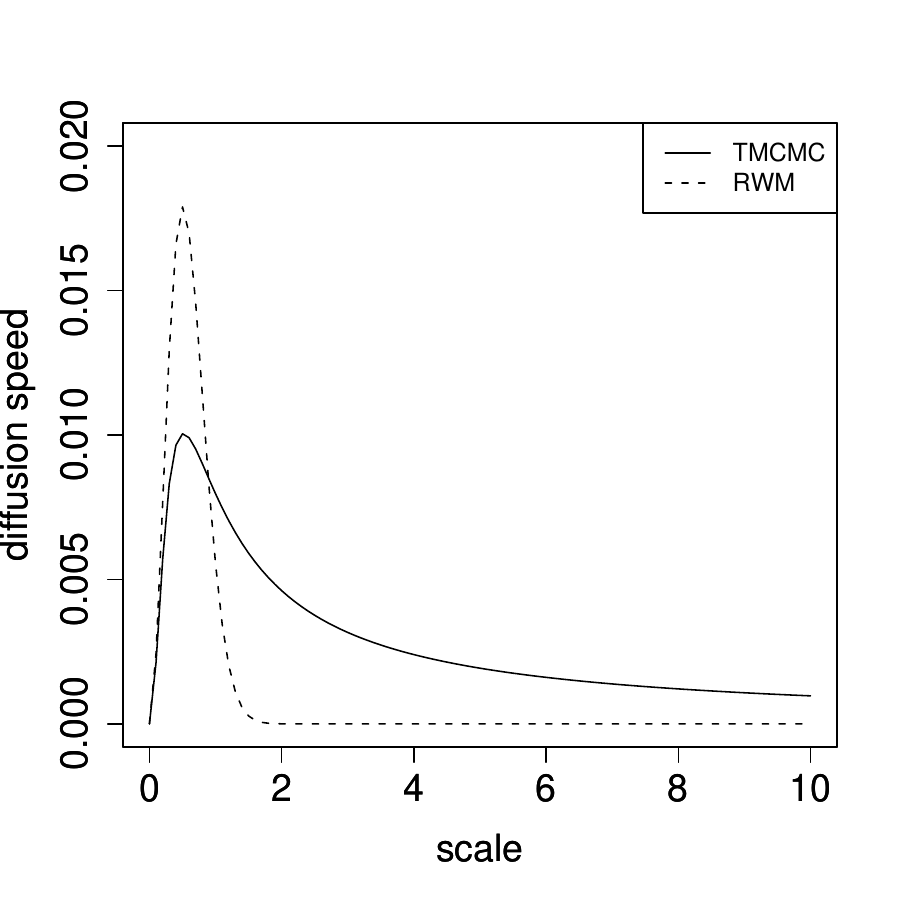}
\caption{Comparison of diffusion speeds 
of TMCMC within Gibbs and RWM
within Gibbs in the dependent case, with $c=0.3$.}
\label{fig:fig6}
\end{figure}
Figure~\ref{fig:fig6}, comparing
the diffusion speeds of TMCMC within Gibbs and RWM within Gibbs in the dependent set-up,
lead to similar observations as before.

\section{Simulation Experiments}
\label{sec:simulation}

So far, we have invested most of our efforts in the theoretical development of optimal scaling mechanism in 
the additive TMCMC case. Now, we shall consider some simulation experiments to illustrate the performance of 
our method with respect to the standard RWM methodology, under the $i.i.d.$, independent but non-identical,
and dependent set-ups.

\subsection{Comparison of additive TMCMC and RWM in the $i.i.d.$ case}
\label{subsec:iid_comparison}
We compare the performance of RWM and TMCMC 
corresponding to three different choices of the proposal variance, with scalings $\ell$ being  
2.4 (approximately optimal for both RWM and additive TMCMC) and 6 (sub-optimal for both RWM
and additive TMCMC) respectively. We consider target densities 
of dimensions ranging from 2 to 200. 
For our purpose we consider 
the target density $\pi$ to be the multivariate normal distribution with mean vector zero and
covariance matrix $I$, the identity matrix. 
The starting point $x_{0}$ is randomly generated from $U(-2,2)$, the uniform distribution
on $(-2,2)$. The univariate density of $\epsilon$ for TMCMC was taken to be a left-truncated normal 
having mean 0 and 
variance $\frac{\ell^{2}}{d}$ for each coordinate, where $\ell$ is the value of the scaling constant. 
For RWM, each coordinate of the $d$ dimensional proposal density was assumed to have the above distribution,
but without the truncation. 

In each run, the chain was observed up to 100,000 trials (including the rejected moves). 
The choice of burn-in was made somewhat subjectively, removing one fourth of the total number of 
iterates initially. This choice was actually a bit conservative as both RWM and TMCMC were found 
to be sufficiently close to the target density well ahead of the chosen point. We measured the efficiency 
of the TMCMC chain with respect to the RWM chain using certain performance evaluation 
measures -- \emph{Acceptance rate}, \emph{Average Jump Size} (AJS), \emph{Integrated Auto-Correlation Time} (IACT) 
and \emph{Integrated Partial Auto-Correlation Time} (IPACT) (see \cite{Roberts09}). All calculations of 
AJS, IACT, IPACT were done corresponding to the process after burn-in in order to ensure stationarity. 
In calculating the integrated autocorrelation time, we considered $25$ lags of ACF. IPACT was similarly 
computed. The first eight columns of Table~\ref{tab:table1} compare the performances of 
TMCMC and RWM with respect to these measures. 

\subsubsection{Average Kolmogorov-Smirnov distance for comparing convergence of TMCMC and RWM}
\label{subsec:ks_distance}

The measures acceptance rate, IACT, IPACT and AJS do not explicitly measure how close the
MCMC-based empirical distribution is to the target distribution. For this
we also considered the Kolmogorov-Smirnov (K-S) distance to evaluate the performances of the 
MCMC algorithms. 
We ran 100 copies of the RWM and TMCMC chains starting 
from the same initial point and with the same target density $\pi$ and observed how well the empirical 
distribution corresponding to these 100 copies, after discarding the burn-in period, fits the true density 
by evaluating the 
K-S distance (\cite{Smirnov1948}) at each time point for both the chains. As an overall measure
we take the average of the K-S distances over all the time points. 
This averaging over the time points makes sense since the chains are assumed to be in stationarity
after the burn-in period, and hence every time point must yield the same (stationary) distribution.
Our average K-S distance can be viewed as quantifying how well the MCMC algorithm explores
the stationary distribution after convergence is attained.
The average K-S distances for RWM and TMCMC are shown in the
last two columns of Table \ref{tab:table1}.

\begin{sidewaystable}[h]
\centering
\caption{The performance evaluation of RWM and TMCMC chains for different dimensions. 
It is assumed that proposal has independent normal components for RWM with same proposal variance 
along all coordinates. The proposal scales are 2.4 (optimal) and 6 (sub-optimal). 
All calculations done after burn in.}
\vspace{2cm}

\begin{tabular}{|c|c|c|c|c|c|c|c|c|c|c|c|}
\hline
\multirow{2}{*}{Dimension} & \multirow{2}{*}{\backslashbox{Scaling}{Test}} 
& \multicolumn{2}{|c|}{$\begin{array}{c} \mbox{Acceptance} \\ \mbox{Rate} ($\%$) \end{array} $} & \multicolumn{2}{|c|}{\mbox{IACT}} & \multicolumn{2}{|c|}{\mbox{IPACT}} & \multicolumn{2}{|c|}{\mbox{AJS}} 
& \multicolumn{2}{|c|}{$\begin{array}{c}\mbox{Average}\\ \mbox{K-S distance}\end{array}$}\\ \cline{3-12}
& & RWM & TMCMC & RWM & TMCMC & RWM & TMCMC & RWM & TMCMC & RWM & TMCMC\\ \hline

\multirow{3}{*}{2} & 2.4 (opt) & 34.9 & 44.6 & 6.08 & 7.04 & 2.46 & 2.55 & 0.93 & 0.74 & 0.1651 & 0.1657 \\ 
 & 6 (sub-opt) & 18.66 & 29.15 & 7.08 & 8.08 & 2.52 & 2.56 & 0.79 & 0.62 & 0.1659 & 0.1655 \\ 
 \hline
\multirow{3}{*}{5} & 2.4 (opt) & 28.6 & 44.12 & 9.98 & 12.45  & 2.67 & 2.77 & 1.15 & 0.79 & 0.1659 & 0.1664 \\ 
& 6 (sub-opt) & 2.77 & 20.20 & 15.6 & 14.11 & 2.77 & 2.81 & 0.39 & 0.48 & 0.1693 & 0.1674 \\ 
\hline
\multirow{3}{*}{10} & 2.4 (opt) & 25.6 & 44.18 & 15.16 & 18.26 & 2.77 & 2.88 & 1.22 & 0.73 & 0.1667 & 0.1677\\ 
& 6 (sub-opt) & 1.37 & 20.34 & 17.55 & 16.31 & 2.91 & 2.86 & 0.25 & 0.49 & 0.1800 & 0.1688\\ 
\hline
\multirow{2}{*}{100} & 2.4 (opt) & 23.3 & 44.1 & 18.14 & 18.46 & 2.88 & 2.89 & 1.34 & 0.73 & 0.1794 & 0.1671\\ 
& 6 (sub-opt) & 0.32 & 20.6 & 18.62 & 18.25 & 2.89 & 2.88 & 0.26 & 0.69 & 0.1787 & 0.1684\\ \hline
\multirow{2}{*}{200} & 2.4 (opt) & 23.4 & 44.2 & 18.4 & 18.67 & 2.88 & 2.89 & 1.3 & 0.92 & 0.1813 & 0.1735 \\ 
& 6 (sub-opt) & 0.33 & 20.7 & 18.86 & 18.74 & 2.89 & 2.89 & 0.09 & 0.54 & 0.1832 & 0.1755\\ \hline
\end{tabular}
\label{tab:table1}
\end{sidewaystable}

\subsubsection{Observations regarding the results presented in Table \ref{tab:table1}}
As evident from Table~\ref{tab:table1}, TMCMC seems to have a uniformly better 
acceptance rate than RWM for all dimensions and all choices of proposal variances. 
There is sufficient gain in acceptance rate over RWM even for 2 dimensions and the difference 
increases once we move to higher dimensions or consider larger proposal variances. 
That large proposal variance would affect the performance of RWM is intuitively clear, because 
in this case getting an outlying observation in any of the $d$ coordinates becomes more likely.

An interesting observation from Table~\ref{tab:table1} is that even for 2 dimensions, 
our acceptance ratio corresponding to the optimal scaling of $2.4$ is very close to 0.44 and 
it remains close to the optimal value for all the dimensions considered. 
It is interesting to note that 0.44 is also the (non-asymptotic) optimal acceptance rate of RWM 
for one-dimensional proposals in certain settings obtained by minimizing the first order
auto-correlation of the RWM algorithm; see \ctn{Roberts01}, \ctn{Roberts09}.
Since in one dimension additive TMCMC is equivalent to RWM and because the former
is effectively a one-dimensional algorithm irrespective of the actual dimensionality, 
this perhaps intuitively suggests that for TMCMC, the optimal acceptance rate will remain
very close to 0.44 irrespective of dimensionality. 
For RWM however, 
the optimal acceptance rate is quite far from 0.234 for smaller dimensions. From the asymptotics
perspective (setting aside the above argument regarding TMCMC being effectively one-dimensional for
any actual dimension), this demonstrates 
that convergence to the diffusion equation occurs at a much faster rate in TMCMC as compared
to RWM. Hence, even in smaller dimensions a TMCMC user can tune the proposal to achieve
approximately 44\% acceptance rate. Indeed, in low dimensions the tuning exercise is far more easier
than in higher dimensions.

When the scale is changed from the optimum value $2.4$ to the sub-optimal value 6,
we witness very significant drop in the acceptance rates of RWM. Particularly for dimensions 
$d=100$ and $d=200$ the acceptance rate of RWM 
falls off very sharply and becomes almost negligible. 
In keeping with the discussion presented in 
Sections S-3 
and S-4 of the supplement
this indicates how difficult
it can be in the case of general, high-dimensional target distributions, to adjust the RWM proposal
to achieve the acceptance rates between 15\% and 50\%, as suggested by \ctn{Roberts01}.
On the other hand, for any dimension, the 
acceptance rate of TMCMC remains more than 20\%, indicating it is a lot more easier and safer to
tune the TMCMC proposal. 

The measure IACT is uniformly higher for TMCMC for all dimensions 
when the optimal scale is considered. This is to be
expected since the maximum diffusion speed is higher for RWM, and IACT decreases as diffusion
speed increases. However, when the scale is sub-optimal, IACT of TMCMC is uniformly lower than
that of RWM in all dimensions. This is in accordance with the discussion on the lack of robustness 
of RWM and the relative robust behaviour 
of the diffusion speed of TMCMC with respect to scale changes, 
presented in Sections S-3 
and S-4 of the supplement. 
Indeed, the sub-optimal scale choice causes the diffusion speed of RWM to drop sharply, 
increasing the integrated autocorrelation in the process. On the other hand, the diffusion
speed of TMCMC remains relatively more stable, thus not allowing IACT to increase significantly.

Although in the lower dimensions IPACT is slightly higher for TMCMC than for RWM, 
in dimensions 10, 100 and 200, it is slightly lesser for TMCMC when the scale is suboptimal
(for $d=200$ IPACT is almost the same for both the algorithms in the sub-optimal case).

The average jump size, AJS, is uniformly somewhat larger for RWM compared to TMCMC when the scale is optimally
chosen. However, for the sub-optimal scaling, AJS for TMCMC is significantly larger
than those for RWM for dimensions $d=5,10,100,200$. Since in general sub-optimal scaling is to be
expected, as per the discussions in Sections S-3 
and S-4 of the supplement, 
one can expect better exploration (in terms of AJS) of the general, high-dimensional target density, 
by additive TMCMC.

For dimensions $d=100$ and $d=200$, the average K-S distance is smaller for TMCMC with respect to both
optimal and sub-optimal scales. Moreover, for the sub-optimal scale, the K-S distance is uniformly smaller
for TMCMC for all the dimensions considered. Furthermore, note that for the sub-optimal scale, 
as the dimension increases, the difference between the average K-S distances of RWM and TMCMC also increases.
This suggests that at least when the scale is sub-optimal, TMCMC performs increasingly better than 
RWM in terms of better exploration of the target density, as dimension increases. 


\subsubsection{Visualizing the rate of convergence of TMCMC and RWM to the stationary distribution
using Kolmogorov-Smirnov distance}
\label{subsec:ks_distance_conv}

Apart from measuring the performance of the chains after stationarity, 
one might be interested in visualizing how fast the chains converge to the target density 
starting from an initial value. In other words, it is of interest to know 
which of these chains 
have a steeper downward trend with respect to the other, when the respective optimal scales are used for
both the algorithms. 
To investigate this empirically, we again use the K-S distance, plotting the distances with respect
to the iteration number (time). Thus, while the average K-S distance, calculated after the burn-in,
provides an overall measure of how well an MCMC algorithm explores the stationary distribution
after convergence,
a simple plot of the K-S distances with respect to time can help visualize the
rate of convergence of the MCMC algorithm to stationarity.

For smaller dimensions like 2 and 10, 
we did not perceive much difference between the two chains in terms of the plots of the K-S distance. 
But for higher dimensions, we observed a significant improvement in convergence for our TMCMC method 
in comparison with that of RWM. Two instances, for dimensions $d=30$ and $d=100$,
are presented in Figures \ref{fig:fig11} and \ref{fig:fig12}, respectively.
\begin{figure}
\centering
\subfigure [$d=30,~\ell=2.4.$]{ \label{fig:KS1_30}
\includegraphics[trim= 0cm 8cm 0cm 8cm, clip=true, width=17cm,height=8cm]{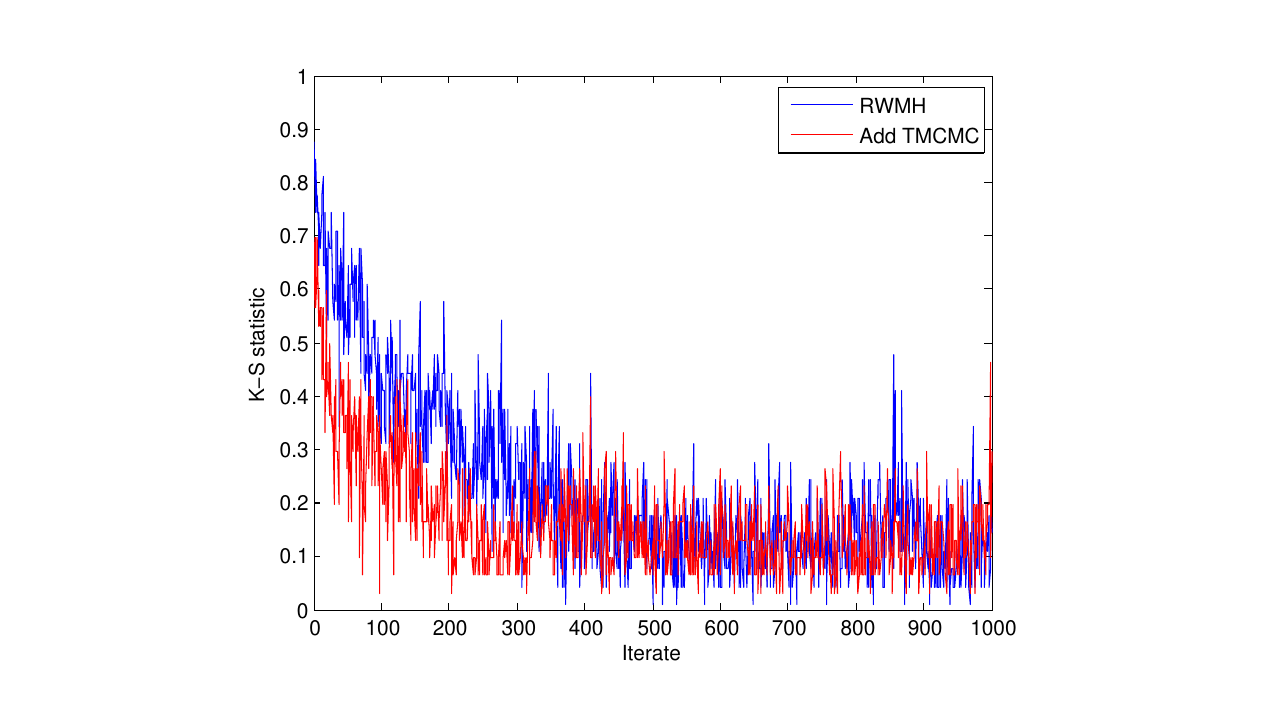}}\\
\subfigure [$d=30,~\ell=6.$]{ \label{fig:KS2_30}
\includegraphics[trim= 0cm 8cm 0cm 8cm, clip=true, width=17cm,height=8cm]{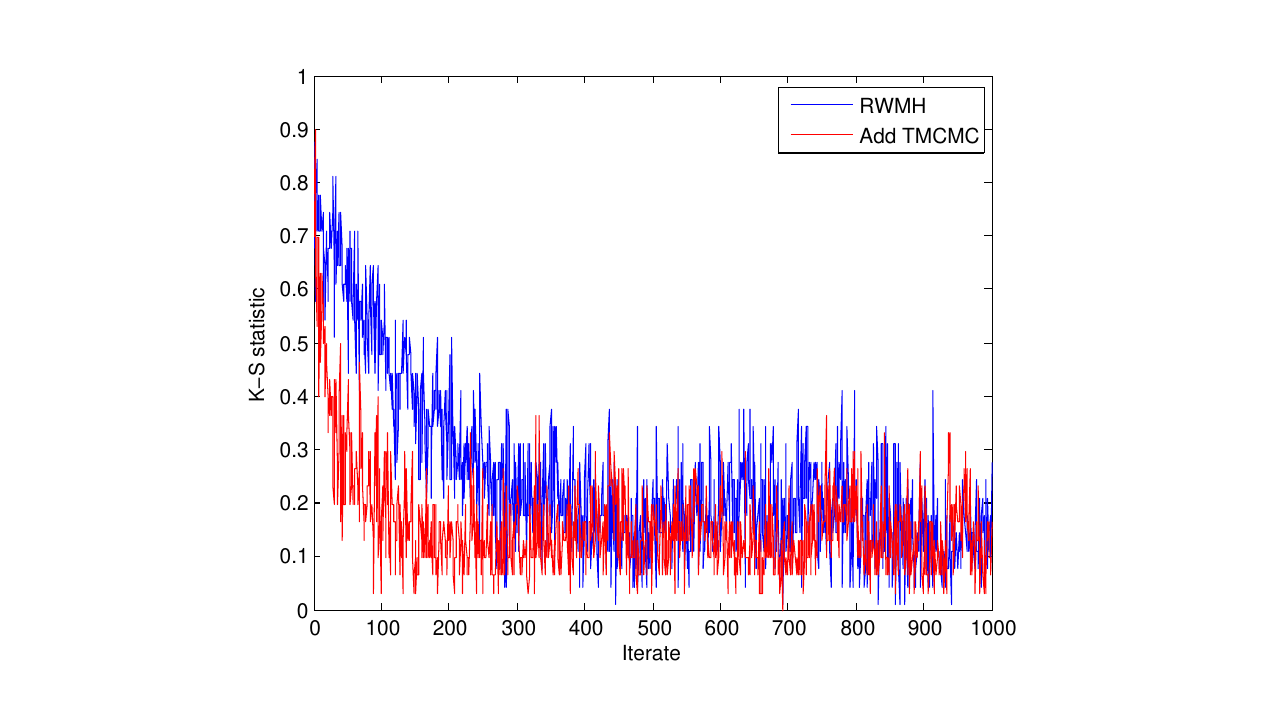}}
\caption{K-S distance comparison before burn-in between the RWM and the TMCMC chains for dimension
$d=30$.}
\label{fig:fig11}
\end{figure}

\begin{figure}
\centering
\subfigure [$d=100,~\ell=2.4.$]{ \label{fig:KS1_100}
\includegraphics[trim= 0cm 8cm 0cm 8cm, clip=true, width=17cm,height=8cm]{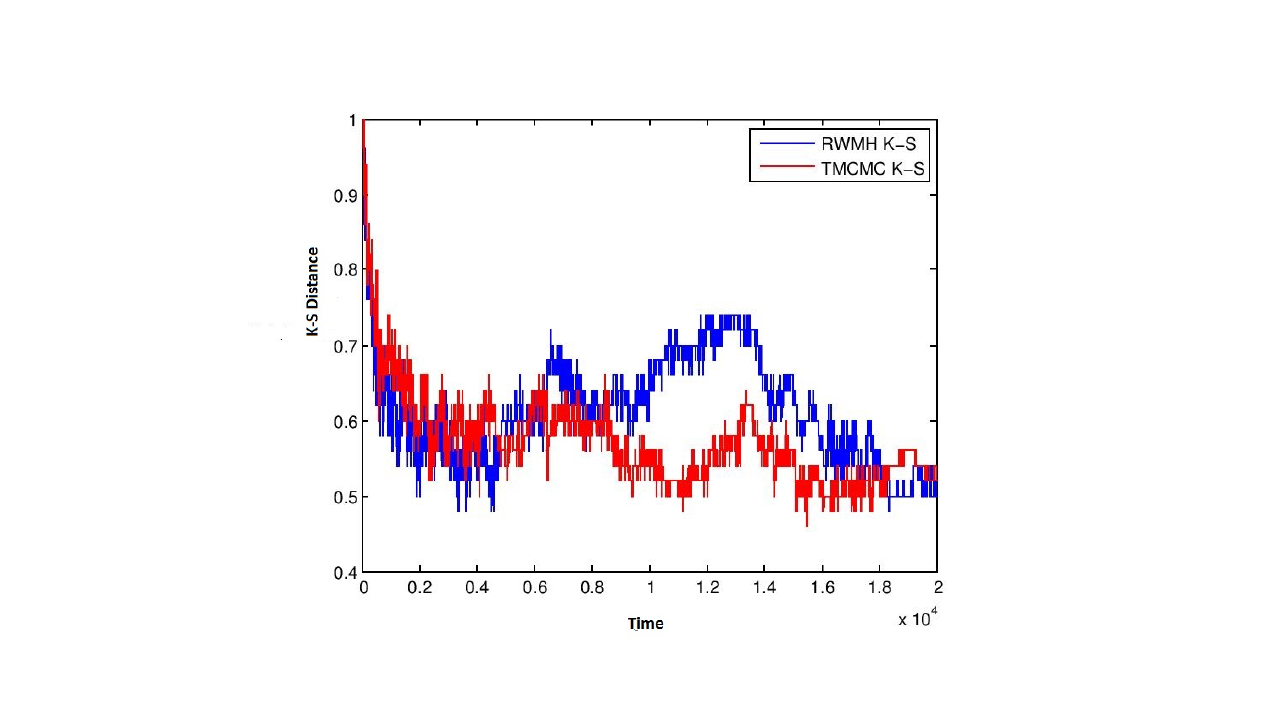}}\\
\subfigure [$d=100,~\ell=6.$]{ \label{fig:KS2_100}
\includegraphics[trim= 0cm 8cm 0cm 8cm, clip=true, width=17cm,height=8cm]{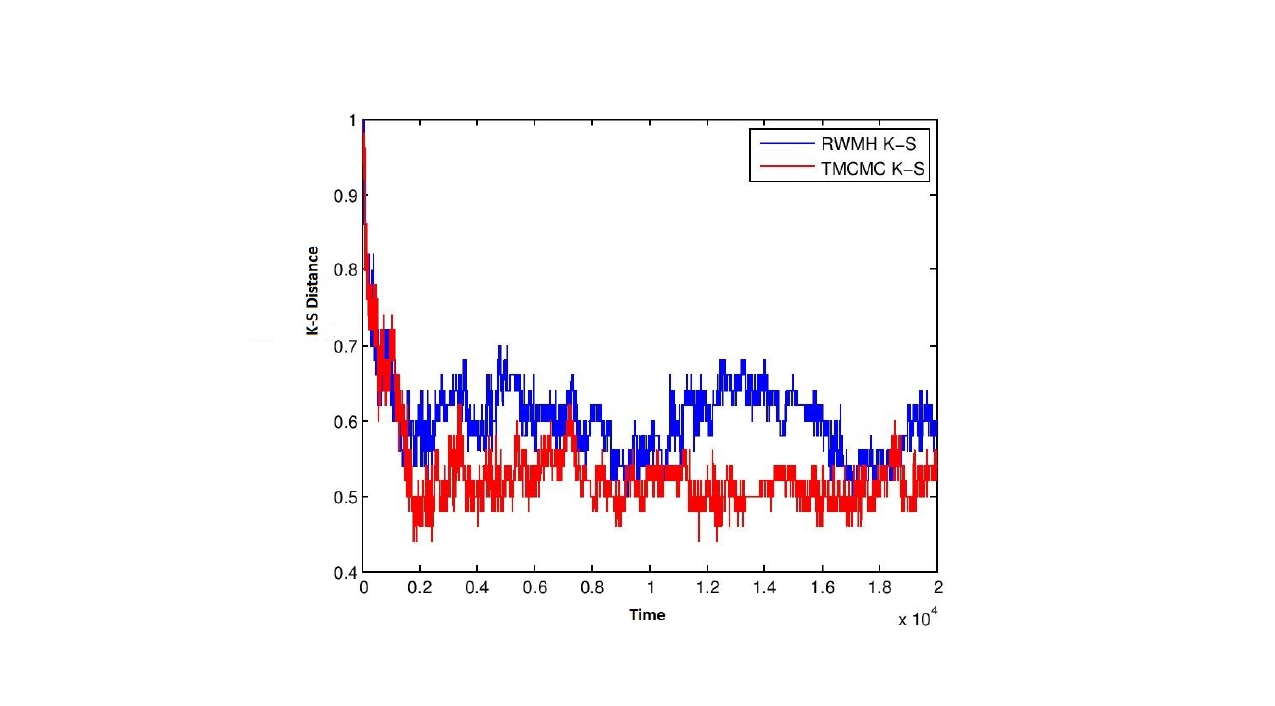}}
\caption{K-S distance comparison before burn-in between the RWM and the TMCMC chains for dimension
$d=100$.}
\label{fig:fig12}
\end{figure}

\subsection{Comparisons between additive TMCMC and RWM in the independent, but non-identical set-up}
\label{subsec:non_iid}

We now compare additive TMCMC with RWM under an instance of independent, but non-identical
situation provided in \ctn{Bedard08b}.
In particular, we assume the target distribution to have independent normal components with all the means zero,
and variances given by 
$\theta^{-2}(d) = (d^{-1/5},d^{-1/5},3,d^{-0.5},3, d^{-0.5},\ldots,3, d^{-0.5})$.
For our purpose, 
we set $d=50$ and implement 30 chains each for additive TMCMC and RWM, each chain run for
$10,000$ iterations. For any given iteration, and for both TMCMC and RWM, we then compute the K-S distances based on the
30 chains, which are then compared. 





Note that $\theta^{-2}(d)$ consists of three forms of coordinates, namely, $\theta^{-2}(d)=3$,
$\theta^{-2}(d)=d^{-0.5}$, and $\theta^{-2}(d)=d^{-\frac{1}{5}}$. These are associated with three
distinct marginal target densities. In the figures below, for these three distinct marginals, 
we separately compare TMCMC and RWM using K-S distances.
That is, Figures \ref{fig:bedard3}, \ref{fig:bedard4} and \ref{fig:bedard5} compare TMCMC and RWM when 
$\theta^{-2}(d)=3$, $\theta^{-2}(d)=d^{-0.5}$, and $\theta^{-2}(d)=d^{-\frac{1}{5}}$, respectively. 

All the three instances, Figures \ref{fig:bedard3}, \ref{fig:bedard4} and \ref{fig:bedard5}, clearly demonstrate the superiority
of additive TMCMC over RWM. 

\begin{figure}[h]
\centering
\includegraphics[width=12cm,height=7cm]{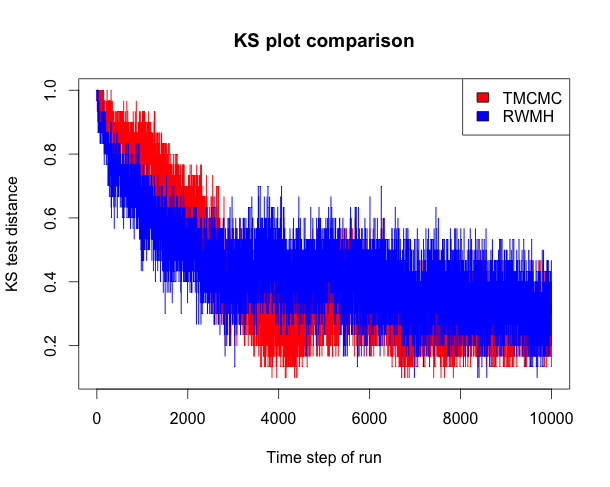}
\caption{Independent but non-identical case: K-S comparisons between additive TMCMC and RWM for $\theta^{-2}(d)=3$.}
\label{fig:bedard3}
\end{figure}

\begin{figure}[h]
\centering
\includegraphics[width=12cm,height=7cm]{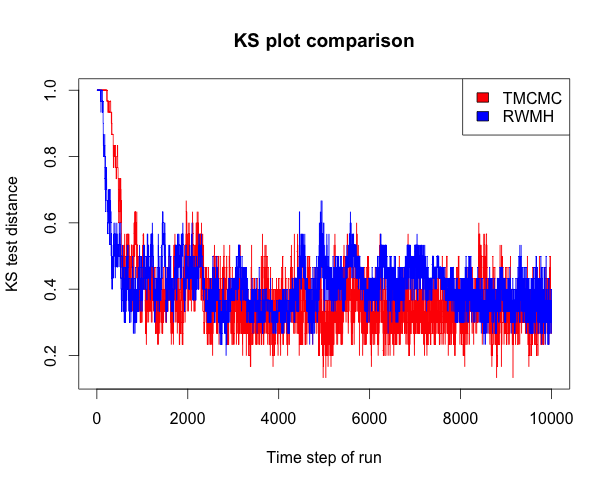}
\caption{Independent but non-identical case: K-S comparisons between additive TMCMC and RWM for $\theta^{-2}(d)=d^{-0.5}$.}
\label{fig:bedard4}
\end{figure}

\begin{figure}[h]
\centering
\includegraphics[width=12cm,height=7cm]{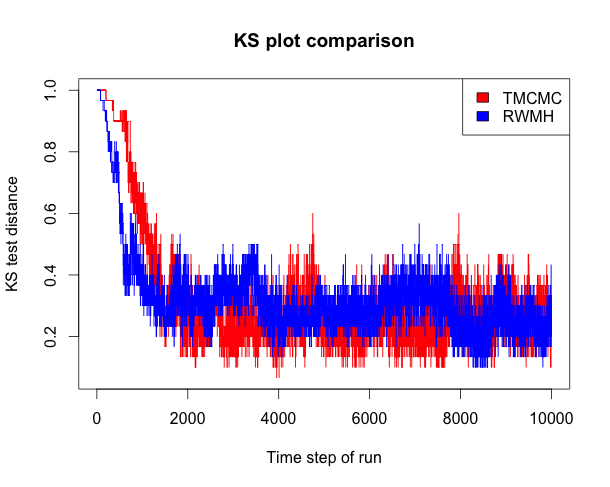}
\caption{Independent but non-identical case: K-S comparisons between additive TMCMC and RWM 
for $\theta^{-2}(d)=d^{-\frac{1}{5}}$.}
\label{fig:bedard5}
\end{figure}

\subsection{Comparisons between additive TMCMC and RWM in the dependent set-up}
\label{subsec:dependent}
We now compare additive TMCMC and RWM in the dependent set-up, as (\ref{eq:dep3}). That is, here we 
consider target densities of the type
$$\pi_d (x_d) = \exp(-x_d^{T} M_d x_d) \prod_{i=1}^{d} \frac{1}{\lambda_i} \phi \left(\frac{x_i}{\lambda_i}\right).$$
For our purpose, we consider the following forms of $\blambda=(\lambda_1,\ldots,\lambda_d)$ and $M_d$:
$$\blambda = \alpha\left (1, \frac{1}{d}, \frac{1}{d^2}, \cdots, \frac{1}{d^d} \right),$$
and
$$M_d = \gamma(1-\rho) \bI_d + \rho \bone_d \bone^{T}_d,$$ where $\bI_d$ is the identity matrix of order $d$ and
$\bone_d$ is the $d$-component vector with all elements $1$.
We report the results of our experiments with three set-ups: (a) $\rho=0.3$, $\alpha=0.1$, $\gamma=100$;
(b) $\rho=0.3$, $\alpha=0.01$, $\gamma=100$, and (c) $\rho=0.3$, $\alpha=0.01$, $\gamma=1000$. The corresponding
K-S based comparisons are provided in Figures \ref{fig:mattingly1}, \ref{fig:mattingly2} and \ref{fig:mattingly3}.
In all the three experiments TMCMC very convincingly outperforms RWM.
With various other choices of $\rho$, $\alpha$ and $\gamma$ we observed similar results (not reported due to lack of space).
\begin{figure}[h]
\centering
\includegraphics[width=12cm,height=7cm]{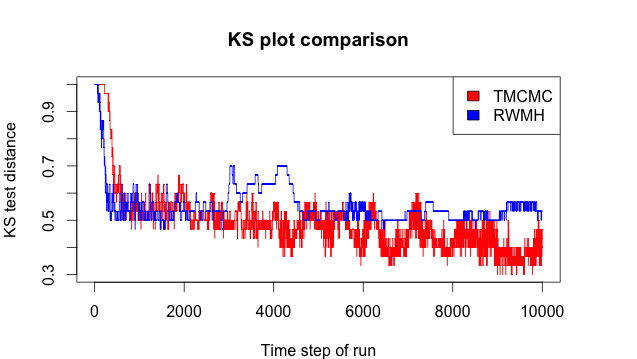}
\caption{Dependent case: K-S comparisons between additive TMCMC and RWM for $\rho=0.3$, $\alpha=0.1$, $\gamma=100$.}
\label{fig:mattingly1}
\end{figure}
\begin{figure}[h]
\centering
\includegraphics[width=12cm,height=7cm]{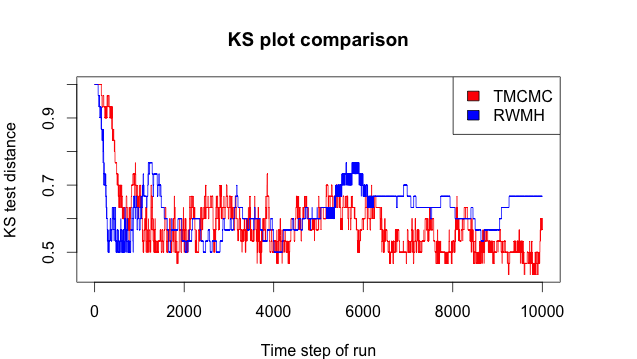}
\caption{Dependent case: K-S comparisons between additive TMCMC and RWM for $\rho=0.3$, $\alpha=0.01$, $\gamma=100$.}
\label{fig:mattingly2}
\end{figure}
\begin{figure}[h]
\centering
\includegraphics[width=12cm,height=7cm]{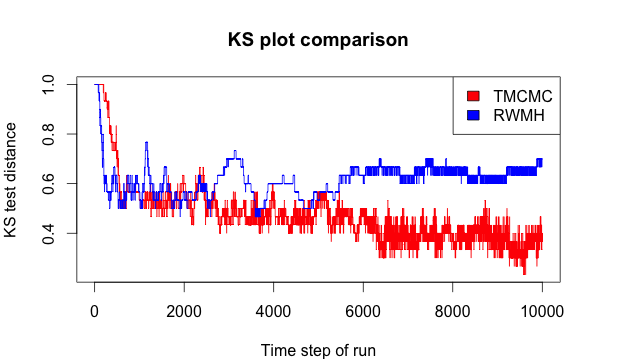}
\caption{Dependent case: K-S comparisons between additive TMCMC and RWM for $\rho=0.3$, $\alpha=0.01$, $\gamma=1000$.}
\label{fig:mattingly3}
\end{figure}

\subsection{Discussion on simulation studies with multivariate Cauchy and multivariate $t$ as target densities}
\label{subsec:multivariate_simstudy}
We have so far restricted ourselves to comparisons between TMCMC and RWM when the target distribution
is Gaussian ($i.i.d.$, independent but non-identical, and dependent).
However, \ctn{Dey15} conduct comparative
studies between the two algorithms when the targets are multivariate Cauchy and multivariate $t$.
Briefly, they compare the algorithms with respect to K-S distance,
when the location vectors and scale matrices are $\bmu=\bzero_d$ and 
$\bSigma=diag\{0.7\bone_d'\}+0.3\bone_d\bone_d'$, respectively, where 
$\bzero_d$ is a $d$-dimensional vector with all elements $0$.
For multivariate $t$, they choose $\nu=10$ degrees of freedom. 
For both RWMH and additive TMCMC they consider the scale of the proposal distribution to be 2.4, and illustrate
the methods for dimension $d=50$.

In fact, they compare the performances of RWM and additive TMCMC by two methods. In one method, since the above-mentioned
target densities are not in the super-exponential family (see \ctn{Dey15} and the references therein),
they transform them to superexponential distributions using a diffeomorphism proposed by \ctn{Johnson12},
obtain samples from the transformed target densities using RWM and TMCMC, and then give inverse transformations
to the simulated values so that they finally represent the original multivariate Cauchy and multivariate $t$.
The other method is direct application of the algorithms to the original targets. 
In other words, \ctn{Dey15} also apply RWM and TMCMC directly to multivariate Cauchy and multivariate $t$,
without resorting to diffeomorphism, and compare their performances. However, they also note that there are
substantial gains with respect to mixing properties in the diffeomorphism based approach.

In either case, \ctn{Dey15} demonstrate that additive TMCMC outperforms RWM quite significantly in the case
of the dependent, high-dimensional target densities. They even compare their performances in the case of
$i.i.d.$ Cauchy and $t$ (with $\nu=10$ degrees of freedom) distributions and reach the same conclusions.

\section{Comparison of additive TMCMC and RWM in the case of a real, spatial data set}
\label{sec:comparison_real}

We now compare additive TMCMC and RWM with respect to a real, spatial dataset on radionuclide count 
data on Rongelap Island, analysed by \ctn{Diggle98} using
a Bayesian hierarchical spatial model. This dataset and the model has been used subsequently by 
\ctn{Chris06} and \ctn{Dutta11}, to evaluate performances of Metropolis-Hastings and TMCMC algorithms, respectively.

\subsection{Model and prior specification}
\label{subsec:model_prior}

For $i=1,\ldots,157$, 
\ctn{Diggle98} model the radionuclide count data as
$$Y_i\sim Poisson(M_i),$$ where $$M_i=t_i\exp\{\beta+S(\bx_i)\};$$ $t_i$ is the duration of observation
at location $\bx_i$, $\beta$ is an unknown parameter and $S(\cdot)$ is a zero-mean Gaussian process
with isotropic covariance function of the form 
$$Cov\left(S(\bz_1),S(\bz_2)\right)=\sigma^2\exp\{-\left(\alpha\parallel\bz_1-\bz_2\parallel\right)^{\delta}\}$$
for any two locations $\bz_1,\bz_2$. In the above, $\parallel\cdot\parallel$ denotes
the Euclidean distance between two locations, and $(\sigma^2, \alpha, \delta)$ are unknown parameters.
Following \ctn{Chris06} we set $\delta=1$, and assume uniform priors on the entire parameter space corresponding to
$(\beta, \log(\sigma^2), \log(\alpha))$. 
Thus, there are $160$ parameters to be updated in each iteration of additive TMCMC and RWM.

\subsection{Optimal scaling}
\label{subsec:optimal_scaling}

Note that the likelihood times the prior in this case can be approximately expressed as (\ref{eq:dep3}), that is,
although the Poisson likelihood is expressible in the form $\exp(-\Psi^d(x^d))$, the Gaussian process prior
for $S(\cdot)$ does not of course admit the form $\prod_{i=1}^{d} \frac{1}{\lambda_i} \phi \left(\frac{x_i}{\lambda_i}\right)$
because of its dependence structure. Hence, this is an instance of a target density which does not fall within
the class of densities for which optimal scaling theory has been developed. As is recommended in general cases,
one may attempt tuning the parameters to approximately achieve the optimal acceptance rate.
But this is a difficult task because of the large dimensionality, as already discussed. A far more important
cause for concern is that, even if one succeeds in approximating the optimal acceptance rate, the corresponding
scales will generally still be sub-optimal, because of the existence of many solutions such that the
optimal acceptance rate holds. 
Indeed, we devise a method for approximately obtaining the optimal acceptance rates, but show that
the corresponding scales lead to really poor performance of RWM, while thanks to the robustness property
of additive TMCMC, the latter yields much reasonable performance. Details follow.

\subsubsection{Pilot TMCMC for facilitating approximate optimal scaling}
\label{subsubsec:pilot_run}
We first consider a pilot TMCMC run consisting of $11\times 10^6$ iterations, with the same TMCMC algorithm
used by \ctn{Dutta11}. In other words, for the pilot run we draw $\epsilon\sim N(0,1)\mathbb I(\epsilon>0)$, 
and consider the following additive transformations: 
\begin{align}
T(\beta,\epsilon)&=\beta\pm 2\epsilon,\notag\\
T(\log(\sigma^2),\epsilon)&=\log(\sigma^2)\pm 5\epsilon,\notag\\
T(\log(\alpha),\epsilon)&=\log(\alpha)\pm 5\epsilon,\notag\\
T(S(\bx_i),\epsilon)&=S(\bx_i)\pm 2\epsilon; \ \ \mbox{for} \ \ i=1,\ldots,157,\notag
\end{align}
where $``+"$ and $``-"$ occur with probability $1/2$ each.

After discarding the first $10^6$ iterations as burn-in, we then store one TMCMC sample in every $100$ iterations
to yield $10^5$ thinned TMCMC realizations. With these stored realizations, we then obtain the empirical variance-covariance
matrix of the $160$ unknowns and store the $160$ eigenvalues of the matrix. 

\subsubsection{Approximately optimal acceptance rates of additive TMCMC and RWM using the stored eigenvalues}
\label{subsubsec:optimal_eigenvalues}

For $i=1,\ldots,160$, letting $\theta_i$ denote the unknowns, and $\lambda_i$, 
the corresponding eigenvalues, we then consider the
following additive transformations for TMCMC: 
$$T(\theta_i,\epsilon)=\theta_i\pm c_{TMCMC}\sqrt{\frac{2\lambda_i\ell^2_{opt,TMCMC}}{d}}\epsilon,$$
where $\epsilon\sim N(0,1)\mathbb I(\epsilon>0)$, $``+"$ and $``-"$ occur with probability $1/2$ each,
$\ell_{opt,TMCMC}=1.715$ (see (\ref{eq:l5})), and $c_{TMCMC}$ is a tuning parameter for adjusting the acceptance rate to about
$44\%$. It turned out, after setting $c_{TMCMC}=0.95$, that the empirical acceptance rate obtained from a TMCMC
run of length $1.01\times 10^8$, after discarding the first $1.7\times 10^7$ iterations as burn-in, 
is very accurately approximated as $0.439$. 
Implementing this TMCMC algorithm, we store one in every $100$ realizations
after the burn-in to obtain $8.5\times 10^5$ thinned additive TMCMC realizations from the posterior distribution.
The total time for the implementation took $46$ hours and $51$ minutes on a 64 bit machine
with CPU MHz $1600$ and about $8$ GB memory.

For RWM, we consider the following proposal:
$$T(\theta_i,\epsilon_i)=\theta_i+c_{RWM}\sqrt{\frac{2\lambda_i\ell^2_{opt,RWM}}{d}}\epsilon_i,$$
where $\epsilon\stackrel{i.i.d.}{\sim} N(0,1)$, $\ell_{opt,RWM}=1.715$, and $c_{RWM}=0.95=c_{TMCMC}$. We obtain the empirical
acceptance rate from a RWM run of length $1.01\times 10^8$, after discarding the first $1.7\times 10^7$ iterations as burn-in,
as approximately $0.228$.  
We implement this RWM algorithm, storing 
one in every $100$ realizations after the burn-in to obtain $8.5\times 10^5$ thinned RWM realizations from the posterior.
The total time for the implementation took $46$ hours and $53$ minutes
on the same machine on which TMCMC was implemented.

\subsection{Results of comparison}
\label{subsection:results_comparison}

Figures \ref{fig:tmcmc_rongelap} and \ref{fig:rwm_rongelap} show the trace plots
of the stored realizations obtained by TMCMC and RWM, respectively, after a further thinning
of size $300$. 
Such substantial further thinning is required to facilitate effortless visual comparison of the autocorrelation
plots for TMCMC and RWM shown in Figure \ref{fig:acf_rongelap}.
In Figures \ref{fig:tmcmc_rongelap} and \ref{fig:rwm_rongelap} it is worth observing that, RWM, composed of
$160$ $\epsilon_i$'s in this example, is prone to require a large number of iterations to return to any given set with 
positive posterior probability, once it leaves it. 
On the other hand, TMCMC marches off to convergence much faster
than RWM, exploiting its more localized move types thanks to a single $\epsilon$. 
This insight is more formalized by the comparison of the associated autocorrelation plots shown in
Figure \ref{fig:acf_rongelap}. It is obvious that TMCMC significantly outperforms RWM in terms of 
autocorrelations in all the cases. 

Since scaling is directly related to autocorrelation (see Sections S-3 and S-4 of the supplement),  
it is clear that poor scaling of RWM in comparison with additive TMCMC is the reason for the relatively poor
performance of the former. 
Indeed, even though we could approximately achieve the desired acceptance rates, 
there are many solutions for the scales, given the same acceptance rate; 
in fact, selecting reasonable scales gets increasing difficult with increasing dimensions. Thus,
it is highly unlikely that the chosen scales are even reasonable in this high-dimensional example, 
for either TMCMC or RWM. As a result it makes sense to conclude with respect to the autocorrelations that, 
in this real data study, sensitivity of RWM with respect to optimal scales is the reason 
for its relatively poor performance, while robustness of TMCMC in this regard is the reason for its 
quite reasonable performance.

Thus, in this real data example, additive TMCMC very clearly and very convincingly outperforms RWM.

\begin{figure}
\centering
\subfigure [Trace plot of $\alpha$.]{ \label{fig:alpha_tmcmc}
\includegraphics[width=4.5cm,height=4.5cm]{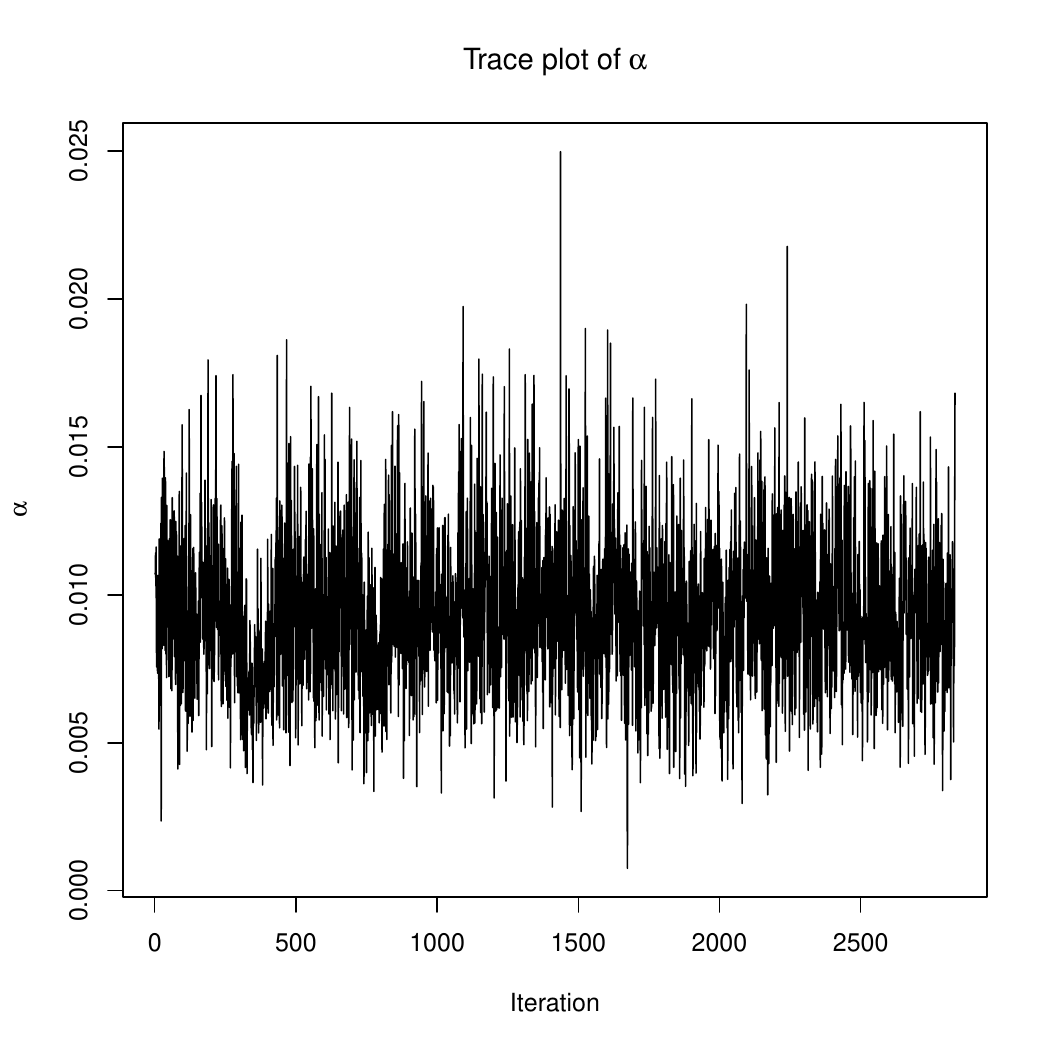}}
\hspace{2mm}
\subfigure [Trace plot of $\beta$.]{ \label{fig:beta_tmcmc}
\includegraphics[width=4.5cm,height=4.5cm]{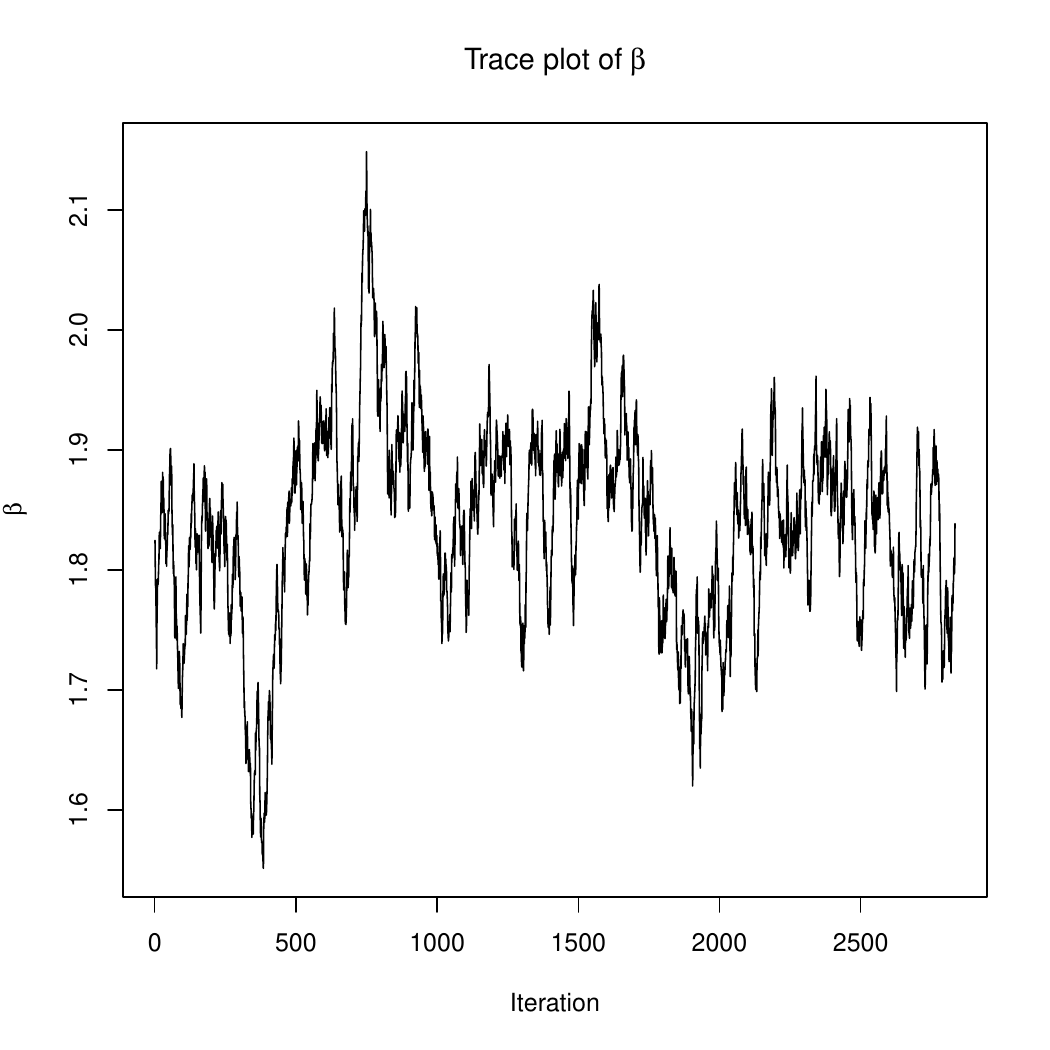}}
\hspace{2mm}
\subfigure [Trace plot of $\sigma^2$.]{ \label{fig:sigmasq_tmcmc}
\includegraphics[width=4.5cm,height=4.5cm]{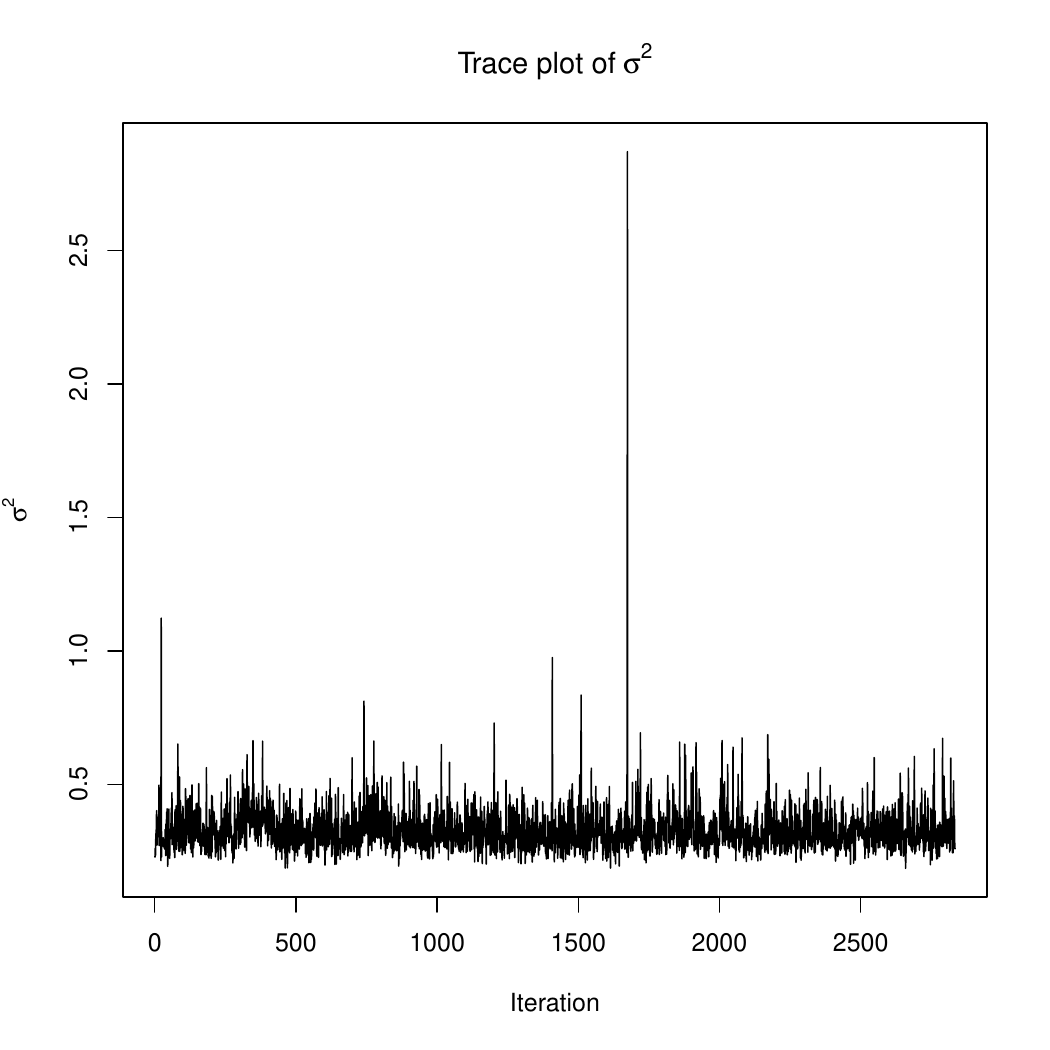}}\\
\vspace{2mm}
\subfigure [Trace plot of $S(\bx_1)$.]{ \label{fig:S1_tmcmc}
\includegraphics[width=4.5cm,height=4.5cm]{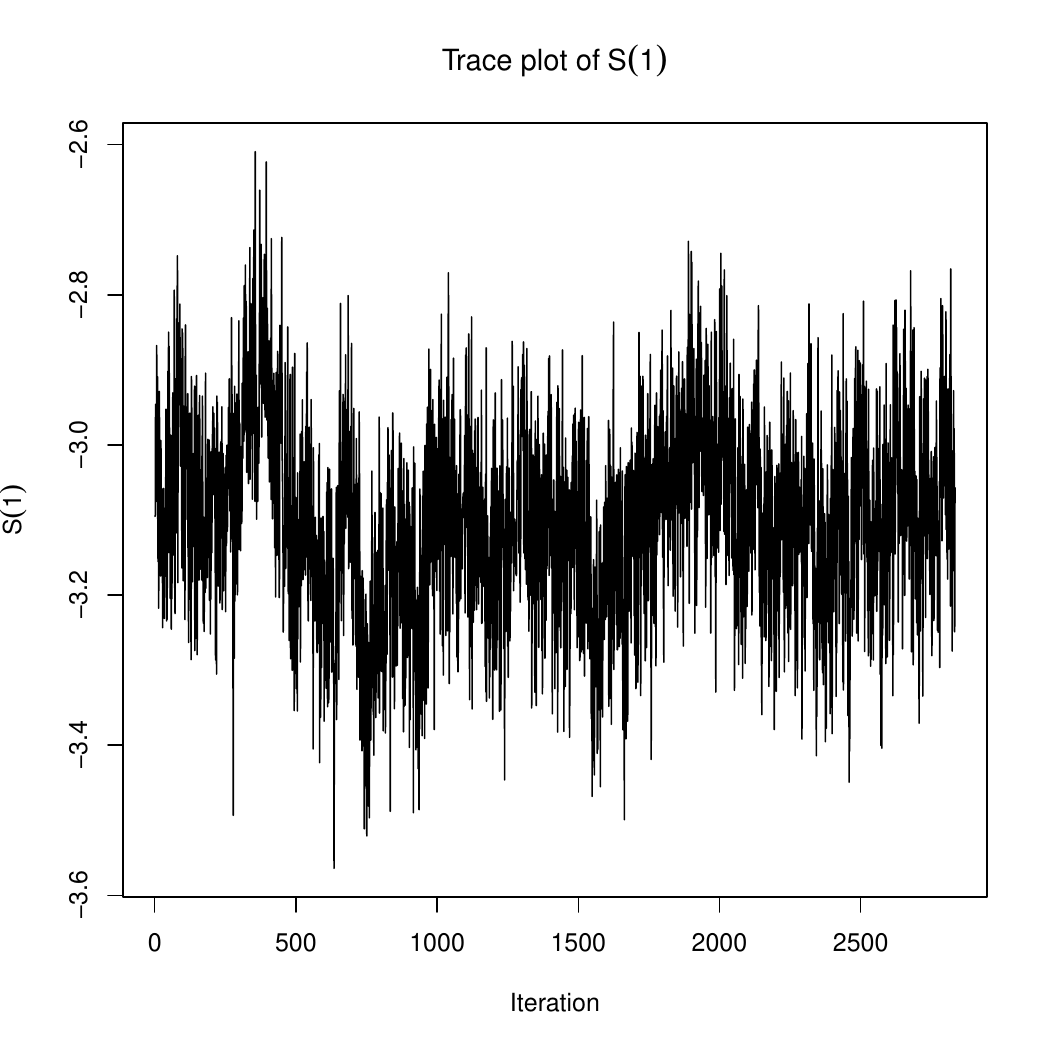}}
\hspace{2mm}
\subfigure [Trace plot of $S(\bx_{10})$.]{ \label{fig:S10_tmcmc}
\includegraphics[width=4.5cm,height=4.5cm]{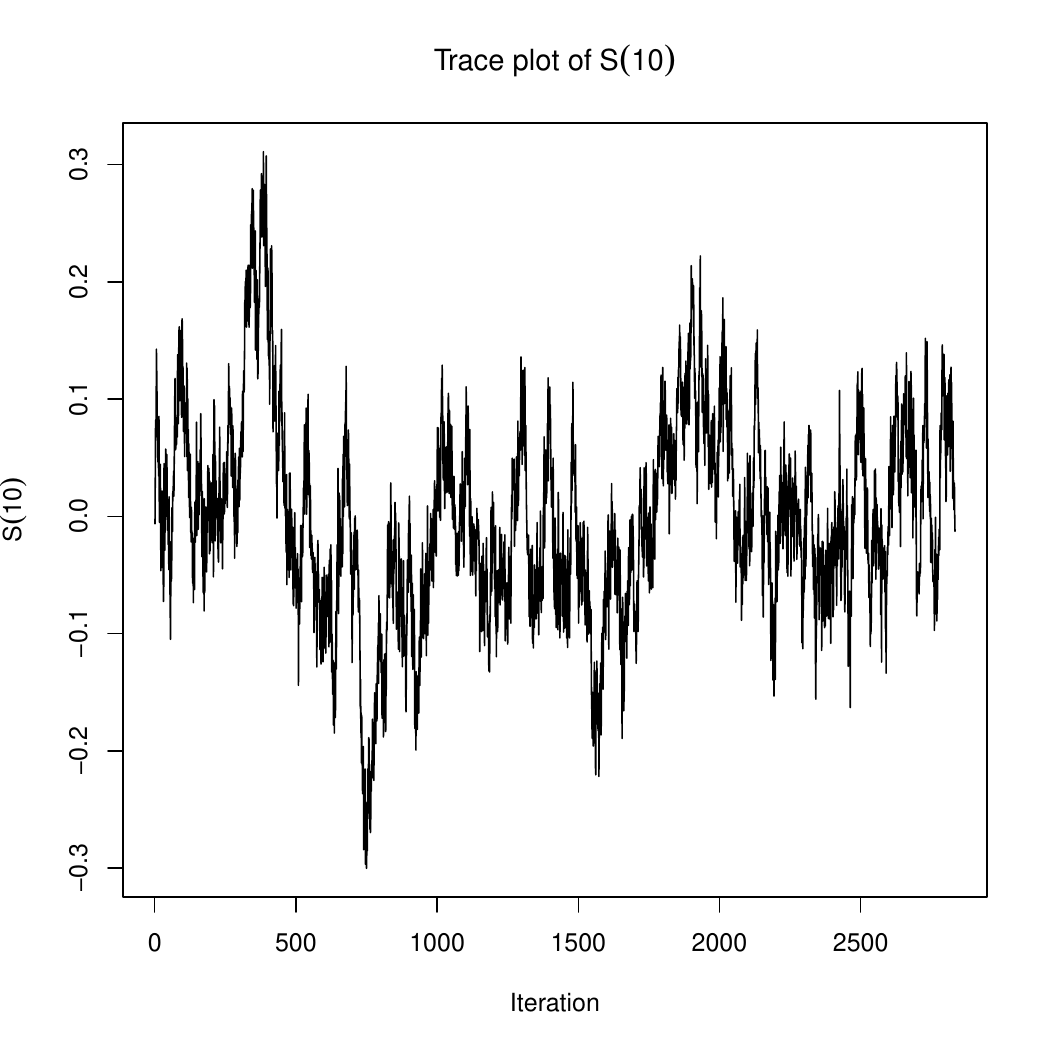}}
\hspace{2mm}
\subfigure [Trace plot of $S(\bx_{50})$.]{ \label{fig:S50_tmcmc}
\includegraphics[width=4.5cm,height=4.5cm]{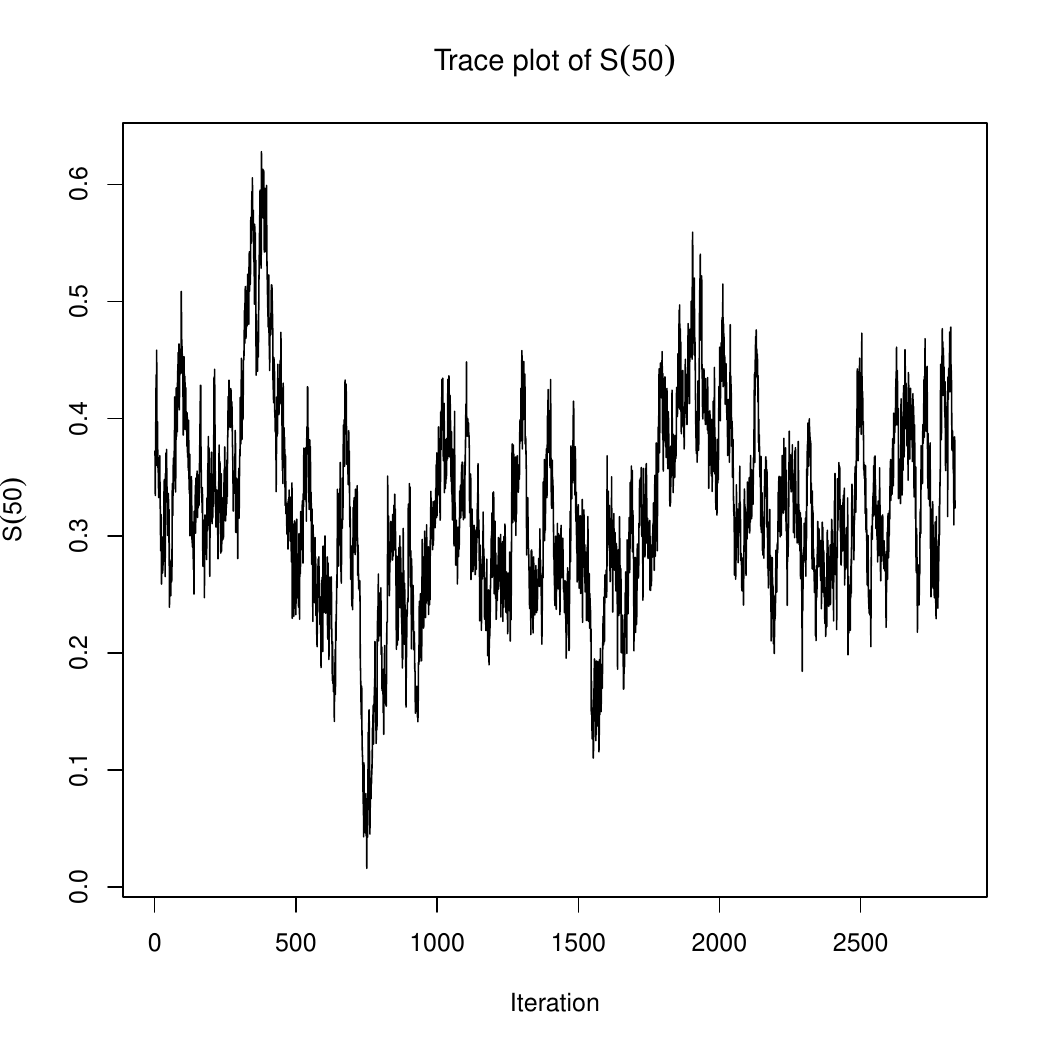}}\\
\vspace{2mm}
\subfigure [Trace plot of $S(\bx_{100})$.]{ \label{fig:S100_tmcmc}
\includegraphics[width=4.5cm,height=4.5cm]{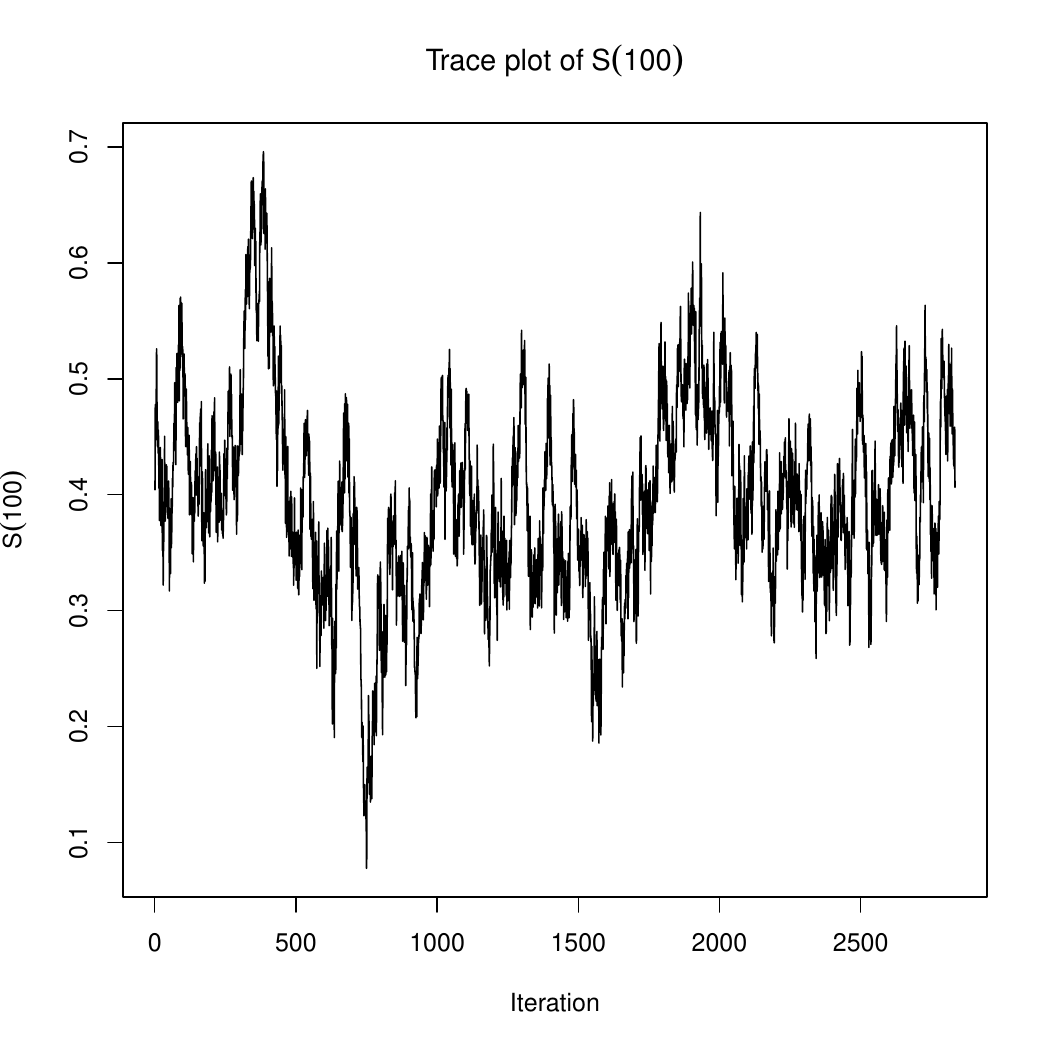}}
\hspace{2mm}
\subfigure [Trace plot of $S(\bx_{150})$.]{ \label{fig:S150_tmcmc}
\includegraphics[width=4.5cm,height=4.5cm]{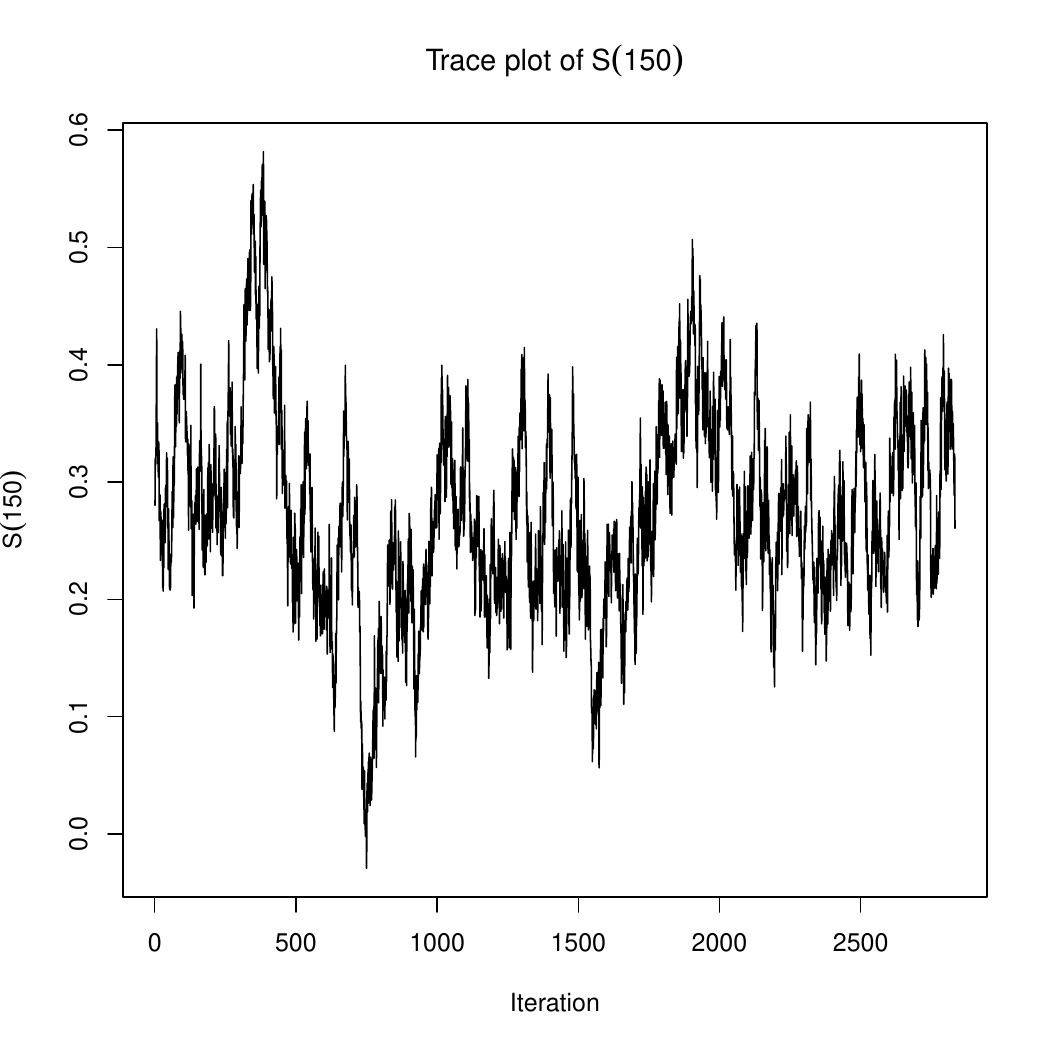}}
\hspace{2mm}
\subfigure [Trace plot of $S(\bx_{157})$.]{ \label{fig:S157_tmcmc}
\includegraphics[width=4.5cm,height=4.5cm]{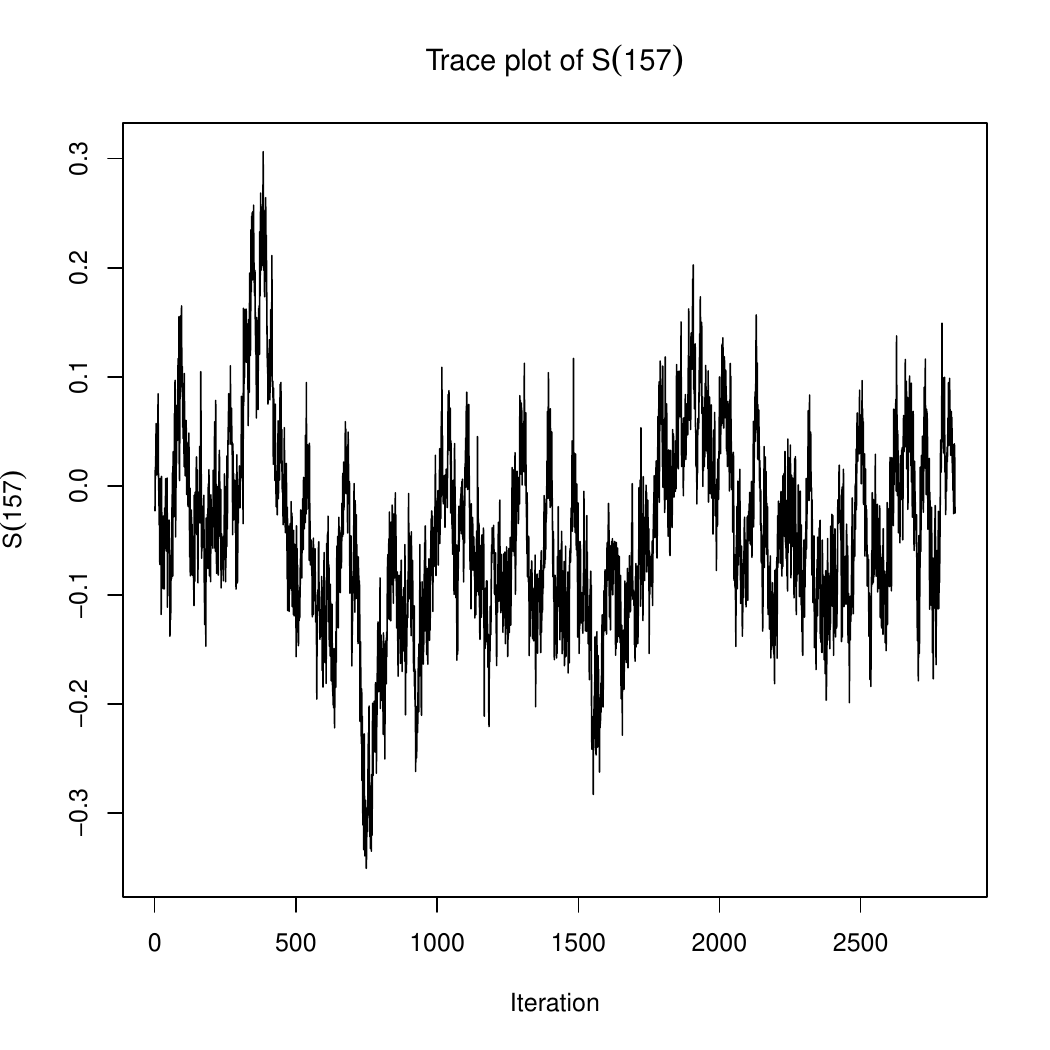}}
\caption{Rongelap island data: TMCMC based trace plots.}
\label{fig:tmcmc_rongelap}
\end{figure}

\begin{figure}
\centering
\subfigure [Trace plot of $\alpha$.]{ \label{fig:alpha_rwm}
\includegraphics[width=4.5cm,height=4.5cm]{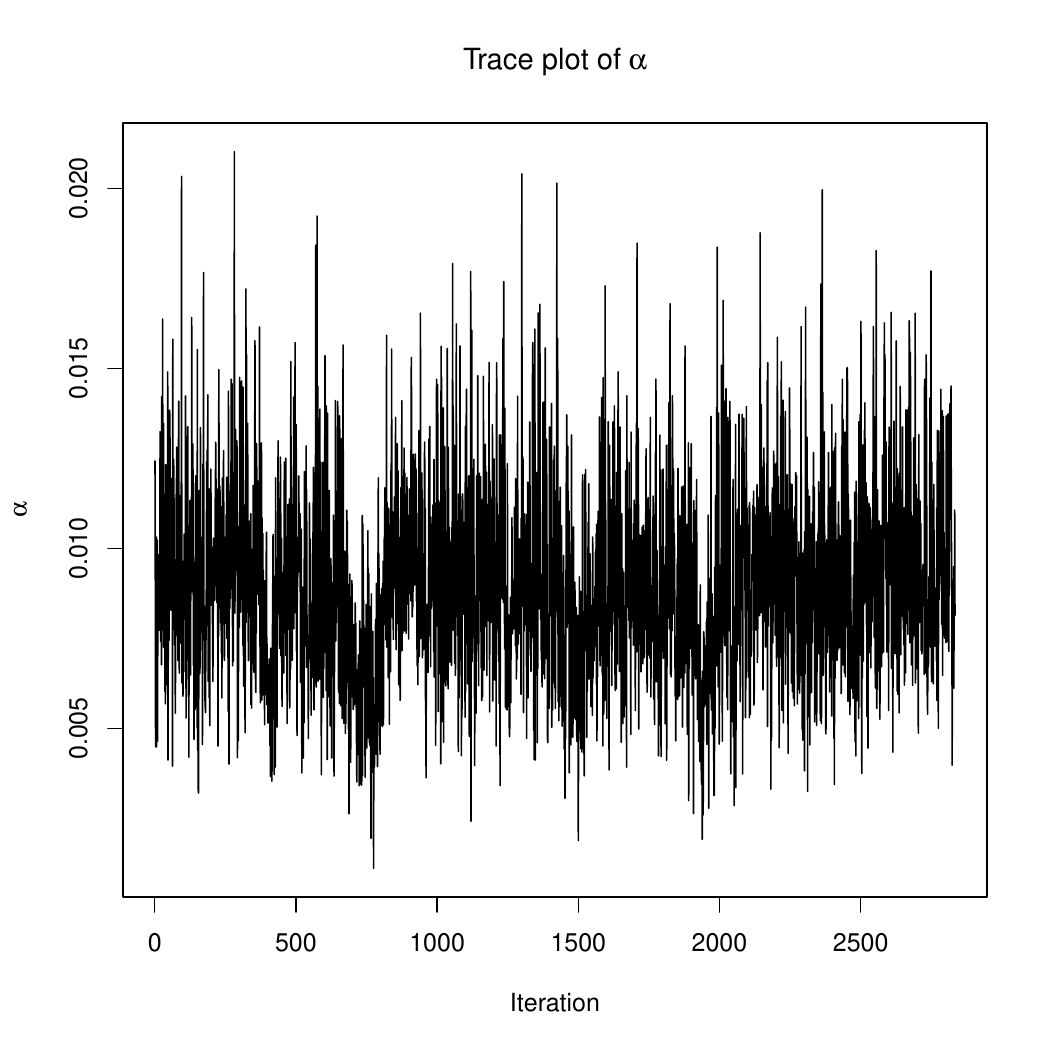}}
\hspace{2mm}
\subfigure [Trace plot of $\beta$.]{ \label{fig:beta_rwm}
\includegraphics[width=4.5cm,height=4.5cm]{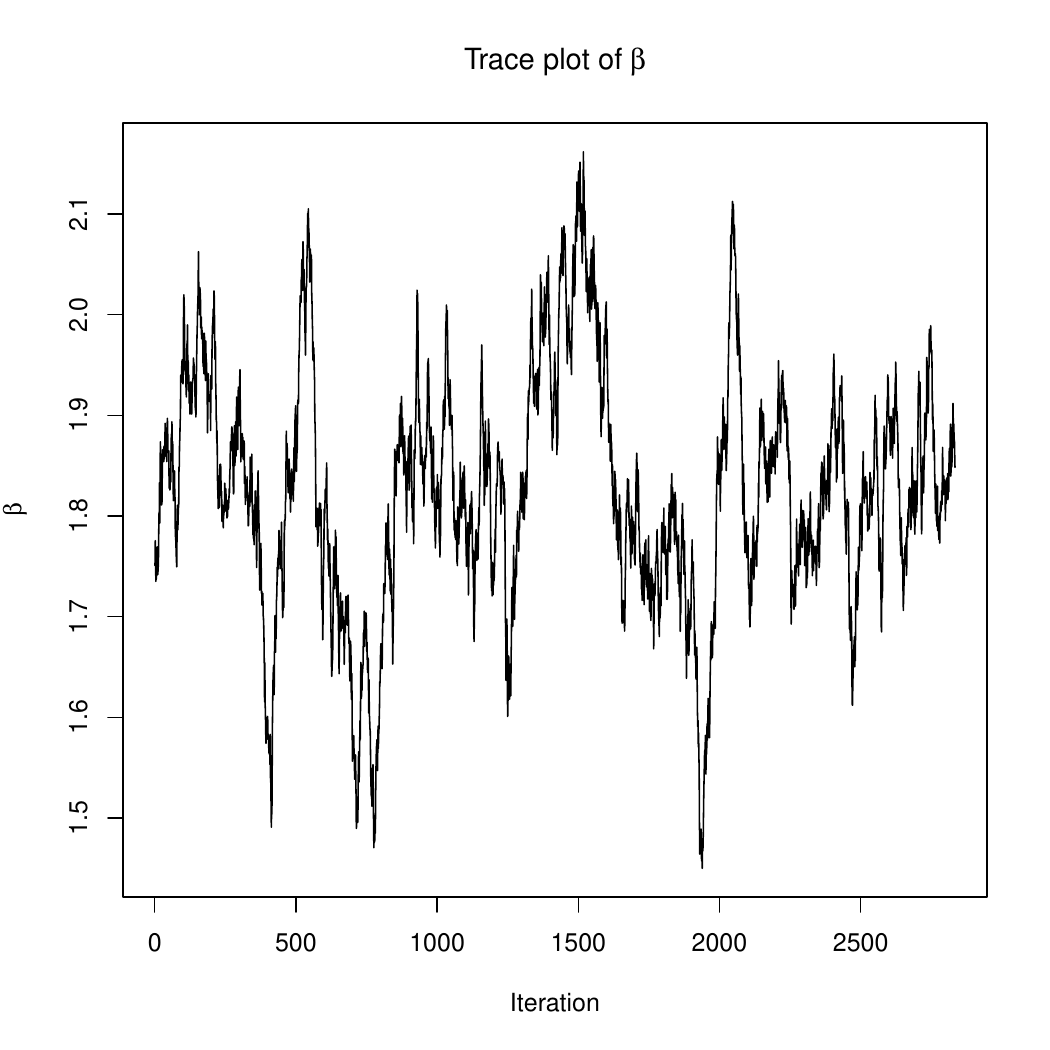}}
\hspace{2mm}
\subfigure [Trace plot of $\sigma^2$.]{ \label{fig:sigmasq_rwm}
\includegraphics[width=4.5cm,height=4.5cm]{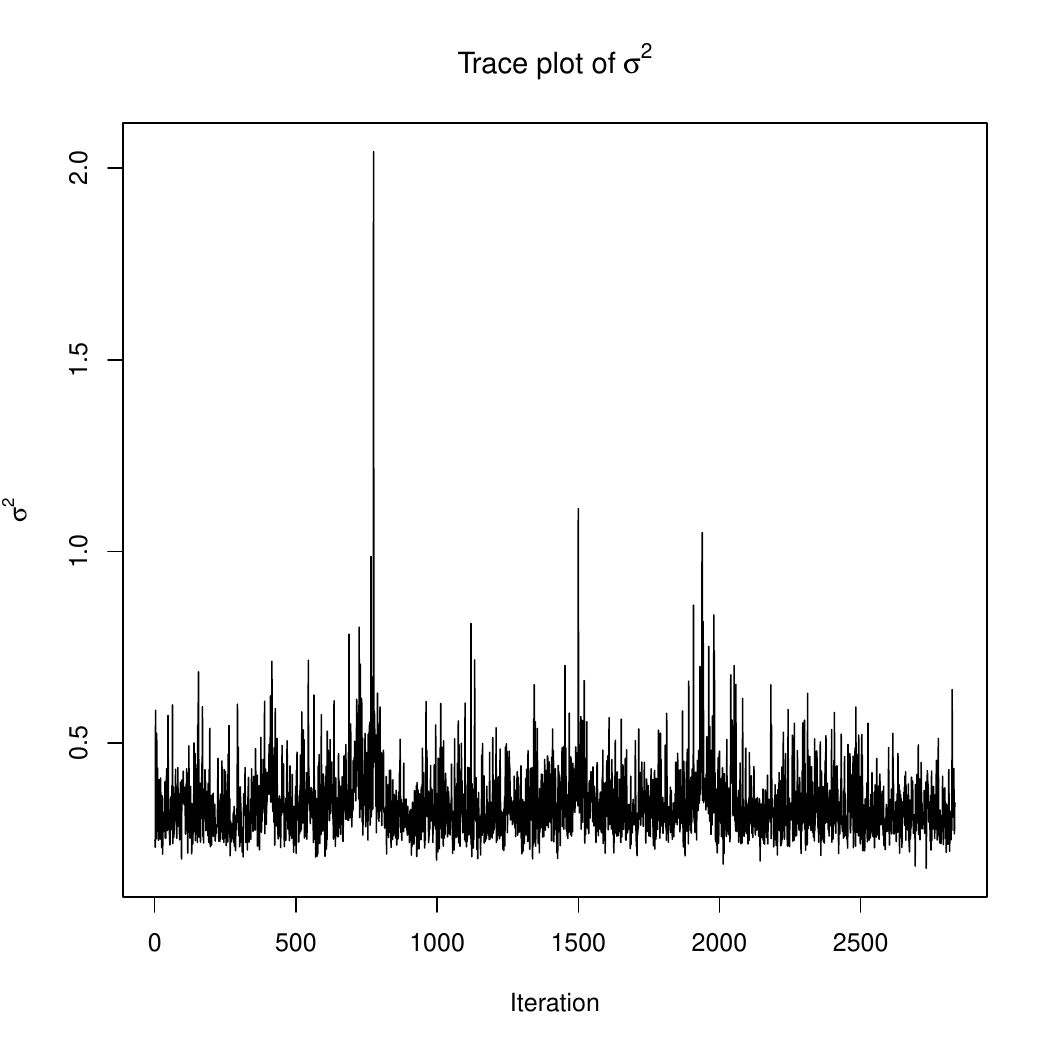}}\\
\vspace{2mm}
\subfigure [Trace plot of $S(\bx_1)$.]{ \label{fig:S1_rwm}
\includegraphics[width=4.5cm,height=4.5cm]{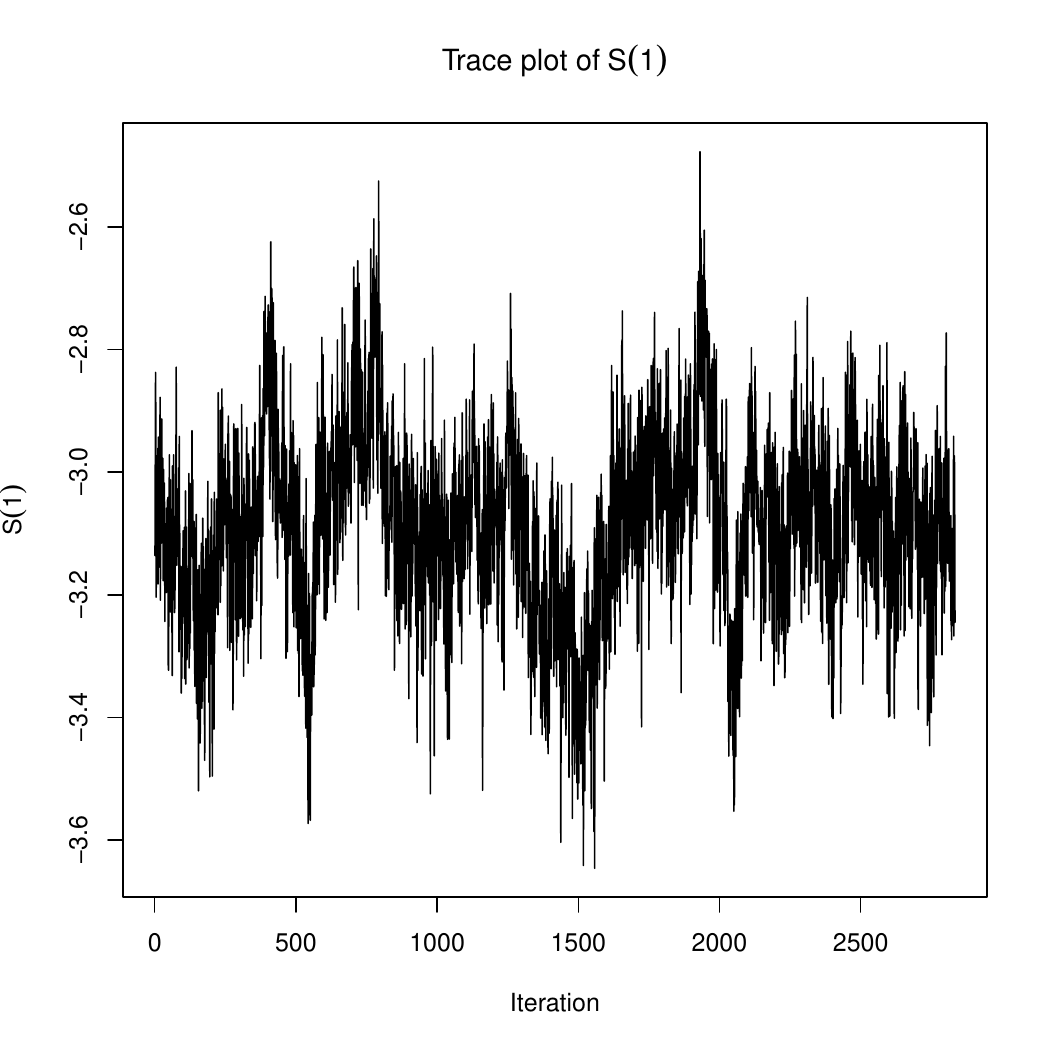}}
\hspace{2mm}
\subfigure [Trace plot of $S(\bx_{10})$.]{ \label{fig:S10_rwm}
\includegraphics[width=4.5cm,height=4.5cm]{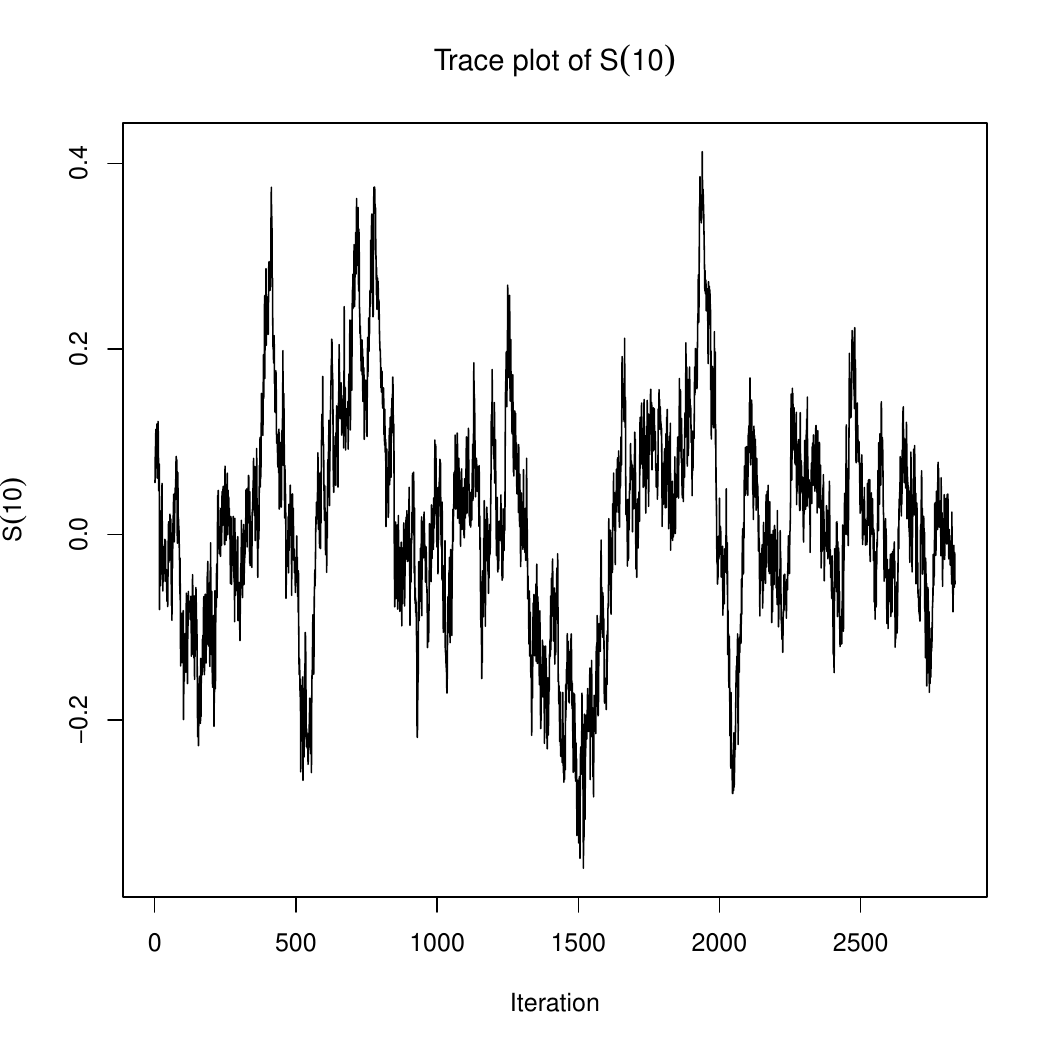}}
\hspace{2mm}
\subfigure [Trace plot of $S(\bx_{50})$.]{ \label{fig:S50_rwm}
\includegraphics[width=4.5cm,height=4.5cm]{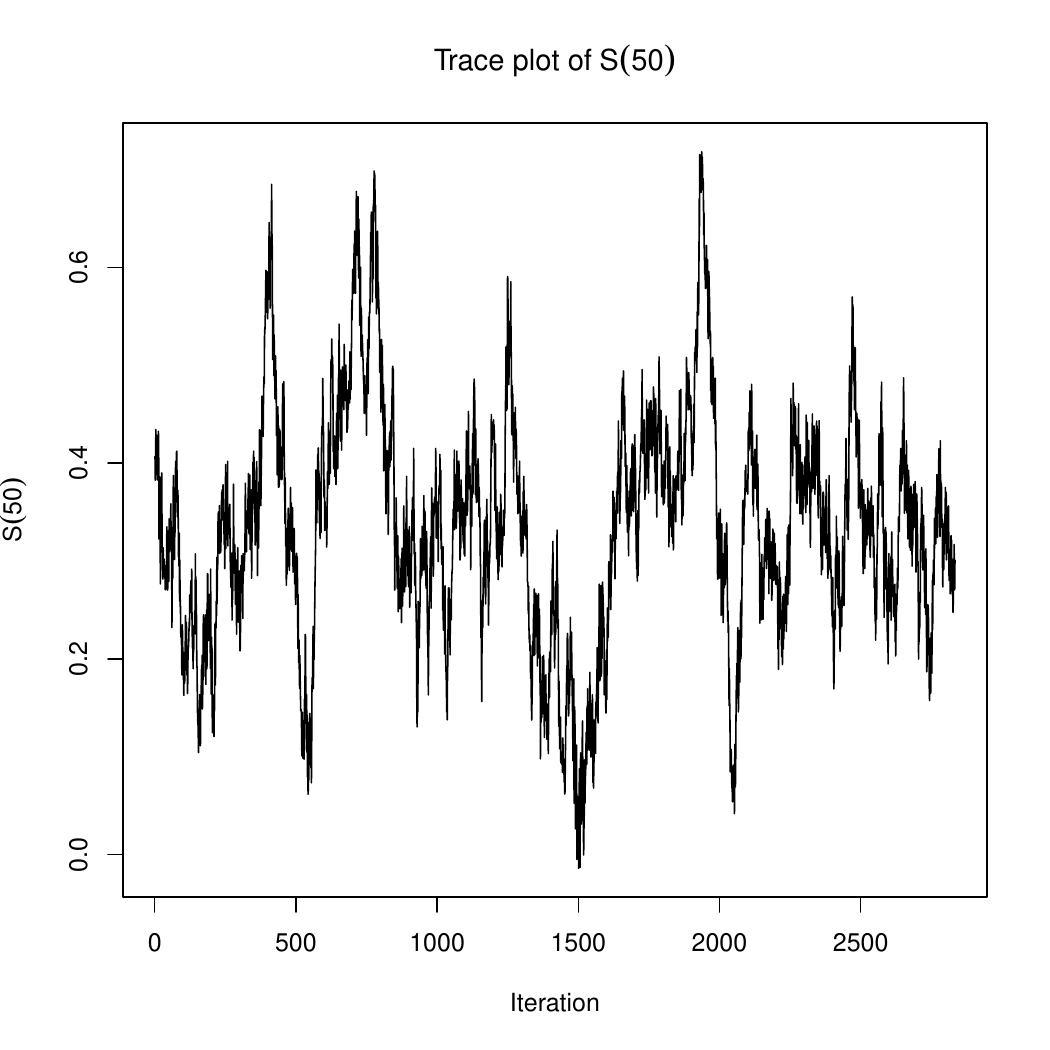}}\\
\vspace{2mm}
\subfigure [Trace plot of $S(\bx_{100})$.]{ \label{fig:S100_rwm}
\includegraphics[width=4.5cm,height=4.5cm]{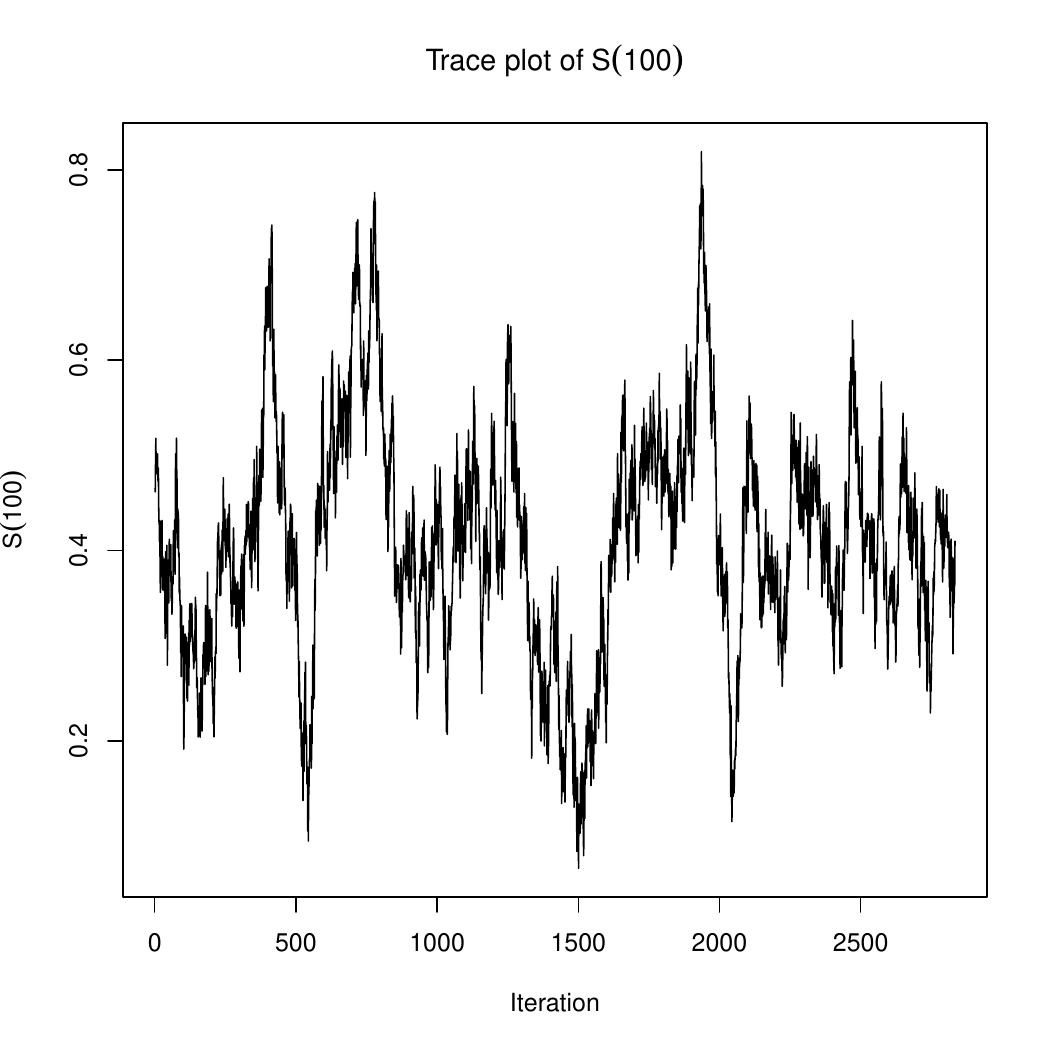}}
\hspace{2mm}
\subfigure [Trace plot of $S(\bx_{150})$.]{ \label{fig:S150_rwm}
\includegraphics[width=4.5cm,height=4.5cm]{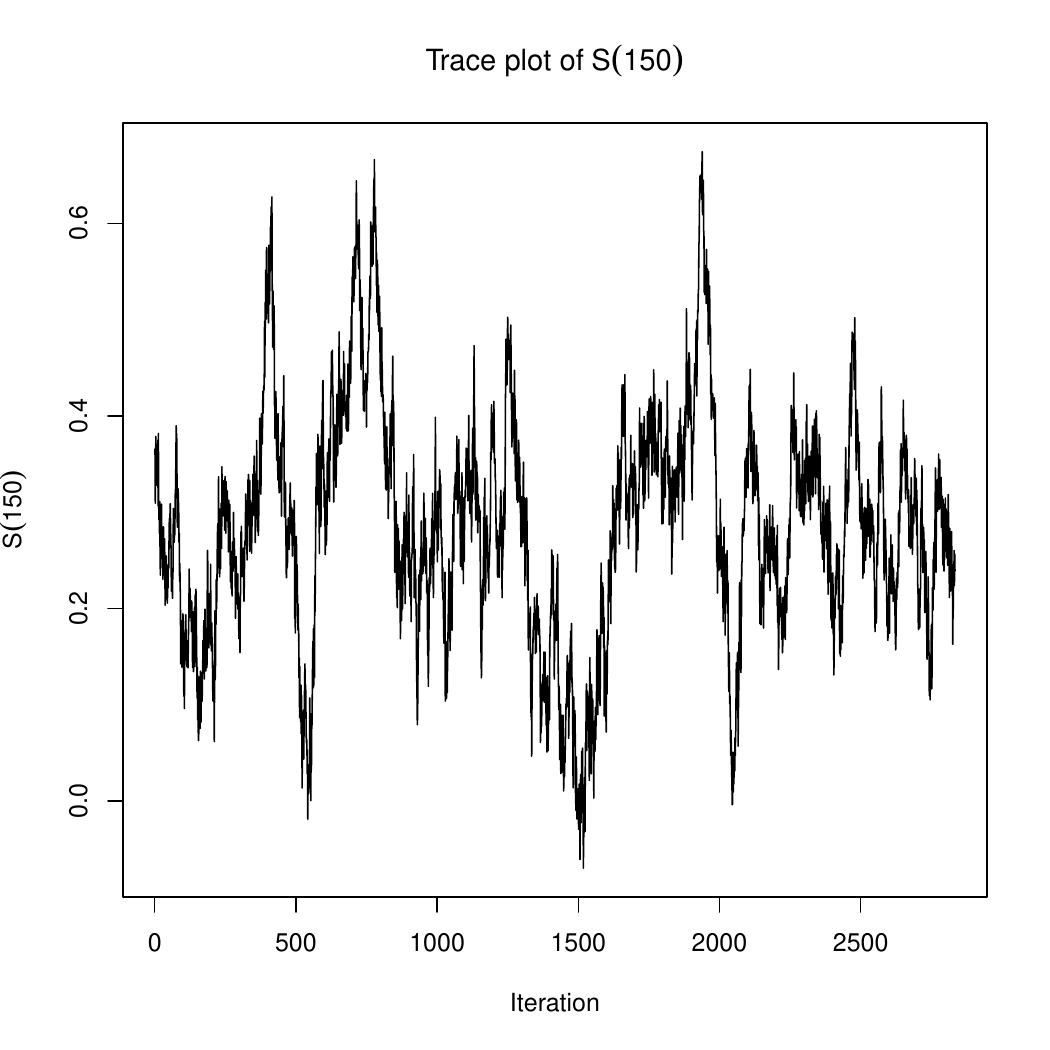}}
\hspace{2mm}
\subfigure [Trace plot of $S(\bx_{157})$.]{ \label{fig:S157_rwm}
\includegraphics[width=4.5cm,height=4.5cm]{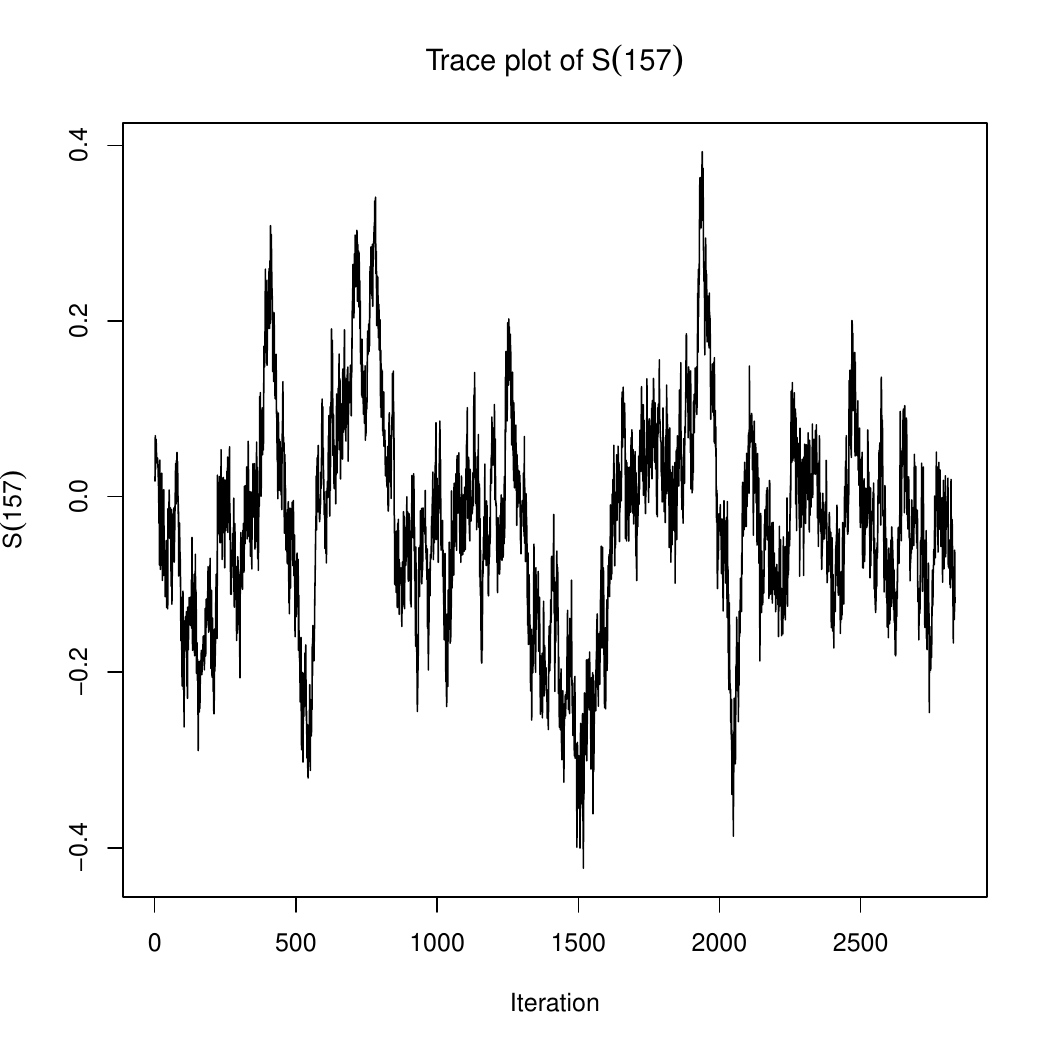}}
\caption{Rongelap island data: RWM based trace plots.}
\label{fig:rwm_rongelap}
\end{figure}

\begin{figure}
\centering
\subfigure [ACF comparison for $\alpha$.]{ \label{fig:acf_alpha}
\includegraphics[width=4.5cm,height=4.5cm]{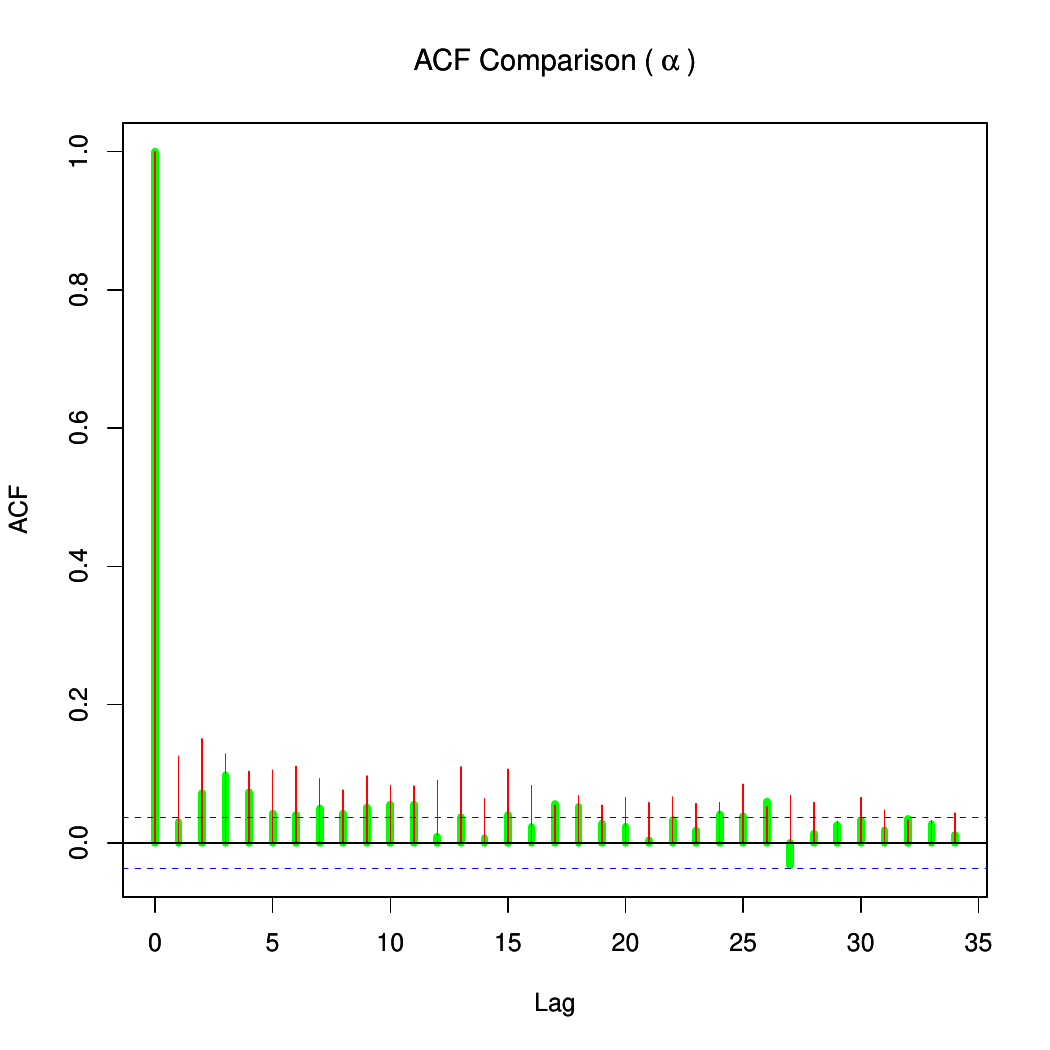}}
\hspace{2mm}
\subfigure [ACF comparison for $\beta$.]{ \label{fig:acf_beta}
\includegraphics[width=4.5cm,height=4.5cm]{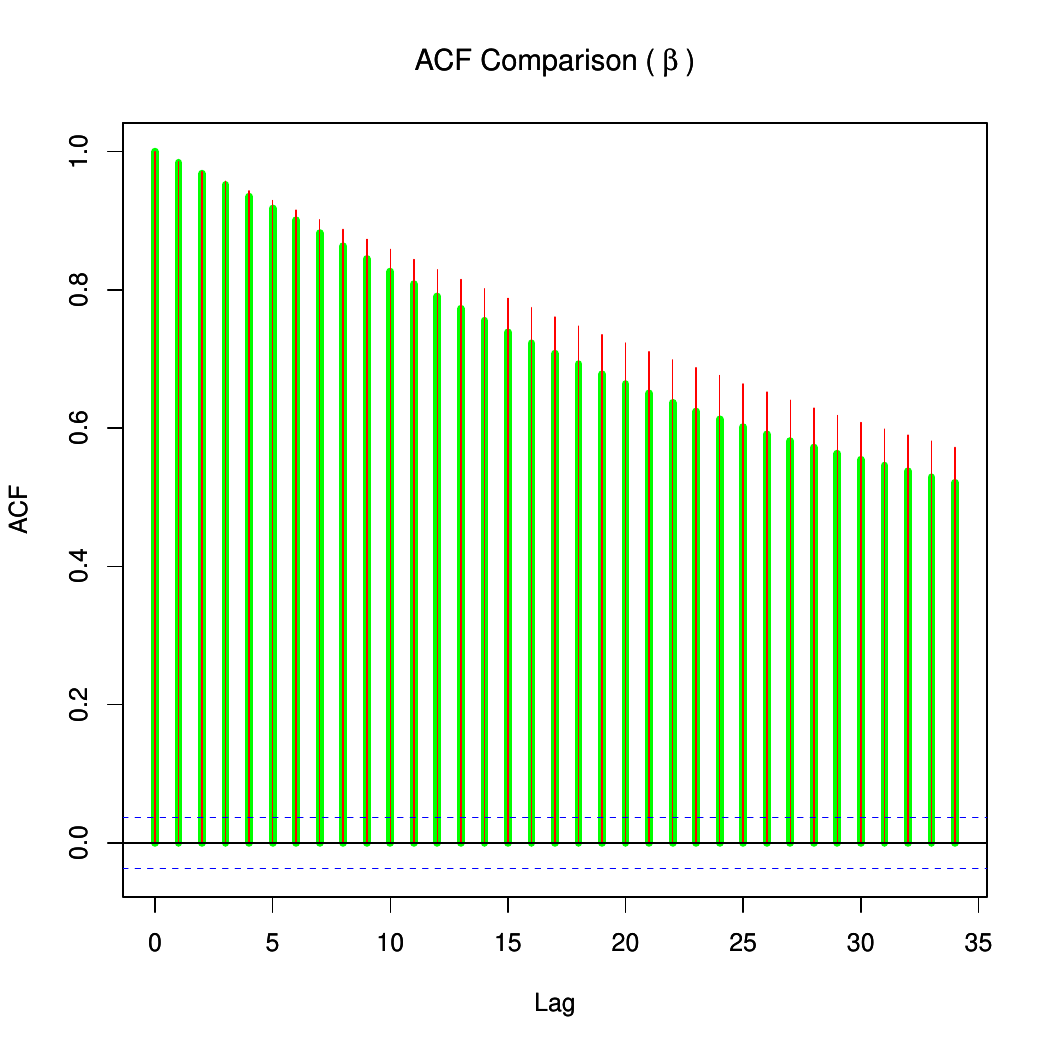}}
\hspace{2mm}
\subfigure [ACF comparison for $\sigma^2$.]{ \label{fig:acf_sigmasq}
\includegraphics[width=4.5cm,height=4.5cm]{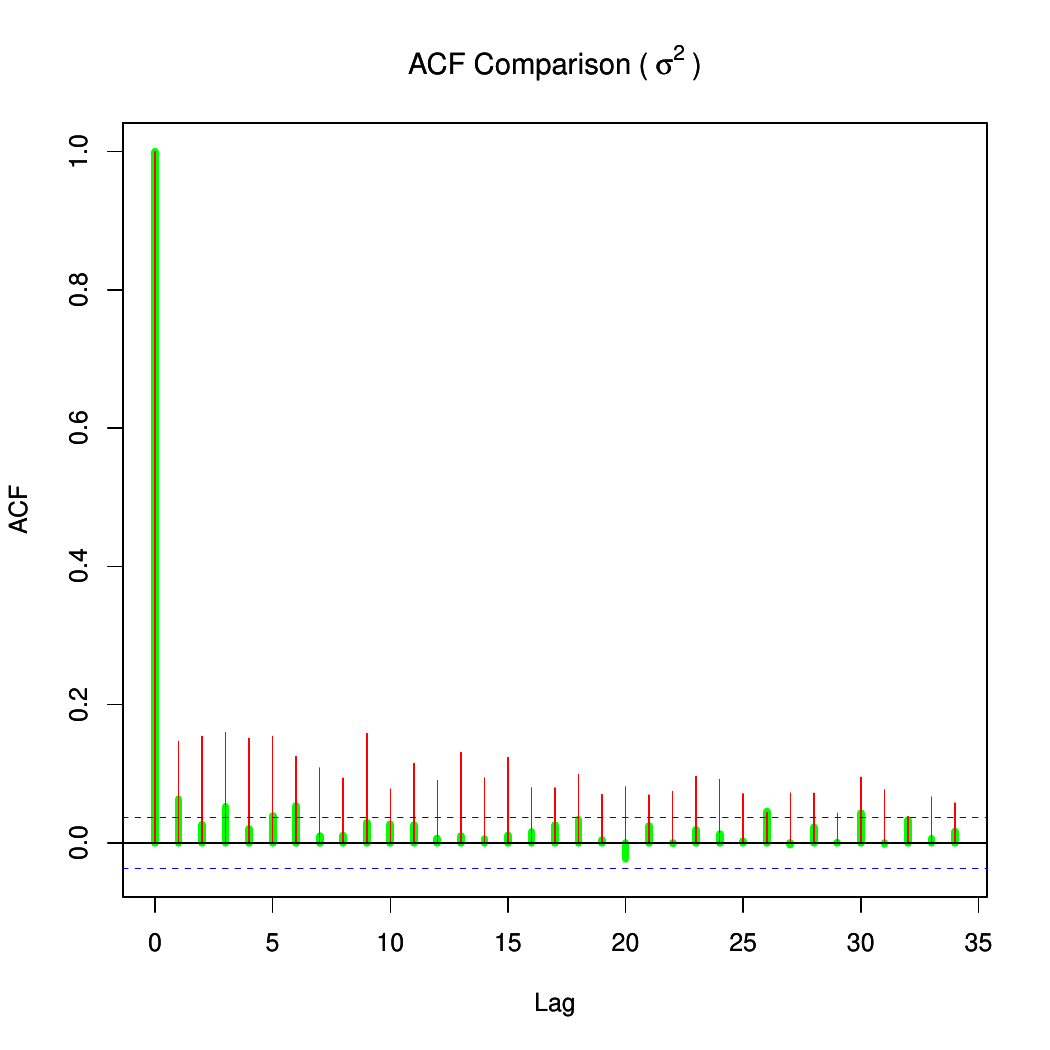}}\\
\vspace{2mm}
\subfigure [ACF comparison for $S(\bx_1)$.]{ \label{fig:acf_S1}
\includegraphics[width=4.5cm,height=4.5cm]{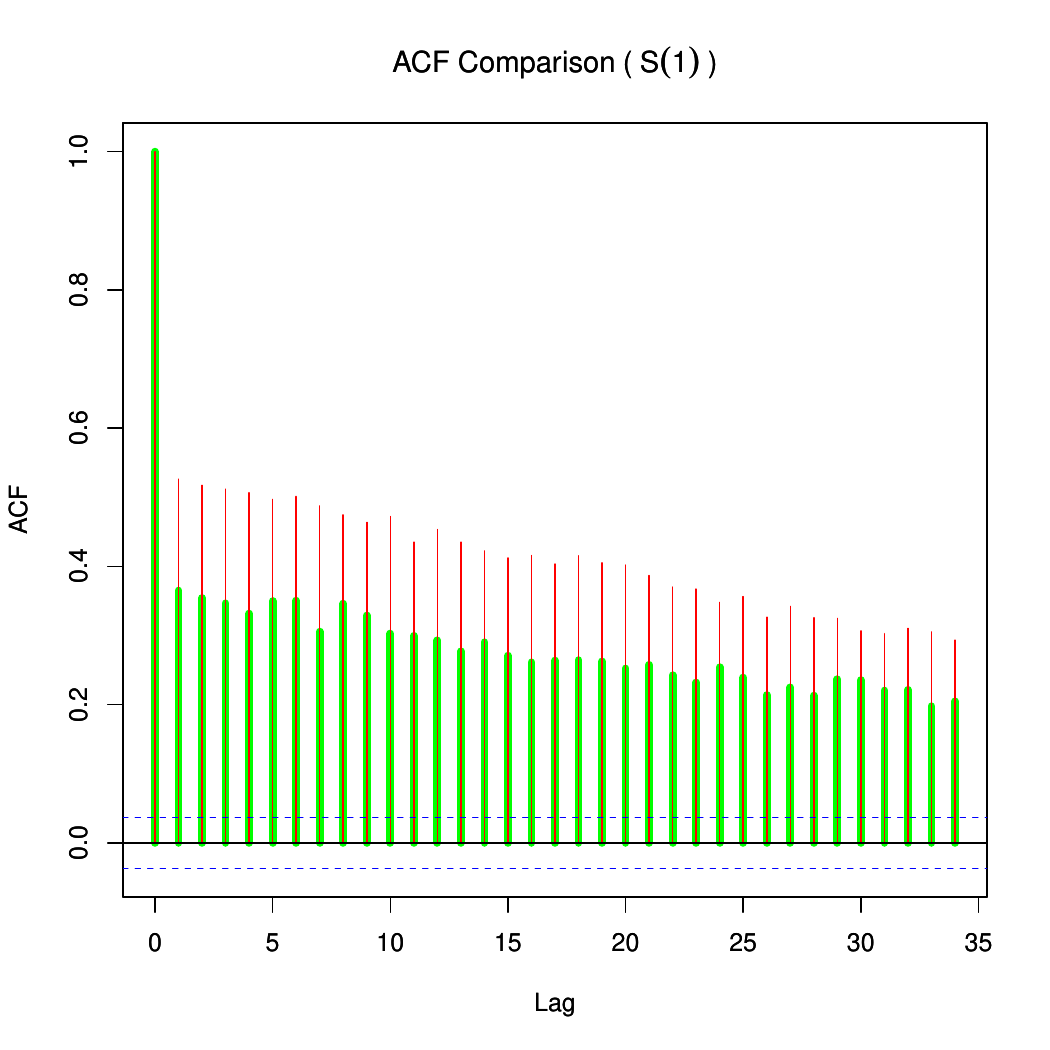}}
\hspace{2mm}
\subfigure [ACF comparison for $S(\bx_{10})$.]{ \label{fig:acf_S10}
\includegraphics[width=4.5cm,height=4.5cm]{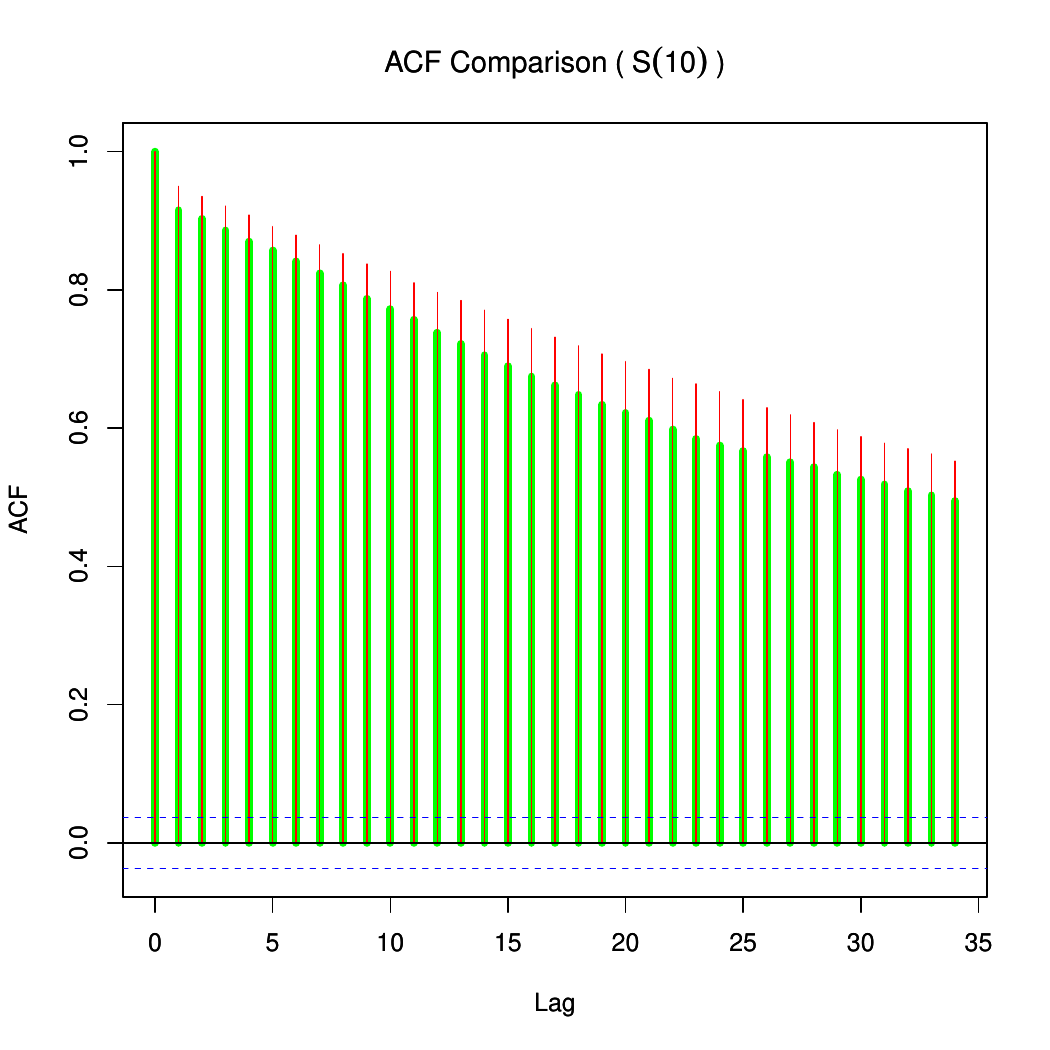}}
\hspace{2mm}
\subfigure [ACF comparison for $S(\bx_{50})$.]{ \label{fig:acf_S50}
\includegraphics[width=4.5cm,height=4.5cm]{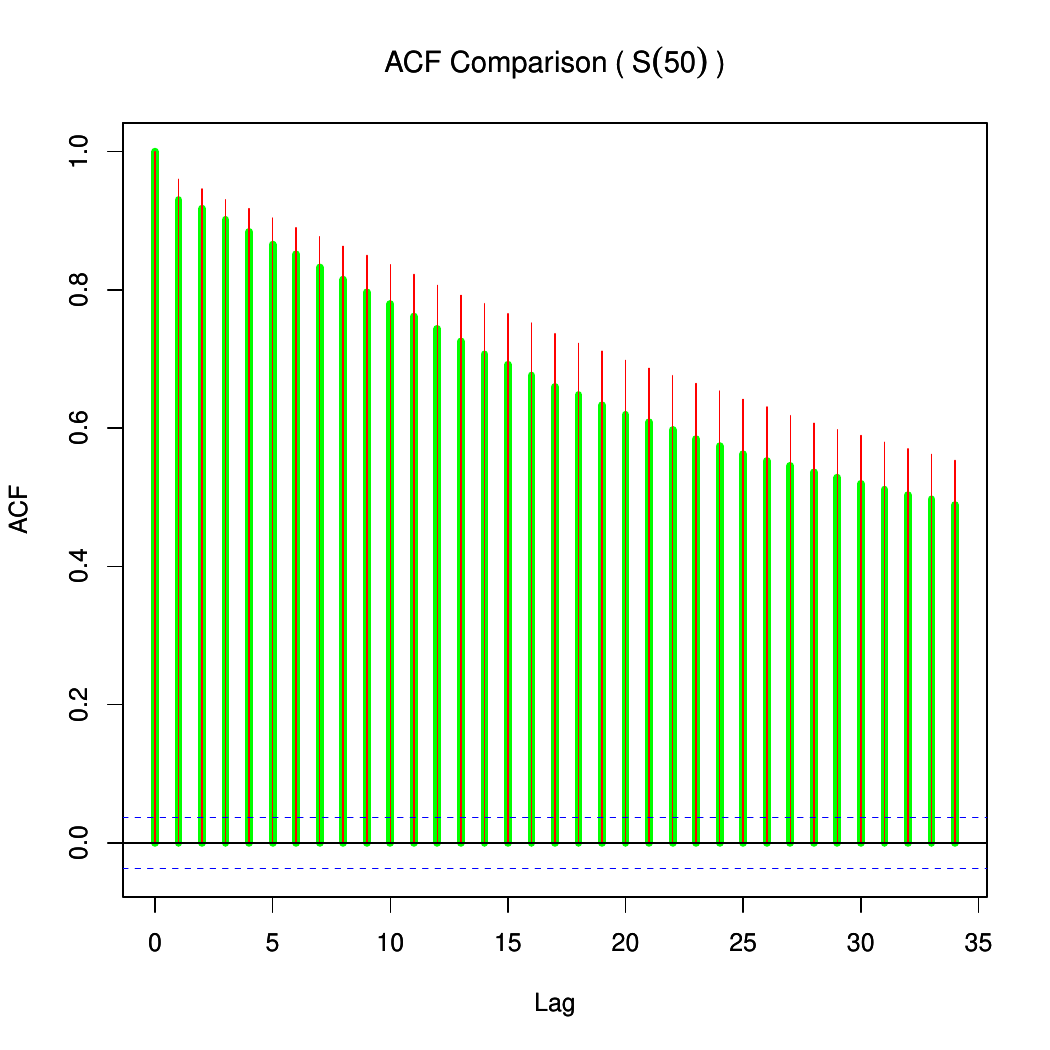}}\\
\vspace{2mm}
\subfigure [ACF comparison for $S(\bx_{100})$.]{ \label{fig:acf_S100}
\includegraphics[width=4.5cm,height=4.5cm]{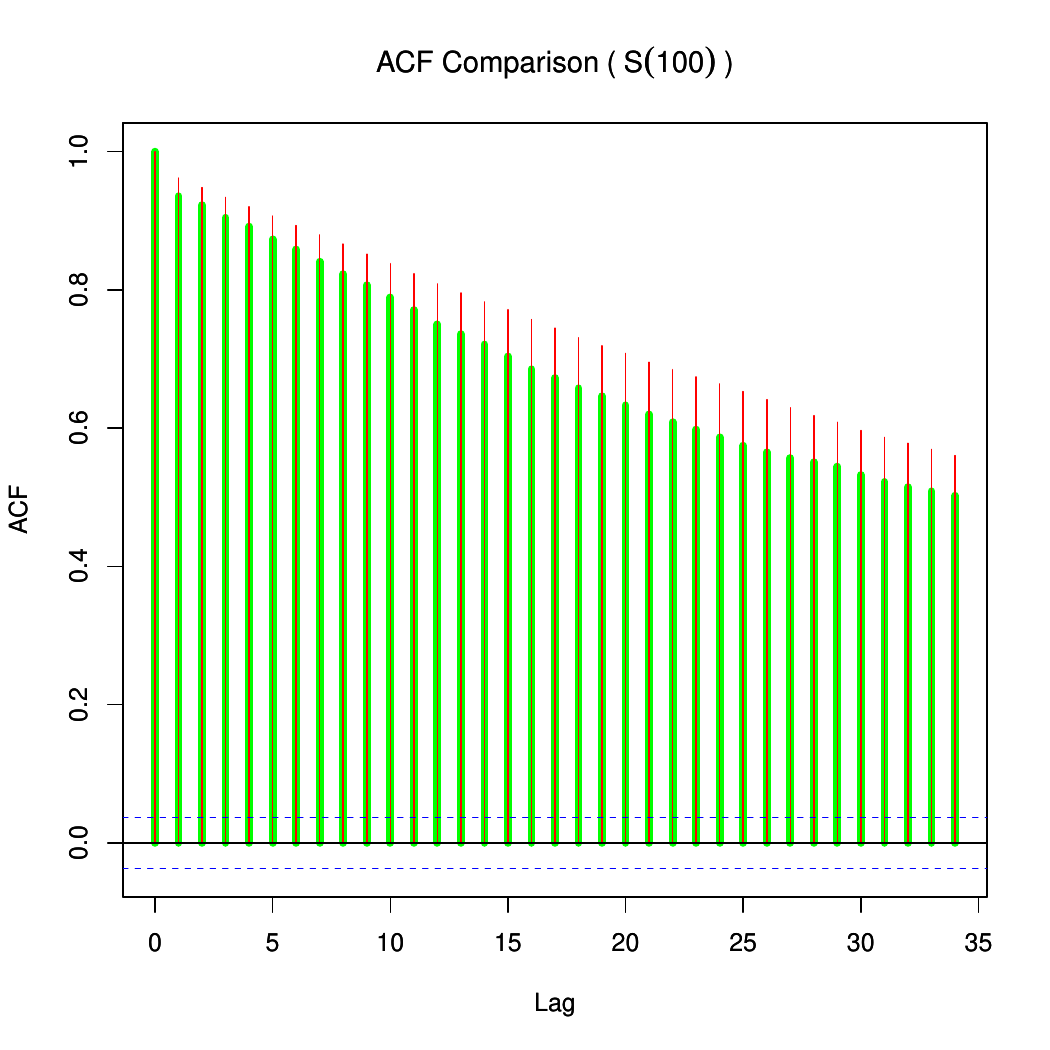}}
\hspace{2mm}
\subfigure [ACF comparison for $S(\bx_{150})$.]{ \label{fig:acf_S150}
\includegraphics[width=4.5cm,height=4.5cm]{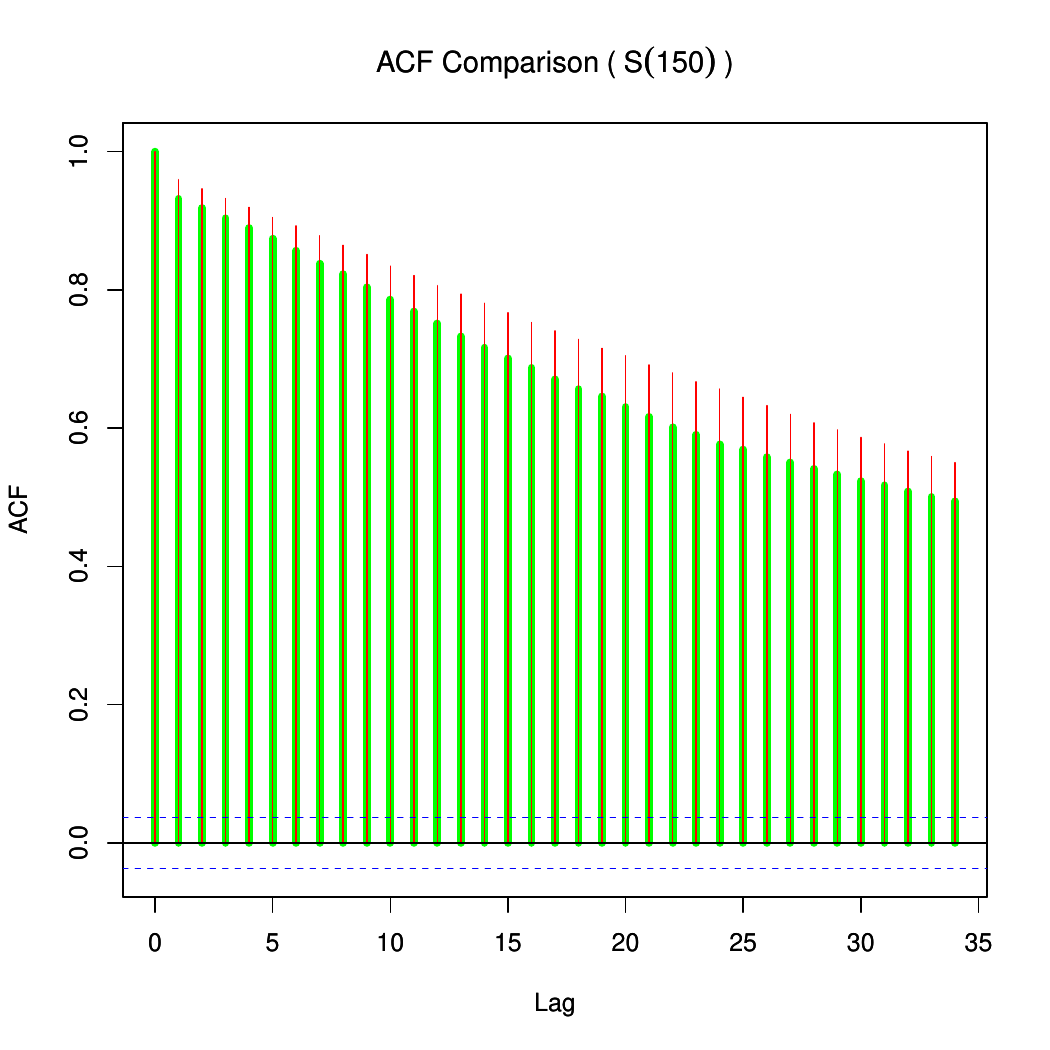}}
\hspace{2mm}
\subfigure [ACF comparison for $S(\bx_{157})$.]{ \label{fig:acf_S157}
\includegraphics[width=4.5cm,height=4.5cm]{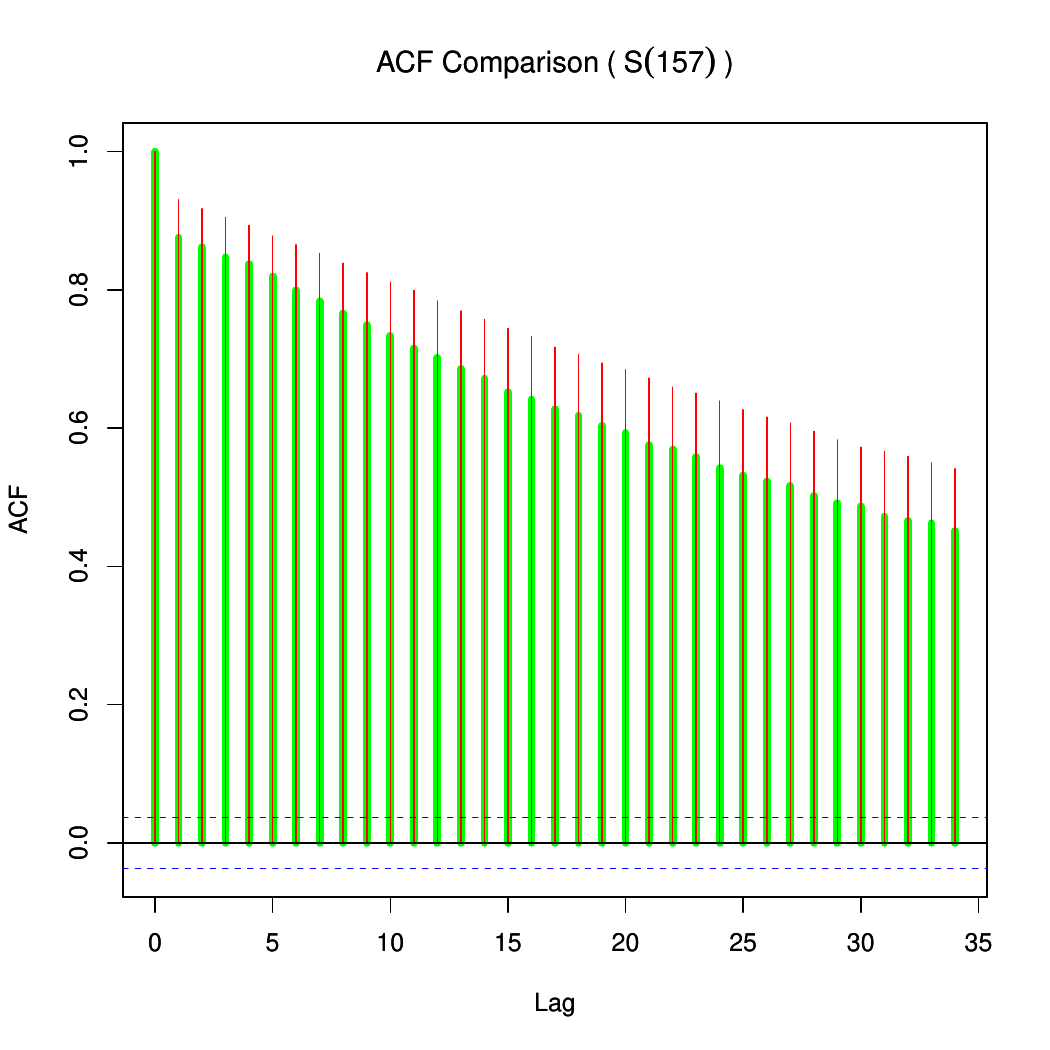}}
\caption{Rongelap island data: Comparisons of the ACF's based on TMCMC and RWM. The TMCMC-based ACF's
are depicted by the green colour and the red colour is associated with RWM-based ACF's.}
\label{fig:acf_rongelap}
\end{figure}

\section{Conclusion}
\label{sec:conclusion}

Overall, our assessment is that TMCMC is clearly advantageous compared to RWM 
from various perspectives. It has significantly less computational complexity and the acceptance rate 
corresponding to the optimal scaling for TMCMC (0.439) is almost twice that of RWM (0.234). 
Although the maximum diffusion speed of RWM is somewhat higher than that of additive TMCMC,
the latter is much more robust with respect to misspecifications of the scales. 
The advantages of such robustness are spelt out in the discussions in Sections S-3 
and S-4 of the supplement. 
Our simulation studies reported in Section \ref{sec:simulation}
and Table \ref{tab:table1}, and the real data analysis in Section \ref{sec:comparison_real},
clearly vindicate these discussions.

Related to the discussions on robustness and the difficulty of choosing proper scalings
in high dimensions is also the issue of increasing computational complexity, particularly
in the Bayesian paradigm.  
Note that complex, high-dimensional posteriors routinely arise in Bayesian applications. 
It is extremely uncommon among MCMC practitioners to use the RWM algorithm for
updating all the parameters in a single block associated with any significantly high-dimensional
posterior arising from any complex Bayesian application. We presume that the extreme 
difficulty of determining proper scalings in practice prevent
the researchers from using the RWM as an algorithm for updating all the parameters in a single block. 
Indeed, as we demonstrated
with our simulation study reported in Table \ref{tab:table1}, misspecification even in the 
case of the simple target distribution being a product of $i.i.d.$ normal densities, leads to
acceptance rates that are almost zero.  
In the context of the real data study in Section \ref{sec:comparison_real} we proposed a method for
approximately obtaining the presumed optimal acceptance rates, and which appears to be generally applicable, 
but as we demonstrated, RWM failed to exhibit adequate mixing properties. 
Adaptive strategies may be thought of as alternative methods, but these are yet to gain 
enough popularity among applied MCMC practitioners; moreover, as we mention in Section S-5 of the supplement, 
extremely long runs may be necessary to reach adequate acceptance rates for adaptive RWM, which may be prohibitive
in very high dimensions, for example, when the acceptance ratio involves high-dimensional matrix
inversions at every iteration, such as in our spatial example. 

The aforementioned difficulties force the researchers to use RWM to {\it sequentially}
update the parameters, either singly, or in small blocks. Since one (or just a few) parameters
are updated at a time by RWM, the acceptance rate can be controlled at each stage of the sequential updating
procedure. However, this sequential procedure also requires computation of the acceptance ratio as many times
as every small block is updated in a sequence. If each parameter is updated singly (that is,
each small block consists of only one element), then the computational complexity increases
$d$-folds compared to the procedure where all the $d$ parameters are updated in a single block.
Thus, when $d$ is large, the computation can become prohibitively slow.

On the other hand, TMCMC is designed to update all the parameters in a single block in such a way
that the acceptance rate remains reasonable in spite of the high dimensionality and complexity
of the target distribution. Our simulation
studies and real data example in Section \ref{sec:comparison_real} 
show that misspecification of the scales do not have drastic effect on the efficiency
of additive TMCMC, thanks to its robustness property. As a result, with much less effort compared to 
that required for RWM, 
we can achieve reasonable scalings that ensure adequate performance of additive TMCMC, so that
resorting to sequential updating will not be necessary.

This also implies that unlike RWM, additive TMCMC can save enormous computational effort 
when the dimension $d$ is large. Finally, adaptive TMCMC may be of much value in very high dimensions because
of its quick convergence to the correct optimal acceptance rate, and for ensuring good performance. 
The details will be covered in \ctn{Dey13}. 

Our empirical findings reported in this article with respect to
the simulation studies pertaining to $i.i.d.$, independent but non-identical, as well as dependent
cases clearly point towards supremacy of TMCMC over RWM. Quite importantly, in the real, spatial
data example, TMCMC outperformed RWM very significantly. Indeed, even though we could tune the scales 
so as to achieve approximately the
respective optimal acceptance rates, the chosen scales need not be actually optimal, for either of TMCMC and RWM.
Here RWM is convincingly outperformed by TMCMC 
thanks to its remarkably robust nature with respect to the choice of scales.
Thus, all our experiments, particularly, the challenging real data example, lead us to clearly recommend
TMCMC in general situations.

Given the importance of the general TMCMC idea, we have decided to create a software for its general usage. In this regard, 
we have now made available an R package {\bf tmcmcR} for implementing TMCMC along with its adaptive versions at the Github 
page $${\bf https://github.com/kkdey/tmcmcR}.$$ The software will be continuously updated in accordance
with further developments of TMCMC; moreover, TTMCMC, the
variable-dimensional version of TMCMC, will also be incorporated, and kept updated.

As part of our future work, we plan to extend TMCMC to multiple-try TMCMC, and investigate the 
corresponding optimal scaling theory. By multiple-try TMCMC we mean the TMCMC algorithm that selects
the next proposal from a set of available, perhaps dependent, proposals. For MH-adapted versions of such
an idea, see, for example, \ctn{Liu2000} and \ctn{Liang10}, \ctn{Martino13}. \ctn{Bedard12} investigated scaling analysis
of such methods in the MH context. The advantages of TMCMC over MH quite reasonably lead us to expect
substantial gains of multiple-try TMCMC over multiple-try MH.

\section*{Acknowledgment}
We are sincerely grateful to the two reviewers for very encouraging and constructive comments
which helped substantially improve the quality and presentation of our manuscript.


\input{supp}

\bibliographystyle{natbib}
\bibliography{irmcmc}

\newpage

\end{document}

%% file: supp.tex
\renewcommand\thefigure{S-\arabic{figure}}
\renewcommand\thetable{S-\arabic{table}} 
\renewcommand\thesection{S-\arabic{section}}

\setcounter{section}{0}
\setcounter{figure}{0}
\setcounter{table}{0}

\begin{center}
{\bf \Large Supplementary Material}
\end{center}

Throughout, we refer to our main manuscript \ctn{Dey15a} as DB. 

\section{Computational efficiency of TMCMC}
\label{subsec:computation}

It may seem that TMCMC is 
computationally more expensive because 
we are randomly generating $d+1$ many values $(\epsilon, b_{1},b_{2},\ldots,b_{d})$ whereas in RWM, 
we are generating $d$ many random variables $(\epsilon_{1},\epsilon_{2},\ldots,\epsilon_{d})$ 
where $\epsilon \sim N(0,\frac{\ell^2}{d})I_{\{\epsilon>0\}}$ and $\epsilon_{i} \sim N(0,\frac{\ell^2}{d}) 
~\forall~ i=1,2,\ldots,d$. However generating  $b_i$ is equivalent to simple tosses of a fair coin which 
is a much easier exercise compared to drawing a set of independent normal random variables required by RWM. 
As a vindication of this, 
we obtain the average computation time (in seconds) 
of 100,00,00 iterations with RWM and TMCMC algorithms across dimensions 2 to 100, averaged over 100
replications of 100,00,00 iterations for each dimension, when the $d$-dimensional target
density is a product of $d$ standard normals. 
The codes are written in C and implemented in an ordinary 32-bit
laptop. The results of our experiment are displayed in Figure \ref{fig:comptime_average}. 
Fgure \ref{fig:comptime} displays the computation times
of RWM and TMCMC with respect to a single replication of the 100,00,00 iterations. 
In both the diagrams, TMCMC 
is seen to take significantly less computational time compared 
to the RWM algorithm, particularly in higher dimensions. Much longer runs, particularly in very high dimensions, 
would see TMCMC saving very substantial amount of computational
time in comparison with RWM. For further discussion on computational gains of TMCMC
over RWM, see Section 9 of DB. 

\begin{figure}
\centering
\includegraphics[width=15cm,height=10cm]{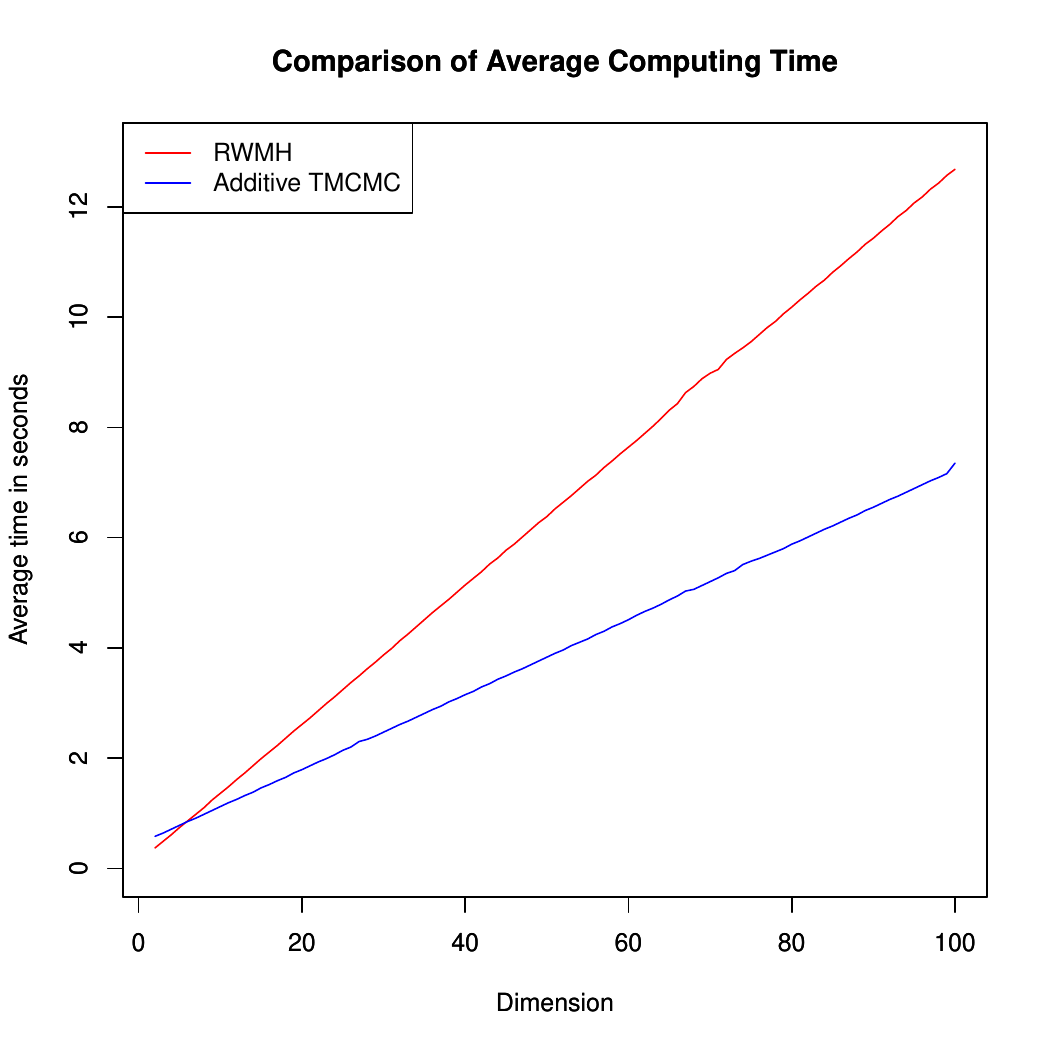}
\caption{Average computation time (in seconds) of 100,00,00 iterations with RWM and TMCMC algorithms corresponding to 
dimensions varying from 2 to 100. TMCMC takes significantly less computation time compared to RWM,
particularly in higher dimensions. The codes are written in C and implemented on an ordinary 32-bit laptop.}
\label{fig:comptime_average}
\end{figure}

\begin{figure}
\centering
\includegraphics[width=15cm,height=10cm]{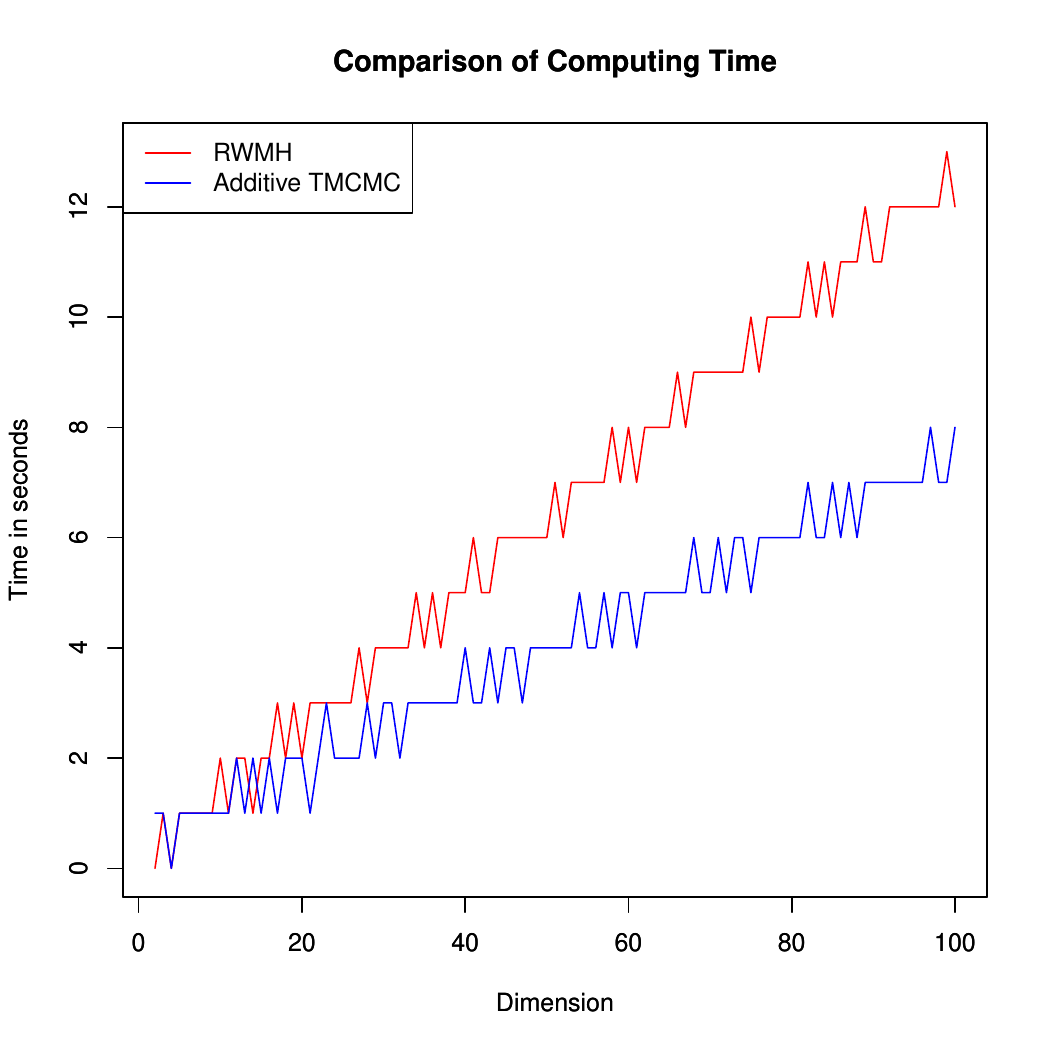}
\caption{Computation time (in seconds) of 100,00,00 iterations with RWM and TMCMC algorithms corresponding to 
dimensions varying from 2 to 100. TMCMC takes significantly less computation time compared to RWM,
particularly in higher dimensions. The codes are written in C and implemented on an ordinary 32-bit laptop.}
\label{fig:comptime}
\end{figure}

It must be emphasized that the proposal density for $\epsilon$ in TMCMC can be any distribution on 
the positive support. Similarly, the RWM algorithm also does not require the proposal to be normal. 
However, the optimal scaling 
results for RWM inherently assume normality and for the sake of comparison, we have also restricted 
our focus on $\epsilon \sim N(0,\frac{\ell^2}{d})I_{\{\epsilon>0\}}$ in the subsequent sections. 

\section{Details on the need for optimal scaling of additive TMCMC}
\label{sec:optimal_scaling}

We first try to impress 
the fact that too small or too large values of $\ell$ can both lead to poor performance of the chain 
and it is this trade-off that draws our interest in finding an optimal value of $\ell$. If the value of $\ell$ 
or equivalently, the proposal variance, is large, then the probability of a move falling in low density regions 
(with respect to the target density) of the space increases as the moves $
(x_{1}+ b_{1}\epsilon, \ldots, x_{d}+b_{d}\epsilon)$ are likely to be quite far apart from 
$(x_{1}, \ldots, x_{d})$. This leads to smaller values of the ratio 
$\frac{\pi(x_{1}+ b_{1}\epsilon, \ldots, x_{d}+b_{d}\epsilon)}{\pi(x_{1}, \ldots, x_{d})}$ and thus 
lower acceptance rates. In fact, for high dimensions, this acceptance rate can be quite low for even
moderately large values of $\ell$. On the other hand, if the value of $\ell$ is too small, then the 
acceptance rate will be higher but we then have to compromise in terms of exploration of the space. 
Much of our moves will lie very close to the initial point and as a result, the chain will move very slowly. 
An instance of the movement of the RWM and additive TMCMC chain for significantly small and large values of 
$\ell$ are depicted in Figure~\ref{fig:fig70} assuming that the target distribution is 
standard normal. For small values of $\ell$, the fact that the chain moves 
slowly gets reflected in the autocorrelation factor (ACF) of the chain, which would be on the higher side 
(Figure~\ref{fig:fig8}). All these motivate us to find an optimal value of $\ell$ that would take care 
of these problems. Our approach would be to derive the diffusion process approximation 
of the additive TMCMC process in the limit as $d \rightarrow \infty$ and then we maximize the 
diffusion speed or the rate of change of variance of the chain in the limit. Intuitively, if the 
scale factor $\ell$ is small, then starting from a point $X_{t}$ at time $t$, the moves corresponding 
to adjacent  time points $X_{t+h}$ are quite close and so the limiting change of variation is quite small 
for the corresponding diffusion process. If on the other hand the scale factor $\ell$ is large, 
the chain hardly moves, 
and hence $X_{t+h}$ for sufficiently small $h$ are often same as $X_{t}$, thereby leading to lower 
value of diffusion speed. On optimizing the diffusion speed for the TMCMC chain, we obtain 
the optimal value of the acceptance rate to be 0.439. Panels (a), (b) and (c) of 
Figure~\ref{fig:fig9} depict the path of the 
TMCMC chain for various choices of proposal variance, ranging from too small through the 
optimal value to quite large. 
Note that the target density is best approximated by the chain at optimal scaling. A better understanding of 
this is achieved by perceiving how well the histogram of observations obtained after running a chain up to 
a certain length of time, approximates the true density (panels (c), (d), (e) of Figure~\ref{fig:fig9}).


\begin{figure}
\centering
\subfigure [Sample paths of RWM and additive TMCMC for small proposal variance.]{ \label{fig:fig71}
\includegraphics[trim= 0cm 10cm 0cm 10cm, clip=true, width=14cm,height=4cm]{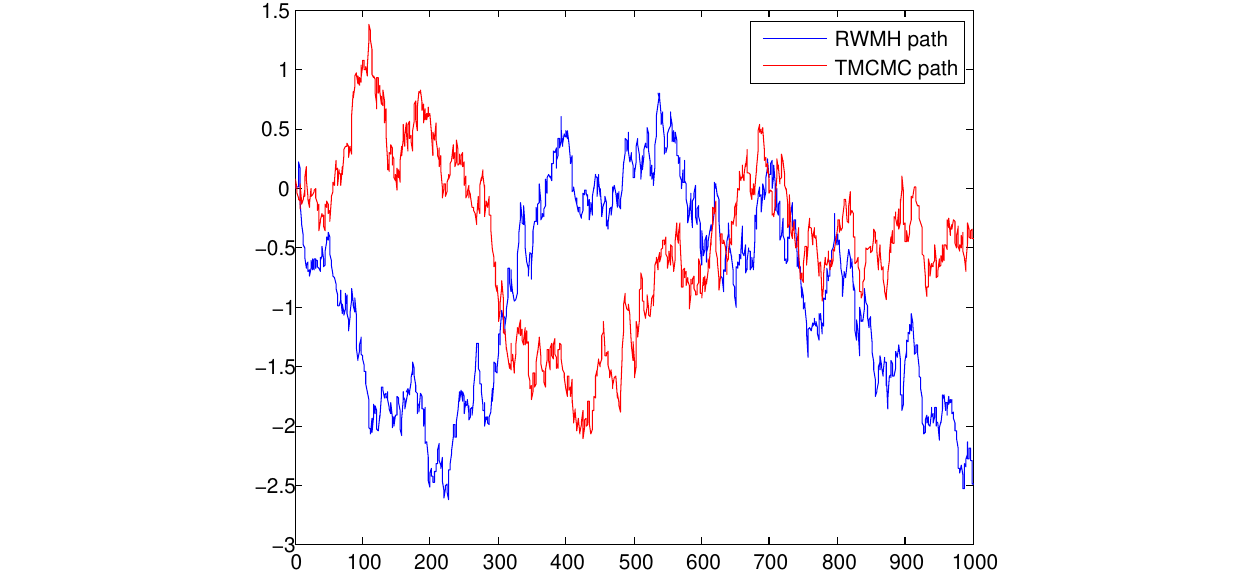}}\\
\subfigure [Sample paths of RWM and additive TMCMC for large proposal variance.]{ \label{fig:fig72}
\includegraphics[trim= 0cm 10cm 0cm 10cm, clip=true, width=14cm,height=4cm]{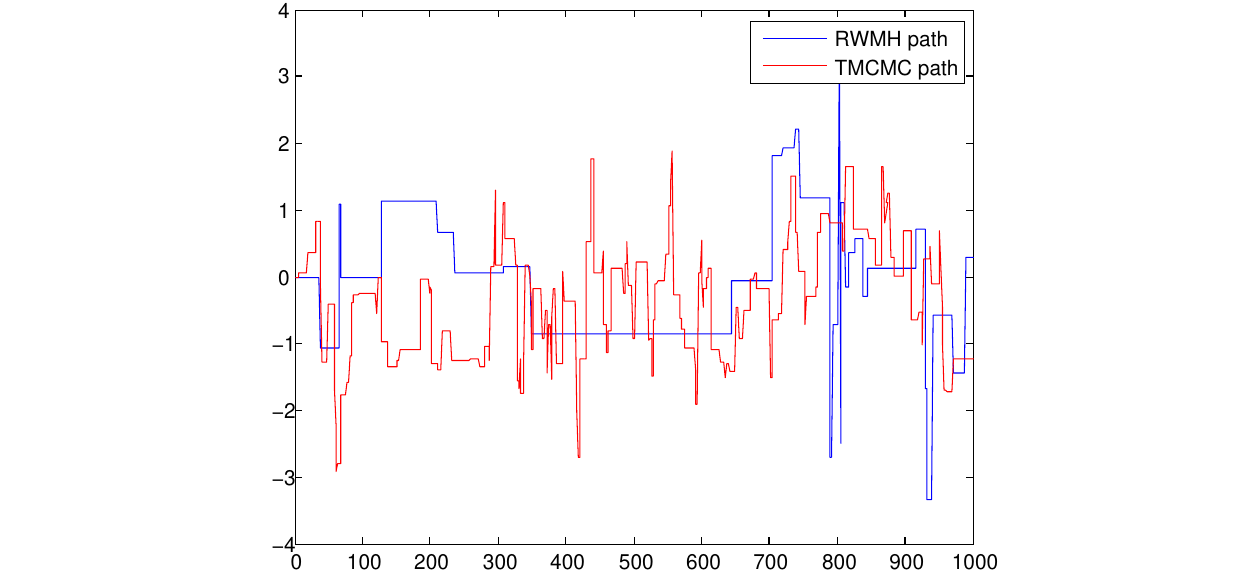}}
\caption{Comparison of RWM and additive TMCMC sample paths for small and large values of the proposal variance.
The target density is $N(0,1)$, the standard normal distribution.}
\label{fig:fig70}
\end{figure}

\begin{figure}
\centering
\includegraphics[width=10cm,height=6cm]{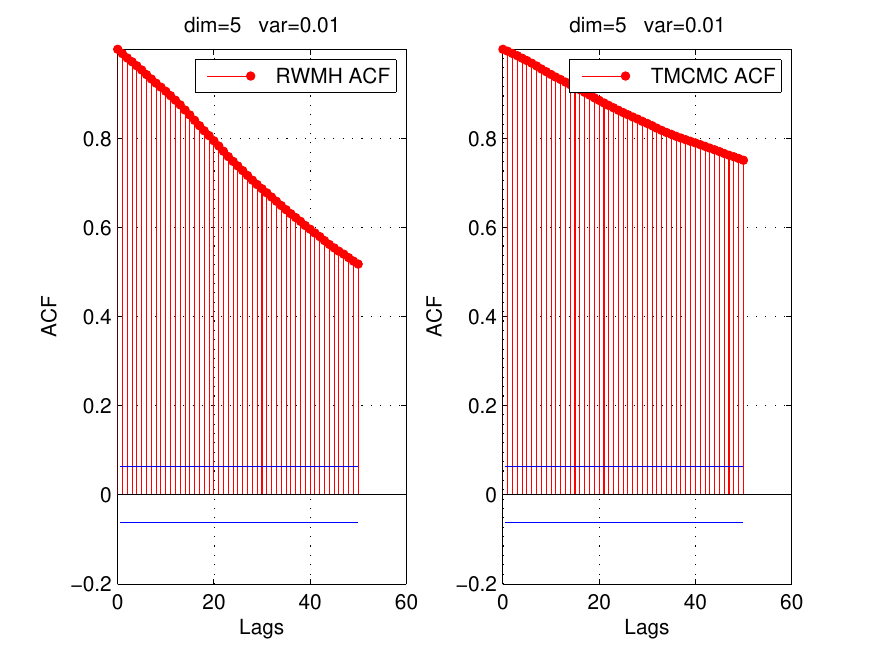}
\caption{ACFs of RWM and additive TMCMC for small proposal variance.
The target density is $N(0,1)$, the standard normal distribution.}
\label{fig:fig8}
\end{figure}

\begin{figure}
\subfigure [$\ell=0.01$]{ \label{fig:fig91}
\includegraphics[width=5cm]{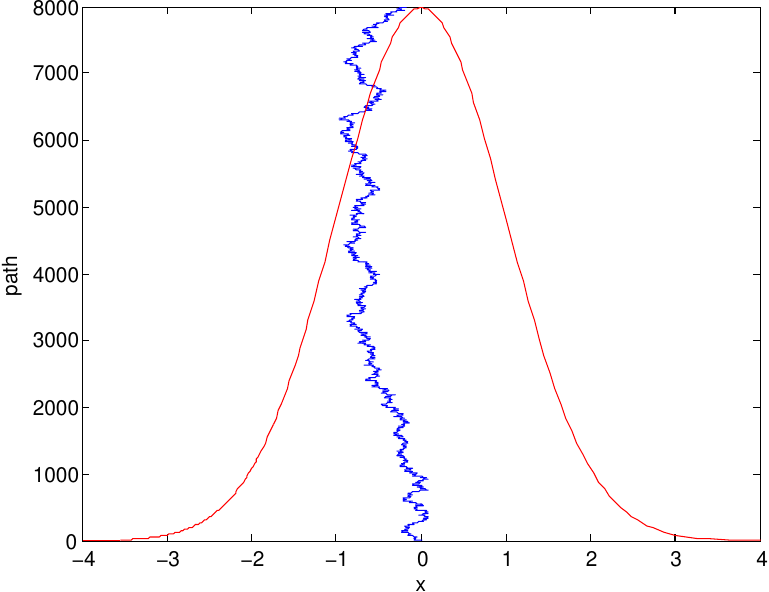}}
\hspace{-2mm}
\subfigure [$\ell=2.4~(\mbox{optimal})$]{ \label{fig:fig92}
\includegraphics[width=5cm]{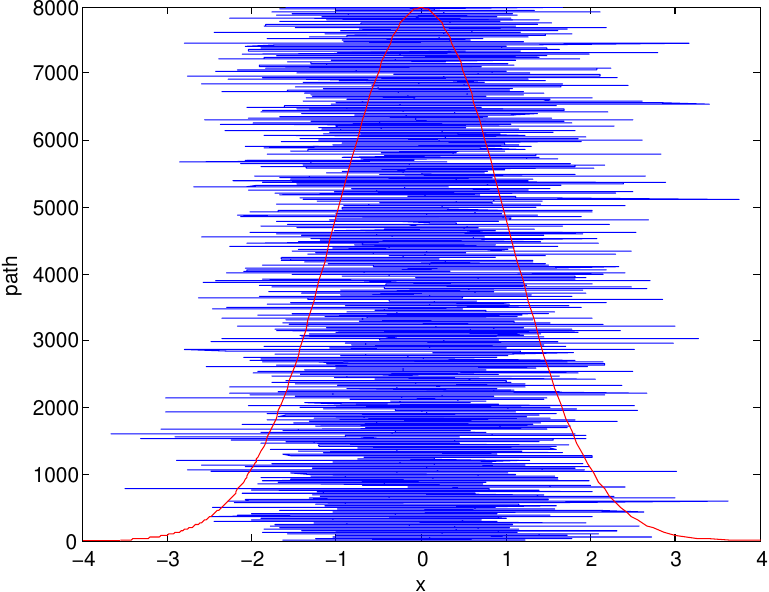}}
\hspace{-2mm}
\subfigure [$\ell=20$]{ \label{fig:fig93}
\includegraphics[width=5cm]{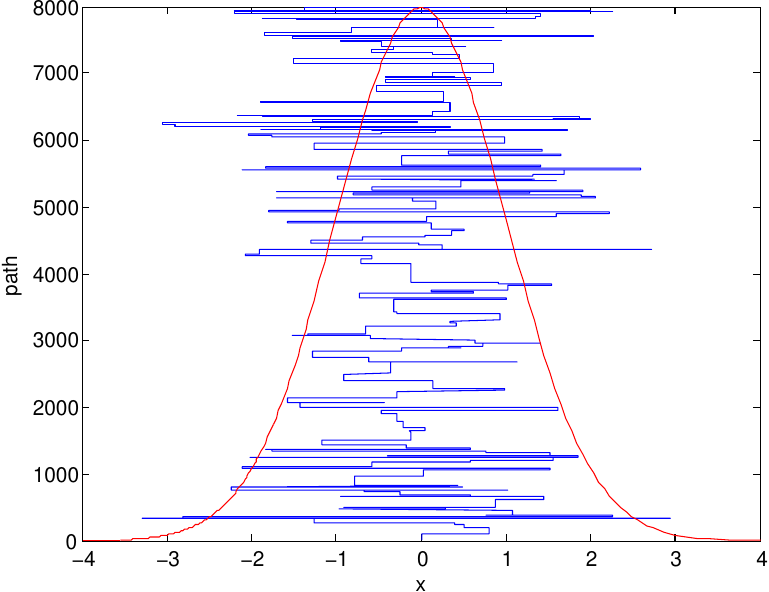}}\\
\subfigure [$\ell=0.01$]{ \label{fig:fig101}
\includegraphics[width=5cm]{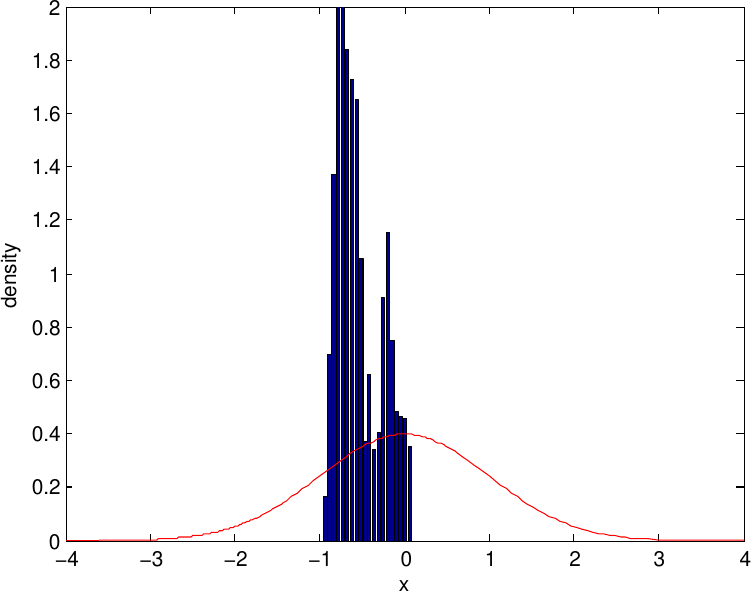}}
\subfigure [$\ell=2.4~ (\mbox{optimal})$]{ \label{fig:fig102}
\includegraphics[width=5cm]{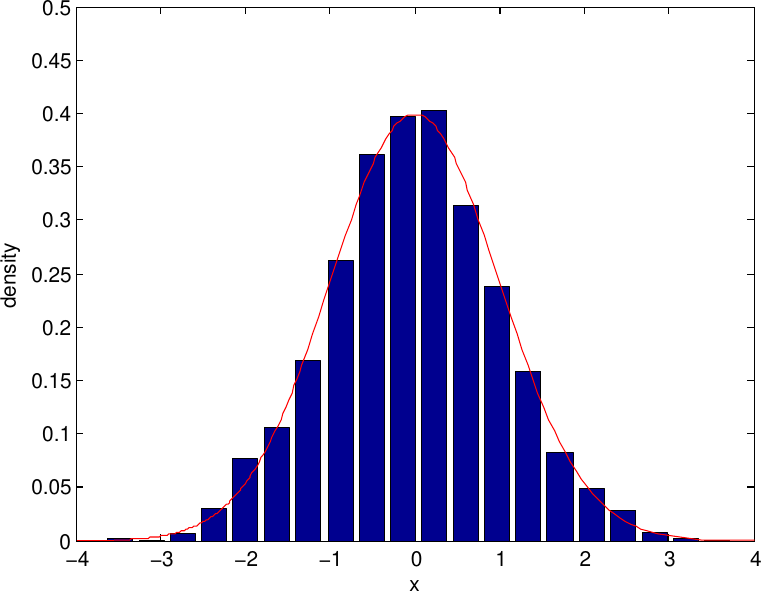}}
\subfigure [$\ell=20$]{ \label{fig:fig103}
\includegraphics[width=5cm]{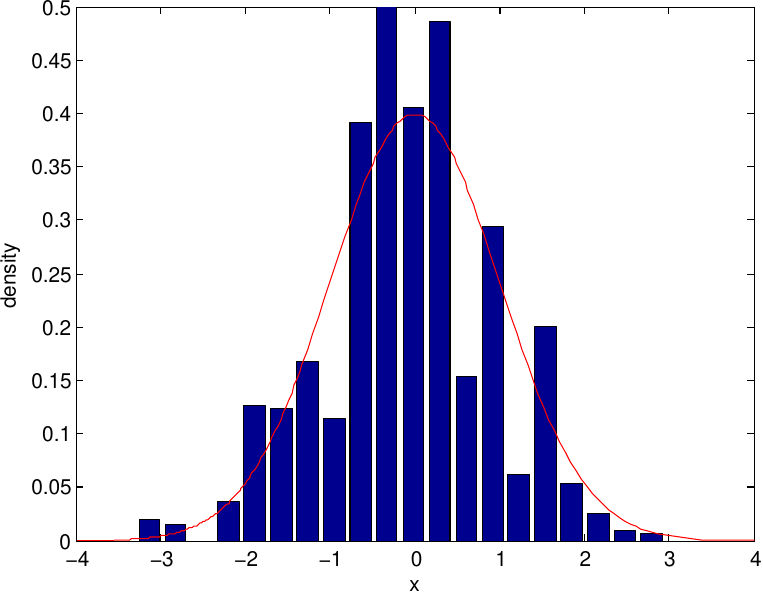}}
\caption{ The upper panels (a), (b) and (c) show the paths of TMCMC chain for three various choices 
of scalings together with the target density $N(0,1)$ plot. These highlight how well the paths explore 
the given target density. The lower panels (c), (d) and (e) display the histograms of the sample observations 
obtained from TMCMC for these choices of scalings together with the target density $N(0,1)$. 
These highlight how well the histograms approximate the target density for the given run of TMCMC.
}
\label{fig:fig9}
\end{figure}

\section{Discussion on consequences of non-robustness of RWM with respect to scale choices}
\label{subsec:non_robustness}

For general, $d$-dimensional target distributions, RWM entails the proposal with 
transitions of the form $(x_1,\ldots,x_d)\rightarrow 
(x_1+\frac{\ell_1}{\sqrt{d}}\epsilon_1,\ldots,x_1+\frac{\ell_d}{\sqrt{d}}\epsilon_d)$,
where, for $i=1,\ldots,d$, $\epsilon_i\sim N(0,1)$, and $\ell_i$ are constants to be chosen appropriately.
Often $\ell_i$ may be of the form $\ell a_i$, where $a_i$ may be needed to determine appropriately
in addition to $\ell$. As instance of this form occurs in the dependent set-up of \ctn{Pillai2011},
but $a_i$ in that set-up are the square roots of the eigenvalues of the Gaussian measure with independent components that 
is assumed to dominate the target density. However, in practice such assumption of independent Gaussian
components will generally not hold, and hence it is far from straightforward
to determine $a_i$ appropriately in realistic cases.

Since all the set-ups considered so far yield the optimal acceptance
rate 0.234 for RWM, it may be anticipated that the result holds quite generally, and applied
MCMC practitioners may be advised to tune $(\ell_1,\ldots,\ell_d)$ such that the acceptance
rate is close to 0.234. In fact, using a measure of efficiency which is the reciprocal of
integrated autocorrelation time, \ctn{Roberts01} demonstrate that the RWM proposal may be tuned
to achieve an acceptance rate between 0.15 to 0.5, which would make the algorithm around 80\%
efficient. However, for large dimension $d$, appropriate tuning of so many scale parameters
seems to be an extremely arduous task. In our simulations presented in Section 8 of DB 
we observe that even in the simple
situation where the target density is an $iid$ product of normal densities, when
the dimension increases, particularly when $d=100$ and $d=200$, departure from the
optimal scale results in drastic fall in acceptance rates, far below what
is prescribed by \ctn{Roberts01}; see Table 1 of DB. 
The diffusion speeds under such mis-specifications tend to be quite low because of
non-robustness with respect to scale choice (see Figures 1 -- 6 of DB). 
Since low diffusion speed is equivalent to high autocorrelation 
(see equation (18) of \ctn{Roberts01}), the efficiency measures of \ctn{Roberts01} 
that use integrated autocorrelation,
are also expected to indicate less efficiency.
Thus, in summary manually tuning the RWM proposals appropriately
in more general and complicated situations and in high dimensions seems to be a very daunting task.
There are methods for automatically and adaptively tuning the scales to 
progress towards the desired optimal acceptance rate as number of iterations increases, but
as we discuss in Section \ref{subsec:adaptive} computationally this can be a very costly exercise
in high dimensions.

\section{Discussion on possible advantages of additive TMCMC for relatively more
robust behaviour with respect to scale choices}
\label{subsec:robustness}

Our results on optimal scaling offers the following general thumb rule to the users of additive TMCMC: tune
the additive TMCMC proposal to achieve approximately 44\% acceptance rate. Note that even though
the optimal acceptance rate of additive TMCMC is significantly higher than that of RWM, both the algorithms
have approximately the same optimal scalings that maximize the diffusion speeds (see Figures 1 -- 6).

The results of our simulation studies reported in Table 1 of DB 
demonstrate that even in 
dimension as low as $d=2$, our optimal acceptance rate 0.439 is remarkably accurate. The table further
demonstrates that even if the scale of additive TMCMC is sub-optimally chosen, the acceptance rates
remain higher than 20\% for all dimensions, whereas for the same sub-optimal scale choice 
the acceptance rate of RWM falls to about 0.33\% in high dimensions. 
Figures 1 -- 6 
show that around the optimal scale, the diffusion speeds 
of additive TMCMC under various set-ups do not change significantly. Using the relationship
between diffusion speed and the measure of efficiency proposed by \ctn{Roberts01} one can conclude
that unlike RWM, the efficiency of additive TMCMC is not substantially affected by sub-optimal scale choices. 
Hence, tuning the additive TMCMC proposal is a far more safe and
easy exercise compared to that of RWM. It seems to us that this is quite an advantage of additive
TMCMC over RWM in general, high-dimensional set-ups.
As we discuss in Section \ref{subsec:adaptive} adaptive methods can be employed to 
approach the exact optimal acceptance rate 0.439 in substantially less number of iterations 
compared to adaptive RWM, facilitating huge computational gain over the latter.

\section{Discussion on adaptive versions of RWM and additive TMCMC for enforcing optimal
acceptance rates in complex, high-dimensional problems}
\label{subsec:adaptive}
Adaptive MCMC methods (see, for example, \ctn{Roberts09} and the references therein) 
are designed to combat the difficulty of determining appropriate proposal scalings.
In the context of RWM, various adaptive strategies are presented in \ctn{Roberts09} to choose
the scalings in an adaptive manner so that the optimal acceptance rate 0.234 is achieved
in the long run. \ctn{Dey13} adopted the strategies in the case of additive TMCMC and made
a detailed comparison with the corresponding adaptive RWM methods. In particular, they found
that even after a very large number of iterations most of the the adaptive methods related to RWM 
yielded acceptance rates which are significantly different from 0.234, while the adaptive TMCMC 
algorithms very quickly yielded acceptance rates reasonably close 0.439, even in dimensions as low as $d=2$. 
This implies quite substantial savings of TMCMC in terms of computation time in comparison with RWM.
Performance wise as well, the results of \ctn{Dey13}  
favour adaptive TMCMC over adaptive RWM in high dimensions, with respect to the various measures
which we also employ in this current work.

\section{Proofs of optimal scaling results of additive TMCMC}
\label{sec:proofs_results}
\input{appendix}




%% file: appendix.tex

\subsection{Proof of Theorem 3.1 of DB} 
\label{proof:theorem1}
\begin{proof}
For our purpose, we define the discrete time generator of the TMCMC approach, as 

\begin{eqnarray}\label{eq:generator}
G_{d}{V(x)}&=& \frac{d}{2^{d}} \displaystyle \sum_{ \left \{\begin{array}{l} b_{i}\in \{-1,+1\} \\  
\forall i=1,\ldots,d \end{array}\right \}} \displaystyle \int_{0}^{\infty} 
\left [\left ( \vphantom{min \left \{ 1, \frac{\pi(x_{1}+b_{1}\epsilon, \ldots, x_{d}+b_{d}\epsilon)}{\pi(x_{1},\ldots, x_{d})} \right \}} 
V \left (x_{1}+b_{1}\epsilon, \ldots, x_{d}+b_{d}\epsilon \right) - V \left (x_{1},\ldots, x_{d}\right ) \right) \right. \nonumber \\
&& \qquad \left. \hspace{3 cm} \times \left ( \normalsize \min \left \{ 1, \frac{\pi(x_{1}+b_{1}\epsilon, \ldots, x_{d}+b_{d}\epsilon)}{\pi(x_{1},x_{2}, 
\ldots, x_{d})} \right \} \right ) \right ] q(\epsilon)d\epsilon. \nonumber \\
\end{eqnarray}
Since by our assumption $(\log f)'$ is Lipschitz, in the above equation we may assume 
that $V$ belongs to the space of infinitely differentiable functions with 
compact support (see, for example, \ctn{Bedard2007} for further details).

The Skorohod topology allows us to treat $G_d$ as a continuous time generator 
that has jumps at the rate $d^{-1}$. Given our restricted focus on a one dimensional component of the actual process, we assume 
$V$ to be a function of the first co-ordinate only. Under this assumption, the generator defined in (\ref{eq:generator}) 
is a function of only $\epsilon$ and $b_{1}$, and can be rephrased as 

\begin{eqnarray}\label{eq:realgenerator}
G_{d}{V(x)} &=& \frac{d}{2}\int_{0}^{\infty}\sum_{b_{1}\in \{-1,+1\}} 
\left [\left (\vphantom{\min \left \{ 1, \frac{\pi(x_{1}+b_{1}\epsilon, \ldots, x_{d}+b_{d}\epsilon)}{\pi(x_{1},\ldots, x_{d})} \right \}} 
V(x_{1}+b_{1}\epsilon) - V (x_{1}) \right)  \right. \nonumber \\
&& \qquad \quad \left. \hspace{2 cm} \times E_{b_{2},\ldots, b_{d}}
\left ( \normalsize \min \left \{ 1, \frac{\pi(x_{1}+b_{1}\epsilon, \ldots, x_{d}+b_{d}\epsilon)}{\pi(x_{1},\ldots, x_{d})} \right \} \right ) 
\right ] q(\epsilon)d\epsilon, \nonumber \\
\end{eqnarray}
where $E_{b_{2},\ldots,b_{d}}$ is the expectation taken conditional on $b_{1}$ and $\epsilon$. 

First we show that the quantity $G_{d}V(x)$ is a bounded quantity. 
\begin{eqnarray}\label{eq:bdd}
G_{d}{V(x)} &\leq & d E_{\{b_1,\epsilon\}}\left[ V(x_1+b_1\epsilon) - V(x_1) \right]  \nonumber \\
&=& dV^{'}(x_1)E_{\{b_1,\epsilon\}}(b_1\epsilon) + \frac{d}{2}V^{''}(x^*_1)E_{\{b_1,\epsilon\}}(\epsilon^{2}) \nonumber \\
&\leq&  \ell^{2}K, \nonumber  \\
\end{eqnarray}
where $x^*_1$ lies between $x_1$ and $x_1+b_1\epsilon$ and $K$ is the maximum value of $V^{''}$. 

Note that
\begin{eqnarray}
& E_{b_2,\ldots,b_d} & \left [\min \left \{ 1, \frac{\pi(x_{1}+b_{1}\epsilon, \ldots, x_{d}+b_{d}\epsilon)}{\pi(x_{1},\ldots, x_{d})} \right \} 
\right ]\nonumber  \\
&=&E_{b_2,\ldots,b_d} \left  [  \min \left \{ 1, \exp \left (  \vphantom{\sum_{j=2}^{d} \left \{ b_{j}\epsilon \left \{ \log(f(x_{j}))\right \}^{'} 
+  \frac{\epsilon^{2}}{2!} \left \{ \log(f(x_{j}))\right \}^{''} +  \frac{b_{j}\epsilon^{3}}{3!} \left \{ \log(f(z_{j}))\right \}^{'''} \right \}} 
\log (f(x_{1}+b_{1}\epsilon)) - \log (f(x_{1})) \right. \right. \right. \nonumber \\
&&  \left. \left. \left. + \sum_{j=2}^{d} \left \{ b_{j}\epsilon \left \{ \log(f(x_{j}))\right \}^{'} +  \frac{\epsilon^{2}}{2!} 
\left \{ \log(f(x_{j}))\right \}^{''} +  
\frac{b_{j}\epsilon^{3}}{3!} \left \{ \log(f(x_{j}))\right \}^{'''} +
\frac{\epsilon^{4}}{4!} \left \{ \log(f(z_{j}))\right \}^{''''}\right\}\right ) \right \} \right ], 
\nonumber \\
\label{eq:mineq}
\end{eqnarray}
where 
$E_{b_2,\ldots,b_d}$ denotes expectation with respect to $b_2,\ldots,b_d$, holding
$\epsilon$, $b_1$, $x_1$, $x_j$ and $z_j$ ($j=2,\ldots,d$) fixed;
and for $j=2,\ldots,d$, $z_j$ lies between $x_j$ and $x_j+b_j\epsilon$.
Since $b_{j}; j = 2,\ldots,d$ are $iid$, as $ d \rightarrow \infty $, conditional on $\epsilon$, $b_1$, $x_1$, 
$x_j$ and $z_j$ ($j=2,\ldots,d$) one can apply Lyapunov's central limit theorem. 



Writing $\zeta_j=b_j\left[\epsilon\{\log(f(x_j))\}'+\frac{\epsilon^3}{6}\{\log(f(x_j))\}'''\right]$, we note that
conditional on $\epsilon$ and $x_j$, $E_{b_j}(\zeta_j)=0$ and $Var_{b_j}(\zeta_j)=
\left[\epsilon\{\log(f(x_j))\}'+\frac{\epsilon^3}{6}\{\log(f(x_j))\}'''\right]^2$. 
Viewing $b_j\epsilon$ as $b_j\epsilon^*\frac{\ell}{\sqrt{d}}$, 
where $\epsilon^*\sim N(0,1)I_{\{\epsilon^*>0\}}$,  
we next show that, almost surely with respect to $\pi$, 
\[ \frac{\sum_{j=2}^dE_{b_j}(\left\vert\zeta_j\right\vert^\delta)}{\left\{\sqrt{\sum_{j=2}^dVar_{b_j}(\zeta_j)}\right\}^\delta}
\rightarrow 0,\]
for $\delta=4$.

First note that $\epsilon\equiv\epsilon^*\frac{\ell}{\sqrt{d}}$, where $\epsilon^*\sim N(0,1)I_{\{\epsilon^*>0\}}$,
and so, for any $\zeta>0$, 
\begin{equation}
\sum_{d=1}^{\infty}P\left(\epsilon^*\frac{\ell}{\sqrt{d}}>\zeta\right)
<\left(\frac{\ell}{\zeta}\right)^4E\left({\epsilon^*}^4\right)\sum_{d=1}^{\infty}\frac{1}{d^2}<\infty.
\label{eq:as_conv}
\end{equation}
That is, $\epsilon\equiv\epsilon^*\frac{\ell}{\sqrt{d}}\stackrel{a.s.}{\longrightarrow}0$,
$a.s.$ denoting ``almost surely".
Thus, there exists a null set $\mathcal N_{\epsilon}$ (with respect to the distribution of $\epsilon^*$) 
such that for all $\omega_{\epsilon}\in\mathcal N^c_{\epsilon}$,
$\epsilon\equiv\epsilon^*(\omega_{\epsilon})\frac{\ell}{\sqrt{d}}\rightarrow 0$, as $d\rightarrow\infty$.
Now observe that, for any given $\omega_{\epsilon}\in\mathcal N^c_{\epsilon}$, as $d\rightarrow\infty$, 
$\frac{1}{d-1}\sum_{j=2}^dE_{b_j}(\left\vert\zeta_j\right\vert^\delta)=\frac{1}{d-1}\sum_{j=2}^d
\left[\epsilon\{\log(f(x_j))\}'+\frac{\epsilon^3}{6}\{\log(f(x_j))\}'''\right]^4\stackrel{a.s.}{\longrightarrow}
E_{x_2}\left[\epsilon\{\log(f(x_2))\}'+\frac{\epsilon^3}{6}\{\log(f(x_2))\}'''\right]^4$, by 
the strong law of large numbers (SLLN).
The expectation, which is with respect to $x_2$, is clearly finite, due to the assumptions 
(13), (14), (15) of DB
and the Cauchy-Schwartz inequality. In other words, given $\omega_{\epsilon}\in\mathcal N^c_{\epsilon}$, 
there exists a null set 
$\mathcal N_{1}$ (with respect to $f$)
such that 
for all $\omega\in\mathcal N^c_{1}$, the convergence takes place deterministically.
Also, 
\begin{align}
\frac{1}{d-1}\sum_{j=2}^dVar_{b_j}(\zeta_j)&=
\frac{1}{d-1}\sum_{j=2}^d\left[\epsilon\{\log(f(x_j))\}'+\frac{\epsilon^3}{6}\{\log(f(x_j))\}'''\right]^2\notag\\
&\stackrel{a.s.}{\longrightarrow}E_{x_2}\left[\epsilon\{\log(f(x_2))\}'+\frac{\epsilon^3}{6}\{\log(f(x_2))\}'''\right]^2,\notag
\end{align}
which
is again finite, thanks to 
(13), (14) and (15) of DB, and the Cauchy-Schwartz inequality.
Let $\mathcal N_{2}$ denote the null set (with respect to $f$) such that
deterministic convergence takes place for all $\omega\in\mathcal N^c_{2}$. 
Let $\mathbb N_1=\mathcal N_{\epsilon}\otimes\mathcal N_1\cup\mathcal N_{\epsilon}\otimes\mathcal N_2$, where
$\otimes$ denotes cartesian product.
Then $\mathbb N_1$ is a null set with respect to the distribution of $\epsilon$ and $f$. For $\omega\in\mathbb N^c_1$,
we have, for $\delta=4$, $\sum_{j=2}^dE_{b_j}(\left\vert\zeta_j\right\vert^\delta) = O(d)$, and 
$\sum_{j=2}^dVar_{b_j}(\zeta_j)=O(d)$. Hence, for $\delta=4$, and for $\omega\in\mathbb N^c_1$, we have,
\[ \frac{\sum_{j=2}^dE_{b_j}(\left\vert\zeta_j\right\vert^\delta)}{\left\{\sqrt{\sum_{j=2}^dVar_{b_j}(\zeta_j)}\right\}^\delta}
=O(d^{-1})\rightarrow 0.\]

Thus, Lyapunov's central limit theorem applies, and we have the following: 


\begin{equation}\label{eq:dist}
\frac{\sum_{j=2}^{d}{ b_{j}\left [\epsilon \left \{\log (f(x_{j}))\right \}^{'}+
\frac{\epsilon^{3}}{6} \left \{ \log(f(x_{j}))\right \}^{'''} \right ]}}{\sqrt {\sum_{j=2}^{d} 
\left [  \epsilon \{\log (f(x_{j}))\right \}^{'}+\frac{\epsilon^{3}}{6}\left \{ \log(f(x_{j}))\right \}^{'''}]^{2} }} 
\stackrel{\mathcal L}{\rightarrow} N(0,1), 
\end{equation}
for almost all $\epsilon$ and $x_2,\ldots,x_d$. 
Also note that, the square of the denominator of (\ref{eq:dist}) is given by 
\begin{equation}
\sum_{j=2}^{d} 
\left [  \epsilon \{\log (f(x_{j}))\right \}^{'}+\frac{\epsilon^{3}}{6}\left \{ \log(f(x_{j}))\right \}^{'''}]^{2}
= \epsilon^{2}\sum_{j=2}^{d}{\left[\left \{\log (f(x_{j}))\right \}^{'}\right]^{2} } +\Delta, 
\label{eq:deno}
\end{equation}
where
\begin{equation}\label{eq:delta}
\Delta= \frac{\epsilon^{4}}{6}\sum_{j=2}^{d}{ 2 \left \{\log (f(x_{j}))\right \}^{'}
\left \{\log (f(x_{j}))\right \}^{'''}}+ 
\frac{\epsilon^{6}}{36}\sum_{j=2}^{d}{\left[\left \{\log (f(x_{j}))\right \}^{'''}\right]^{2}}. 
\end{equation}
With the representation $\epsilon\equiv \epsilon^*\frac{\ell}{\sqrt{d}}$, where $\epsilon^*\sim N(0,1)I_{\{\epsilon^*>0\}}$,
for $\omega_{\epsilon}\in\mathcal N^c_{\epsilon}$,
the first term of (\ref{eq:deno}) is given by 
${\epsilon^*}^2\frac{\ell^2}{d}\sum_{j=2}^d\left[\left\{\log (f(x_{j}))\right\}^{'}\right]^2$. As $d\rightarrow\infty$,
by SLLN, there exists a null set $\mathcal N_3$ with respect to $f$ such that
for all $\omega\in\mathcal N^c_3$, 
\[
\frac{1}{d-1}\sum_{j=2}^d\left[\left\{\log (f(x_{j}))\right\}^{'}\right]^2
=-\frac{1}{d-1}\sum_{j=2}^{d} \left \{ \log(f(x_{j}))\right \}^{''}
\rightarrow\mathbb I,
\]
where  $\mathbb{I}$  is the information matrix corresponding to the density $f$. 

Also, 
\begin{align}
\Delta &=\frac{{\epsilon^*}^{4}\ell^4}{6d}\frac{1}{d}\sum_{j=2}^{d}{ 2 \left \{\log (f(x_{j}))\right \}^{'}
\left \{\log (f(x_{j}))\right \}^{'''}}+
\frac{{\epsilon^*}^{6}\ell^6}{36d^2}\frac{1}{d}\sum_{j=2}^{d}{\left[\left \{\log (f(x_{j}))\right \}^{'''}\right]^{2}}\notag\\
&\stackrel{a.s.}{\longrightarrow}0,\notag
\end{align}
since 
\[
\frac{1}{d}\sum_{j=2}^{d}{ 2 \left \{\log (f(x_{j}))\right \}^{'}\left \{\log (f(x_{j}))\right \}^{'''}}\stackrel{a.s.}{\longrightarrow}
E_f\left[2\left\{\log (f(x_{2}))\right \}^{'}\left \{\log (f(x_{2}))\right \}^{'''}\right]<\infty,
\]

\[
\frac{1}{d}\sum_{j=2}^{d}{\left[\left \{\log (f(x_{j}))\right \}^{'''}\right]^{2}}
\stackrel{a.s.}{\longrightarrow}E_f{\left[\left \{\log (f(x_{2}))\right \}^{'''}\right]^{2}}<\infty,
\]
$\frac{{\epsilon^*}^4}{d}\stackrel{a.s.}{\longrightarrow}0$, and $\frac{{\epsilon^*}^6}{d^2}\stackrel{a.s.}{\longrightarrow}0$. 
Hence, there exists a null set $\mathcal N_4$ (with respect to $f$) such that for all $\omega\in\mathcal N^c_4$, 
$\Delta(\omega)\rightarrow 0$ as $d\rightarrow\infty$.

We now deal with the last term of (\ref{eq:mineq}) given by 
$\sum_{j=2}^d\frac{\epsilon^4}{4!}\left\{\log(f(z_j))\right\}^{''''}$.
This tends to zero almost surely as $d\rightarrow\infty$ because 
$\epsilon^4\equiv{\epsilon^*}^4\frac{\ell^4}{d}\stackrel{a.s.}{\longrightarrow} 0$, and
$$\frac{1}{d}\sum_{j=2}^d\left\{\log(f(z_j))\right\}^{''''}\stackrel{a.s.}{\longrightarrow}
E_f\left(\left\{\log(f(z_2))\right\}^{''''}\right)<\infty,$$ by (16)of DB. 
Let $\mathcal N_5$ denote the null set associated with the above almost sure convergence.

Let $\mathbb N_2=\mathcal N_{\epsilon}\otimes\mathcal N_3\cup \mathcal N_{\epsilon}\otimes\mathcal N_4
\cup \mathcal N_{\epsilon}\otimes\mathcal N_5$. Further,
let $\mathbb N=\mathbb N_1\cup\mathbb N_2$. Then $\mathbb N$ is a null set with respect to the distribution of
$\epsilon^*$ and $f$.
Hence, given any $\omega\in\mathbb N^c$, 
\[
\left\vert E_{b_2,\ldots,b_d} \left [\min \left \{ 1, \frac{\pi(x_{1}+b_{1}\epsilon, \ldots, x_{d}+b_{d}\epsilon)}{\pi(x_{1},\ldots, x_{d})} \right \} 
\right ]-E_{b_2,\ldots,b_d}\left[\min\left\{1,e^X\right\}\right]\right\vert\rightarrow 0,\hspace{2mm}\mbox{as}\hspace{2mm}d\rightarrow\infty,
\]
where
\begin{equation}
X \sim N \left( \eta (x_{1}, b_{1}, \epsilon) - \frac{(d-1) 
\epsilon^{2}}{2}\mathbb{I}, \hspace{0.1 cm} (d-1)\epsilon^{2}\mathbb{I} \right),
\end{equation}
with
\begin{equation}\label{eq:eta}
\eta (x_{1}, b_{1}, \epsilon) =  \log (f(x_{1}+b_{1}\epsilon)) - \log (f(x_{1})).
\end{equation}

We now recall the following result (Proposition 2.4 of \ctn{Roberts97a}): 
if $X \sim N(\mu,\sigma^{2})$, then 
\begin{equation*}
E \left [ \min \left \{1, e^{X}\right \} \right] = 
\Phi \left (\frac{\mu}{\sigma}\right ) + e^{\left\{\mu +\frac{\sigma^{2}}{2}\right\}}
\Phi \left ( -\sigma - \frac{\mu}{\sigma} \right ), 
\end{equation*}
where $\Phi$ is the standard Gaussian cumulative distribution function (cdf). 
Applying this result to (\ref{eq:mineq}) we obtain 

\begin{eqnarray}\label{eq:simpexp}
& E_{b_2,\ldots,b_d} & \left 
[\min \left \{ 1, \frac{\pi(x_{1}+b_{1}\epsilon, \ldots, x_{d}+b_{d}\epsilon)}
{\pi(x_{1}, \ldots, x_{d})} \right \} \right ]\nonumber  \\
&=& \Phi \left (\frac{\eta (x_{1}, b_{1}, \epsilon) - \frac{(d-1)\epsilon^{2}}{2}\mathbb{I}}
{\sqrt{(d-1)\epsilon^{2}\mathbb{I}}}\right ) 
+ e^{\eta (x_{1}, b_{1}, \epsilon)}\Phi \left (-\sqrt{(d-1)\epsilon^{2}\mathbb{I}} 
- \frac{\eta (x_{1}, b_{1}, \epsilon) 
- \frac{(d-1)\epsilon^{2}}{2}\mathbb{I}}{\sqrt{(d-1)\epsilon^{2}\mathbb{I}}} \right) \nonumber\\
&=& \mathbb{W}(b_{1},\epsilon, x_{1}). \nonumber \\
\end{eqnarray}

Note that using Taylor series expansion around $x_{1}$, we can write (\ref{eq:eta}) as 
\begin{equation}
\eta (x_{1}, b_{1}, \epsilon) = b_{1}\epsilon \left [\log f(x_{1}) \right ]^{'} + \frac{\epsilon^{2}}{2}
\left [\log f(x_{1}) \right] ^{''} + b_{1}\frac{\epsilon^{3}}{3!}\left [\log f(\xi_{1}) \right] ^{'''},
\end{equation}
where $\xi_1$ lies between $x_1$ and $x_1+b_1\epsilon$. 
Re-writing $b_{1}\epsilon$ as $\frac{\ell}{\sqrt{d}}z^*_1$, 
where $z^*_1$ follows a $N(0,1)$ distribution, $\eta$ and $\mathbb W$ 
can be expressed in terms of $\ell$ and $z^*_1$ as 

\begin{equation}
\eta (x_{1}, z^*_1, d) = \frac{\ell z^*_{1}}{\sqrt{d}} \left [\log f(x_{1}) \right ]^{'} 
+ \frac{\ell^{2}{z^*_{1}}^{2}}{2!d}\left [\log f(x_{1}) \right] ^{''} + \frac{\ell^{3}{z^*_{1}}^{3}}
{3!d^{\frac{3}{2}}}\left [\log f(\xi_{1}) \right] ^{'''}
\label{eq:eta_expand}
\end{equation}
and
\begin{equation}\label{eq:W}
\mathbb{W}(x_1, z^*_{1}, d) =  
\Phi \left (\frac{\eta (x_{1}, z^*_{1}, d) - \frac{{z^*_{1}}^{2}\ell^{2}}{2}\mathbb{I}}
{\sqrt{{z^*_{1}}^{2}\ell^{2}\mathbb{I}}}\right ) + e^{\eta (x_{1}, z^*_{1}, d )}
\Phi \left (\frac{- \eta (x_{1}, z^*_{1}, d)
-\frac{{z^*_{1}}^{2}\ell^{2}\mathbb{I}}{2}}
{\sqrt{{z^*_{1}}^{2}\ell^{2}\mathbb{I}}} \right).
\end{equation}

Now we consider the Taylor series expansion around $x_{1}$ of the term 
\begin{eqnarray}\label{eq:gen1}
&& d E_{z^*_{1}} \left [\vphantom{\frac{1}{2}}\left ( V\left (x_{1}+ \frac{z^*_{1}\ell}{\sqrt{d}}\right ) 
- V \left (x_{1}\right ) \right) \mathbb{W} \left (z^*_{1}, x_{1},d\right ) \right] \nonumber \\
&=& d E_{z^*_{1}} \left [ \left \{ V^{'}(x_{1})\frac{z^*_{1}\ell}{\sqrt{d}} 
+ \frac{1}{2} V^{''}(x_{1})\frac{{z^*_{1}}^{2}\ell^{2}}{d} 
+ \frac{1}{6}V^{'''}(\xi_{1})\frac{{z^*_{1}}^{3}\ell^{3}}{d^{\frac{3}{2}}} \right \} 
\mathbb{W} \left (z^*_{1}, x_{1},d\right) \right ]. \nonumber \\
\end{eqnarray}

From (\ref{eq:W}) it is clear that $\mathbb{W}(z^*_{1}, x_{1}, d)$ is continuous but not differentiable 
at the point $0$. So, this can not be expanded as a Taylor series around $0$. Also, note that 
$\mathbb{W}$ is an almost surely bounded function with respect to $d$. This follows from the fact that 
$\Phi$ is a bounded function and that $\eta(x_{1},z^*_{1},d)\stackrel{a.s.}{\longrightarrow}0$ as 
$d \rightarrow \infty$. The latter is easily proved by showing, as in (\ref{eq:as_conv}), that each term of
$\eta(x_{1},z^*_{1},d)$ tends to zero almost surely; here we need to use the facts that $E_f\left[\{\log(f(x_1))\}'\right]^4<\infty$,
$E_f\left[\{\log(f(x_1))\}''\right]^2<\infty$ and $E_f\left[\left\vert\{\log(f(\xi_1))\}'''\right\vert\right]<\infty$ which follow
from assumptions 
(13), (14) and (15) of DB.
By expanding the individual terms in the expression in (\ref{eq:W}) we obtain, for appropriate $w_1,w_2,\xi_1$, the following: 

\begin{eqnarray}
\Phi \left (\frac{\eta (x_{1}, z^*_{1}, d) - \frac{{z^*_{1}}^{2}\ell^{2}}{2}\mathbb{I}}
{\sqrt{{z^*_{1}}^{2}\ell^{2}\mathbb{I}}}\right )  &=& \Phi \left (-\frac{\sqrt{{z^*_{1}}^{2}\ell^{2}\mathbb{I}}}{2} \right)+ \frac{1}{\sqrt{d\mathbb{I}}}[ \log f(x_{1})]'\phi\left (-\frac{\sqrt{{z^*_{1}}^{2}\ell^{2}\mathbb{I}}}{2} \right) + \frac{1}{2d\mathbb{I}}\phi^{'}(w_{1}), \nonumber \\
\hspace{0.5 cm}\Phi \left (\frac{-\eta (x_{1}, z^*_{1}, d) - \frac{{z^*_{1}}^{2}\ell^{2}}{2}\mathbb{I}}
{\sqrt{{z^*_{1}}^{2}\ell^{2}\mathbb{I}}}\right ) & =&  \Phi \left (-\frac{\sqrt{{z^*_{1}}^{2}\ell^{2}\mathbb{I}}}{2} \right) - \frac{1}{\sqrt{d\mathbb{I}}}[ \log f(x_{1})]'\phi\left (-\frac{\sqrt{{z^*_{1}}^{2}\ell^{2}\mathbb{I}}}{2} \right) + \frac{1}{2d\mathbb{I}}\phi^{'}(w_{2}), \nonumber \\
 e^{\eta (x_{1}, z^*_{1}, d )} &=& 1+ \frac{\ell z^*_{1}}{\sqrt{d}} \left [\log f(x_{1}) \right ]^{'} 
+ \frac{\ell^{2}{z^*_{1}}^{2}}{2!d}\left [\log f(x_{1}) \right] ^{''} + \frac{\ell^{3}{z^*_{1}}^{3}}
{3!d^{\frac{3}{2}}}\left [\log f(\xi_{1}) \right] ^{'''}. \nonumber \\
\end{eqnarray}

Using these expanded forms and then simplifying the expression in (\ref{eq:gen1}), we obtain the following form of $G_dV(x)$: 

\begin{equation*}
G_{d}{V(x)} = 
V^{'}(x_{1})\frac{1}{2}\ell^{2} (\log f(x_1))'E_{z^*_{1}} \left [ {z^*_{1}}^{2}\mathbb{\mathcal{H}} \left (z^*_{1}\right ) 
\right] 
+ \frac{1}{2}V^{''}(x_{1})\ell^{2}E_{z^*_{1}} \left [ {z^*_{1}}^{2}\mathbb{\mathcal{H}} \left (z^*_{1} \right ) 
\right]+O\left(d^{-\frac{1}{2}}\right),
\end{equation*}
where

\begin{equation}\label{eq:h}
\mathbb{\mathcal{H}}(z^*_{1}) = 2 \Phi \left (- \frac{|z^*_{1}|\ell\sqrt{\mathbb{I}}}{2}\right)
=2\left[1-\Phi \left (\frac{|z^*_{1}|\ell\sqrt{\mathbb{I}}}{2}\right)\right].
\end{equation}
%
%

Hence, the limiting form of our generator is Langevin and  is given by
\begin{equation}
GV(x)=\frac{1}{2}g(\ell) (\log f(x_{1}))'V'(x_1)
+\frac{g(\ell)}{2}V''(x_1),
\label{eq:diff_equation_iid}
\end{equation}
where $g(\ell)$ is given by (18) of DB. 
Since $G_dV(x)$ and $V''(x_1)$ are bounded, and $G_dV(x)$ converges pointwise to
$GV(x)$, Dominated Convergence Theorem implies that
\[ \lim_{d\rightarrow\infty}E\left\vert G_dV(x)-GV(x)\right\vert\rightarrow 0.\]

\end{proof}

\subsection{Proof of Theorem 3.2 of DB} 
\label{proof:theorem2}
\begin{proof}
We can write down the generator $G_{d}V(x)$ as follows: 

\begin{eqnarray}
G_{d}{V(x)} &=& \frac{d}{2^d}P (\mathbb{\chi}_{1}=1) 
\int_{0}^{\infty}\sum_{b_{1}\in \{-1,+1\}} \left 
[\left (\vphantom{\min \left \{ 1, \frac{\pi(x_{1}+b_{1}\epsilon, \ldots, x_{d}+b_{d}\epsilon)}{\pi(x_{1},\ldots, x_{d})} \right \}} V(x_{1}+b_{1}\epsilon) - V (x_{1}) \right)  \right. \nonumber \\
&& \qquad \quad \left. \hspace{2 cm} \times E_{\left \{\begin{array}{l} 
b_{2},b_{3},\ldots,b_{d}, \\
\mathbb{\chi}_{2},\mathbb{\chi}_{3},\ldots,\mathbb{\chi}_{d} 
\end{array} \right \}}\left ( \normalsize \min \left \{ 1, \frac{\pi(x_{1}+b_{1}\epsilon,\ldots,x_{d}+ \mathbb{\chi}_{d}b_{d}\epsilon)}{\pi(x_{1},\ldots, x_{d})} \right \} \right ) \right ] q(\epsilon)d\epsilon. \nonumber \\
\label{eq:gibbs_tmcmc}
\end{eqnarray}

Note that since $V$ is a function of $x_{1}$ only, if $\mathbb{\chi}_{1}$ is equal to 0, 
then no transition takes place and $V(x_{1}+\chi_1b_{1}\epsilon) - V (x_{1})=0$, so that 
the value of the generator is 0. 
In other words, the part of the generator associated with $P (\mathbb{\chi}_{1}=0)$ is zero,
and hence does not feature in (\ref{eq:gibbs_tmcmc}).

Since $b_{j}$ and $\mathbb{\chi}_{j}$ always occur as products,
we have
\newcommand{\I}{\mathbb{\chi}}

\begin{equation}
E_{\left \{\begin{array}{l} 
b_{2},b_{3},\ldots,b_{d}, \\
\mathbb{\chi}_{2},\mathbb{\chi}_{3},\ldots,\mathbb{\chi}_{d} 
\end{array} \right \}} = E_{\left \{b_{2}\I_{2}, b_{3}\I_{3},\ldots, b_{d}\I_{d}\right \}}
\end{equation}

Our approach to obtaining the diffusion limit in this problem will be similar to that 
in the previous problem, where all the components of $x$ are updated simultaneously at every iteration
of TMCMC. 
Here we leave $(1-c_{d})(d-1)$ terms unchanged at each step and sum over 
$c_{d}d$ many terms inside the exponential. We make a very vital assumption that 
$c_{d} \rightarrow c$, which forces $c_{d}(d-1)$ to go to $\infty$ as $d \rightarrow \infty$. 
We apply Lyapunov's central limit theorem as before (again the Lyapunov assumption holds good for $\delta=4$), 
to obtain, 
given any $\omega\in\mathbb N^c$, 
\[
\left\vert E_{b_2\chi_2,\ldots,b_d\chi_d} \left [\min \left \{ 1, \frac{\pi(x_{1}+b_{1}\epsilon, \ldots, x_{d}+\chi_db_{d}\epsilon)}{\pi(x_{1},\ldots, x_{d})} \right \} 
\right ]-E_{b_2\chi_d,\ldots,b_d\chi_d}\left[\min\left\{1,e^X\right\}\right]\right\vert\rightarrow 0,\hspace{2mm}\mbox{as}\hspace{2mm}d\rightarrow\infty,
\]
where, using Lyapunov's theorem and the same techniques as before, we obtain
\begin{equation}
X \sim N \left( \eta (x_{1}, b_{1}, \epsilon) - \frac{(c_{d}d-1) \epsilon^{2}}{2}\mathbb{I}, 
\hspace{0.1 cm}(c_{d}d-1)\epsilon^{2}\mathbb{I} \right).
\end{equation}

Analogously, we define $\mathbb{W}(x_1, z^*_{1}, c_{d}, d)$ as the following
\begin{equation}\label{eq:cW}
\mathbb{W}(x_1, z^*_{1}, c_{d}, d) =  \Phi \left (\frac{\eta (x_{1}, z^*_{1}) - 
\frac{{z^*_{1}}^{2}\ell^{2}c_{d}}{2}\mathbb{I}}{\sqrt{{z^*_{1}}^{2}\ell^{2}c_{d}\mathbb{I}}}\right ) 
+ e^{\eta (x_{1}, z^*_{1} )}
\Phi \left (\frac{-\eta (x_{1}, z^*_{1}) -
\frac{{z^*_{1}}^{2}\ell^{2}c_{d}}{2}\mathbb{I}}{\sqrt{{z^*_{1}}^{2}\ell^{2}c_{d}\mathbb{I}}}\right ).
\end{equation}

Proceeding in the same way as in the previous case, we obtain

\begin{equation}
G_{d}{V(x)} = V^{'}(x_{1}) \frac{1}{2}c_{d}\ell^{2} E_{z^*_{1}} \left [ {z^*_{1}}^{2}\mathbb{\mathcal{H}} \left (z^*_{1}, x_{1},c_{d},d \right ) 
\right] 
+ \frac{1}{2}V^{''}(x_{1})c_{d}\ell^{2}E_{z^*_{1}} \left [ {z^*_{1}}^{2}\mathbb{\mathcal{H}} \left (z^*_{1}, x_{1},c_{d},d \right ) 
\right]+O\left(d^{-\frac{1}{2}}\right),
\end{equation}
where
\begin{equation}
\mathbb{\mathcal{H}}(z^*_{1},x_{1},c_{d}) =
2 \Phi \left (- \frac{|z^*_{1}|\ell \sqrt{c_d\mathbb{I}}}{2}\right)
=2 \left[1-\Phi \left (\frac{|z^*_{1}|\ell \sqrt{c_d\mathbb{I}}}{2}\right)\right].
\end{equation}

Finally, the limiting form of the generator in this case of partial updating based additive 
TMCMC turns out to be analogous to the previous case where all the components of $x$ are updated
simultaneously at every step. This is given by 
\begin{equation}
GV(x)=\frac{1}{2}g_{c}(\ell) (\log f(x_{1}))'V'(x_1)
+ \frac{g_{c}(\ell)}{2}V''(x_1),
\label{eq:diff_equation_tmcmc_gibbs}
\end{equation}
where the diffusion speed $g_c(\ell)$ is given by
\begin{equation}
g_c(\ell)=4c\ell^{2}\int_0^{\infty} u^{2}\Phi 
\left (- \frac{u\ell\sqrt{c\mathbb{I}}}{2}\right) \phi(u)du.
\label{eq:diff_speed_tmcmc_gibbs_appendix}
\end{equation}
As before, the Dominated Convergence Theorem implies that
\[ \lim_{d\rightarrow\infty}E\left\vert G_dV(x)-GV(x)\right\vert\rightarrow 0.\]

\end{proof}

\subsection{Proof of Theorem 4.1 of DB} 
\label{proof:theorem4}
\begin{proof}
The generator function of the process can be written as 
\begin{eqnarray}\label{eq:genbedard}
G_{d}{V(x)} &=& \frac{d^{\alpha}}{2^d}\int_{0}^{\infty}\sum_{b_{1}\in \{-1,+1\}} 
\left [\left (\vphantom{\min \left \{ 1, \frac{\pi(x_{1}+b_{1}\epsilon, \ldots, x_{d}+b_{d}\epsilon)}
{\pi(x_{1},\ldots, x_{d})} \right \}} V(x_{1}+b_{1}\epsilon) - V (x_{1}) \right)  \right. \nonumber \\
&& \qquad \quad \left. \hspace{2 cm} \times E_{b_{2}, \ldots, b_{d}}
\left ( \normalsize \min \left \{ 1, \frac{\pi(x_{1}+b_{1}\epsilon, \ldots, x_{d}+b_{d}\epsilon)}
{\pi(x_{1},\ldots, x_{d})} \right \} \right ) \right ] q(\epsilon)d\epsilon, \nonumber \\
\end{eqnarray}
%
where
\begin{eqnarray}\label{eq:mineqBedard}
& E_{b_2,\ldots,b_d} & \left [\min \left \{ 1, \frac{\pi(x_{1}+b_{1}\epsilon, \ldots, x_{d}+b_{d}\epsilon)}
{\pi(x_{1}, \ldots, x_{d})} \right \} \right ]\nonumber  \\
&=&E_{b_2,\ldots,b_d} \left  [\min \left \{ 1, \exp \left (  \vphantom{\sum_{j=2}^{d} 
\left \{ b_{j}\epsilon \left \{ \log(f(x_{j}))\right \}^{'} +  \frac{\epsilon^{2}}{2!} 
\left \{ \log(f(x_{j}))\right \}^{''} +  \frac{b_{j}\epsilon^{3}}{3!} 
\left \{ \log(f(z_{j}))\right \}^{'''} \right \}} 
{\log (f({x_{1}+b_{1}\epsilon})) - \log (f(x_{1})) }\right. \right. \right. \nonumber \\
&&   \left. \left. \left. + \sum_{j=2}^{k} \left \{ b_{j}\epsilon 
\left \{ \log(f(\theta_{j}(d)x_{j}))\right \}^{'} +  \frac{\epsilon^{2}}{2!}
\left \{ \log(f(\theta_{j}(d)x_{j}))\right \}^{''} +  \frac{b_{j}\epsilon^{3}}{3!} 
\left \{ \log(f(\theta_{j}(d)x_{j}))\right \}^{'''}\right.\right.\right.\right.\nonumber\\ 
&&\qquad\left. \left. \left. \left. + \frac{\epsilon^{4}}{4!}\left \{ \log(f(\theta_{j}(d)z_{j}))\right \}^{''''}\right\} \right. \right. \right. \nonumber \\
&&  \left. \left. \left. +\sum_{j=k+1}^{d} \left \{ b_{j}\epsilon 
\left \{ \log(f(\theta_{j}(d)x_{j}))\right \}^{'} +  \frac{\epsilon^{2}}{2!} 
\left \{ \log(f(\theta_{j}(d)x_{j}))\right \}^{''} +  \frac{b_{j}\epsilon^{3}}{3!} 
\left \{ \log(f(\theta_{j}(d)x_{j}))\right \}^{'''}\right.\right.\right.\right.\nonumber\\ 
&&\qquad\left.\left.\left.\left.+\frac{\epsilon^{4}}{4!}\left \{ \log(f(\theta_{j}(d)z_{j}))\right \}^{''''}
\right \} \right ) \right \} \right ] \nonumber \\
\end{eqnarray}

Note that since $\epsilon$ can be represented, as before, as 
$\frac{\ell z^*_{1}}{d^{\frac{\alpha}{2}}}$ (where we assume that $\alpha > 0$), and, 
due to assumptions 
(13), (14), (15)
and (34) of DB, and 
because $k$ is finite, 
it is easy to see that the first sum in the expression in (\ref{eq:mineqBedard}) goes to 0 almost surely. 
Then, we apply Lyapunov's central limit theorem on $b_{j}$ for $j= k+1,\ldots,d$, 
which deals with infinitely many random variables as $d \rightarrow \infty$, and 
we obtain, for every fixed $\omega\in\mathbb N^c$, where $\mathbb N$ is an appropriate null set as before, 
\begin{equation}\label{eq:dist3}
\frac{\sum_{j=k+1}^{d}{ b_{j}\left [\epsilon \left \{\log (f(\theta_j(d)x_{j}))\right \}^{'}+
\frac{\epsilon^{3}}{6} \left \{ \log(f(\theta_j(d)x_{j}))\right \}^{'''} \right ]}}{\sqrt {\sum_{j=k+1}^{d} 
\left [  \epsilon \{\log (f(\theta_j(d)x_{j}))\right \}^{'}+\frac{\epsilon^{3}}{6}\left \{ \log(f(\theta_j(d)x_{j}))\right \}^{'''}]^{2} }} 
\stackrel{\mathcal L}{\rightarrow} N(0,1). 
\end{equation}


The square of the denominator of (\ref{eq:dist3}) can be written as 
$\epsilon^2\sum_{j=k+1}^{d}{\left[\left \{\log (f(\theta_{j}(d)x_{j}))\right \}^{'}\right]^{2} } + \Delta$,
where
\begin{equation}\label{eq:Bedarddelta}
\Delta= \frac{\epsilon^{4}}{6}\sum_{j=k+1}^{d}{2 \left \{\log (f(\theta_{j}(d)x_{j}))\right \}^{'}
\left \{\log (f(\theta_{j}(d)x_{j}))\right \}^{'''}}+ \frac{\epsilon^{6}}{36}\sum_{j=k+1}^{d}
{\left[\left \{\log (f(\theta_{j}(d)x_{j}))\right \}^{'''}\right]^{2}}. 
\end{equation}
Representing $\epsilon$ as $\frac{\ell z^*_{1}}{d^{\frac{\alpha}{2}}}$, it can be seen as before that 
$\Delta\stackrel{a.s.}{\longrightarrow}0$ as $d\rightarrow\infty$.
%
Writing $u_j=\theta_j(d)x_j$, we have
\begin{eqnarray}
\epsilon^{2}\sum_{j=k+1}^{d}{\left[\left \{\log (f(\theta_{j}(d)x_{j}))\right \}^{'}\right]^{2} } 
&=&\sum_{i=1}^m\frac{\ell^{2}{z^*_{1}}^{2}}{d^{\alpha}}\theta^2_j(d)r(i,d)
\left\{\frac{1}{r(i,d)}\sum_{j=1}^{r(i,d)}\left(\frac{f'(u_j)}{f(u_j)}\right)^2\right\}\nonumber\\
&=&\sum_{i=1}^{m}\frac{\ell^{2}{z^*_{1}}^{2}d^{\gamma_{i}}r(i,d)}{K_{k+i}d^{\alpha}}r(i,d)
\left\{\frac{1}{r(i,d)}\sum_{j=1}^{r(i,d)}\left(\frac{f'(u_j)}{f(u_j)}\right)^2\right\}.
\label{eq:simplify_sum}
\end{eqnarray}
As $d\rightarrow\infty$, almost surely, 
\begin{equation}
\frac{1}{r(i,d)}\sum_{j=1}^{r(i,d)}\left(\frac{f'(u_j)}{f(u_j)}\right)^2 
\rightarrow E\left [\left\{\frac{f^{'}(U)}{f(U)}\right \}^{2} \right]=\mathbb{I}.
\label{eq:convergence_information_appendix}
\end{equation}
Hence, as $d\rightarrow\infty$, almost surely,
\begin{eqnarray}
\epsilon^{2}\sum_{j=k+1}^{d}{\left[\left \{\log (f(\theta_{j}(d)x_{j}))\right \}^{'}\right]^{2} } 
&\rightarrow& \ell^{2}{z^*_{1}}^{2}\xi^{2}\mathbb{I},
\end{eqnarray}
where
\[\xi^2=\lim_{d\rightarrow\infty}\sum_{i=1}^m\frac{d^{\gamma_{i}}r(i,d)}{K_{k+i}d^{\alpha}}\]
is finite due to (34) of DB 
and the fact that $m$ is finite.
Also, as before, the fourth order terms involved in (\ref{eq:mineqBedard}) go to zero almost surely.
Hence, given any $\omega\in\mathbb N^c$, 
\[
\left\vert E_{b_2,\ldots,b_d} \left [\min \left \{ 1, \frac{\pi(x_{1}+b_{1}\epsilon, \ldots, x_{d}+b_{d}\epsilon)}{\pi(x_{1},\ldots, x_{d})} \right \} 
\right ]-E_{b_2,\ldots,b_d}\left[\min\left\{1,e^X\right\}\right]\right\vert\rightarrow 0,\hspace{2mm}\mbox{as}\hspace{2mm}d\rightarrow\infty,
\]
where
\begin{equation}
X \sim N \left( \eta (x_{1}, b_{1}, \epsilon) - \frac{(d-1) 
\epsilon^{2}}{2}\xi^2\mathbb{I}, \hspace{0.1 cm} (d-1)\epsilon^{2}\xi^2\mathbb{I} \right).
\end{equation}

We then follow a similar approach as in the previous two cases to obtain 

\begin{equation}
\mathbb{W} \left (z^*_{1}, x_{1},d, \xi \right ) = \Phi \left (\frac{\eta (x_{1}, z^*_{1}, d) - 
\frac{{z^*_{1}}^{2}\ell^{2}\xi^{2}}{2}\mathbb{I}}{\sqrt{{z^*_{1}}^{2}\ell^{2}\xi^{2}\mathbb{I}}}\right ) 
+ e^{\eta (x_{1}, z^*_{1} )}
\Phi \left (\frac{-\eta (x_{1}, z^*_{1}, d) -
\frac{{z^*_{1}}^{2}\ell^{2}\xi^{2}}{2}\mathbb{I}}{\sqrt{{z^*_{1}}^{2}\ell^{2}\xi^{2}\mathbb{I}}}\right ).
\end{equation}

This expression when simplified yields the following expression for the generator term:
\begin{equation}
G_{d}{V(x)} = V^{'}(x_{1}) \frac{1}{2}\ell^{2} E_{z^*_{1}} \left [ {z^*_{1}}^{2}\mathbb{\mathcal{H}} \left (z^*_{1}, x_{1},\xi \right ) 
\right] 
+ \frac{1}{2}V^{''}(x_{1})\ell^{2}E_{z^*_{1}} \left [ {z^*_{1}}^{2}\mathbb{\mathcal{H}} \left (z^*_{1}, x_{1},\xi \right ) 
\right] + O\left(d^{-\frac{1}{2}}\right),
\end{equation}
where 
\begin{equation}
\mathbb{\mathcal{H}}(z^*_{1},x_{1},\xi) =
2 \Phi \left (- \frac{|z^*_{1}|\ell\xi \sqrt{\mathbb{I}}}{2}\right)
=2 \left[1-\Phi \left (\frac{|z^*_{1}|\ell\xi \sqrt{\mathbb{I}}}{2}\right)\right].
\end{equation}

By the same arguments as in the previous cases, we have
\[ \lim_{d\rightarrow\infty}E\left\vert G_dV(x)-GV(x)\right\vert\rightarrow 0,\]
where
\begin{equation}
GV(x)=\frac{1}{2}g_{\xi}(\ell) (\log f(x_{1}))'V'(x_1)
+ \frac{g_{\xi}(\ell)}{2}V''(x_1),
\label{eq:diff_equation_tmcmc_bedard}
\end{equation}
with
\begin{equation}
g_{\xi}(\ell)= 4\ell^{2}\int_0^{\infty} \left \{{u}^{2}\Phi 
\left (- \frac{u\ell\xi\sqrt{\mathbb{I}}}{2}\right) \right \}\phi(u)du.
\label{eq:diff_speed_non_iid_appendix}
\end{equation}

\end{proof}

\section{Calculations related to the dependent set-up}
\label{sec:dependent_calculation}

\subsection{Verification of the conditions of Lyapunov's central limit theorem}
\label{subsec:dep_clt}

To apply Lyapunov's central limit theorem we need to show the following: 
with probability 1 with respect to $\pi$, 
\begin{equation}
\frac{\sum_{j=1}^dE\left(\frac{b_j\eta_j}{\sqrt{d}}\right)^4}
{\left(\sqrt{\sum_{j=1}^d\frac{\eta^2_j}{d}}\right)^4}
=\frac{\sum_{j=1}^d\frac{\eta^4_j}{d^2}}{\left(\frac{\|\eta\|^2}{d}\right)^2}
\rightarrow 0, \ \ \mbox{as} \ \ d\rightarrow\infty.
\label{eq:Lyapunov_general}
\end{equation}
By Lemma 5.2 of \ctn{Pillai2011}, $\frac{\|\eta\|^{2}}{d}\rightarrow 1$ $\pi$-almost surely as 
$d \rightarrow \infty$. 
This implies that
the denominator of the left hand side of (\ref{eq:Lyapunov_general})
goes to 1 $\pi$-almost surely, as $d\rightarrow\infty$. 
Now, $\left(\frac{\|\eta\|^2}{d}\right)^2=\sum_{j=1}^d\frac{\eta^4_j}{d^2}
+\sum_{i=1}^d\frac{\eta^2_i}{d}\left(\underset{j\neq i}\sum\frac{\eta^2_j}{d}\right)$.
Except on a $\pi$-null set $\mathbb N$, where $\sum_{j=1}^d\frac{\eta^2_j}{d}$ need not converge to 1, we have, 
for given $\zeta>0$ and $d_0$ depending upon $\zeta$, 
$1-\zeta<\underset{j\neq i}\sum\frac{\eta^2_j}{d}<1+\zeta$ and 
$1-\zeta<\sum_{i=1}^d\frac{\eta^2_i}{d}<1+\zeta$, for $d\geq d_0$. Hence,
for $d\geq d_0$, 
$-\zeta^2-2\zeta<\zeta^2-2\zeta=(1-\zeta)^2-1<
\sum_{i=1}^d\frac{\eta^2_i}{d}\left(\underset{j\neq i}\sum\frac{\eta^2_j}{d}\right)-1
<(1+\zeta)^2-1=\zeta^2+2\zeta$, so that 
$\left\vert \sum_{i=1}^d\frac{\eta^2_i}{d}\left(\underset{j\neq i}\sum\frac{\eta^2_j}{d}\right)-1\right\vert<\zeta^2+2\zeta$,
showing that $\sum_{i=1}^d\frac{\eta^2_i}{d}\left(\underset{j\neq i}\sum\frac{\eta^2_j}{d}\right)\rightarrow 1$ on $\mathbb N^c$,
the complement of $\mathbb N$. 
Since on $\mathbb N^c$, $\left(\frac{\|\eta\|^2}{d}\right)^2\rightarrow 1$, we must have
$\sum_{j=1}^d\frac{\eta^4_j}{d^2}\rightarrow 0$ on $\mathbb N^c$, showing that Lyapunov's condition (\ref{eq:Lyapunov_general}) holds
almost surely with respect to $\pi$.

Using Lyapunov's central limit theorem on $b_{j}$, and using the result that 
$\frac{\|\eta\|^{2}}{d}\rightarrow 1$ $\pi$-almost surely as
$d \rightarrow \infty$,
we obtain, for sufficiently large $d$,
\begin{equation}
R(x,\xi) \sim AN ( -\ell^{2}\epsilon^{2}, 2\ell^{2}\epsilon^{2}), 
\end{equation}
where $``AN"$ stands for ``asymptotic normal".

Now, (57) of DB 
and the fact that for large $d$, $\mathbb{Q}(x,\xi)\approx R(x,\xi)$, imply 

\begin{equation}
\mathbb{Q}(x,\xi)\approx -\epsilon \sqrt{\frac{2\ell^{2}}{d}} \left (\eta_{i}b_{i} + 
\sum_{j=1, j\neq i}^{d} {\eta_{j}b_{j}} \right) - \ell^{2}\epsilon^{2}, 
\end{equation}
so that
\begin{equation}
\left [\mathbb{Q}(x,\xi) \vert b_{i}, \epsilon \right ] \sim AN \left ( -\ell^{2}\epsilon^{2} 
-\epsilon\sqrt{\frac{2\ell^{2}}{d}}\eta_{i}b_{i}, 2\ell^{2}\epsilon^{2}\right ). 
\label{eq:Q_CLT}
\end{equation}

\subsection{Expected drift}
\label{subsec:expected_drift}

In order to obtain the diffusion approximation, we first obtain the expected drift conditions. In order to do that, 
we first define, as in \ctn{Pillai2011}, $\mathcal F_k$ to be the sigma algebra generated by 
$\{x^n,\xi^n,\gamma^k,n\leq k\}$, and denote the conditional expectations 
$E(\cdot\vert\mathcal F_k)$ by $E_k(\cdot)$. Following \ctn{Pillai2011} we let $x^0=x$ and $\xi^1=\xi$,
and set $\xi^0=0$ and $\gamma^0=0$.
We then note that 
under stationarity, $E_{k} \left ( x^{k+1} -x^{k} \right ) = E_{0} \left ( x^{1} -x \right )$, 
and using (51) of DB 
we can write 
\begin{eqnarray}\label{eq:onestepexpdiff}
d E_{0} \left ( {x_{i}}^{1} -  {x_{i}} \right ) & =& d  {E_{0} \left[ {\gamma}^{1} 
\left ({y_{i}}^{1} - {x_{i}}^{1} \right ) \right] }\nonumber \\
&=& {d} { E_{0} \left [ \alpha(x, \xi) \sqrt{\frac{2\ell^{2}}{d}}\left ({\Sigma}^{\frac{1}{2}}\xi \right )_{i} \right] } \nonumber \\
&=& \frac{1}{\eta_i} { {\lambda_{i}}\sqrt{2\ell^{2}d} E_{0} \left [ \min \left \{ 1, e^{\mathbb{Q}(x,\xi)} \right \} 
\xi_{i} \right] \eta_{i}},  \nonumber \\
\end{eqnarray}
where  $\alpha(x, \xi) = \min \left \{ 1, \frac{\pi({y_{i}}^{1})}{\pi({x_{i}}^{1})} \right \}$. 
The last step follows from (47) of DB, 
noting that $\xi=\sum_{i=1}^d\xi_i\phi_i$. 

Noting that $\lambda_{i}\Sigma^{-\frac{1}{2}}\phi_{i}= \phi_{i}$, 
(55) and self-adjointness of $\Sigma^{-1/2}$ yields 
\begin{eqnarray}\label{eq:lambda}
\lambda_{i}\eta_{i} &=& \lambda_{i} \left \langle \Sigma^{-\frac{1}{2}}(P^{d}x) 
+ \Sigma^{\frac{1}{2}}\nabla\Psi^{d}(x), \phi_{i} \right \rangle \nonumber \\
&=&\lambda_{i} \left \langle \Sigma^{-\frac{1}{2}}(P^{d}x)
+ \Sigma^{-\frac{1}{2}}\Sigma\nabla\Psi^{d}(x), \phi_{i} \right \rangle \nonumber \\
&=& \left  \langle P^{d}x +\Sigma^d\nabla\Psi^{d}(x), \phi_{i} \right \rangle \nonumber \\
&=& \left (P^{d}x + \Sigma^d\nabla\Psi^{d}(x) \right)_{i}. \nonumber \\
\end{eqnarray}


Thus, we can write 
\begin{eqnarray}
d E_{0} \left ( {x_{i}}^{1} -  {x_{i}} \right ) &= & \frac{1}{\eta_{i}} 
\left (P^{d}x + \Sigma^d\nabla\Psi^{d}(x) \right)_{i} \sqrt{2\ell^{2}d} E_{0} 
\left [ \min \left \{ 1, e^{\mathbb{Q}(x,\xi)} \right \} \xi_{i} \right].  \label{eq:drift0}\nonumber \\
\end{eqnarray}

Now, writing $\mu=-\ell^2\epsilon^2-\epsilon\sqrt{\frac{2\ell^2}{d}}\eta_ib_i$, $\sigma=\sqrt{2}\ell\epsilon$, 
using (\ref{eq:Q_CLT}) and Proposition 2.4 of \ctn{Roberts97a}, it follows that \\[1cm]
$\sqrt{d}E_0\left [ \min \left \{ 1, e^{\mathbb{Q}(x,\xi)} \right \} \xi_{i} \right]$
\begin{eqnarray}
&=&\sqrt{d} E_{b_{i}\epsilon} \left [ b_{i}\epsilon E_0 \left \{\min \left \{ 1, e^{\mathbb{Q}(x,\xi)} \right \}\bigg| b_{i}, 
 \epsilon \right \} \right ]\nonumber\\
&\approx& \sqrt{d}E_{b_i\epsilon}\left[b_i\epsilon\left\{\Phi\left(\frac{\mu}{\sigma}\right)
+e^{\mu+\frac{\sigma^2}{2}}\Phi\left(-\sigma-\frac{\mu}{\sigma}\right)\right\}\right]\nonumber\\ 
&=& \sqrt{d}E_{b_i\epsilon}\left[b_i\epsilon\left\{\Phi\left(-\frac{\ell\epsilon}{\sqrt{2}}-\frac{\eta_{i}b_{i}}{\sqrt{d}}\right)
\right.\right.\nonumber\\
&&\quad\quad\quad\left.\left.+e^{-\epsilon\sqrt{\frac{2\ell^2}{d}}\eta_ib_i}
\Phi\left(-\frac{\ell\epsilon}{\sqrt{2}} + \frac{\eta_{i}b_{i}}{\sqrt{d}} \right)\right\}\right].\nonumber\\ 
\label{eq:drift1}
\end{eqnarray} 
Using the following Taylor's series expansions
\begin{eqnarray}
\Phi\left(-\frac{\ell\epsilon}{\sqrt{2}}-\frac{\eta_ib_i}{\sqrt{d}}\right)
&=&\Phi\left(-\frac{\ell\epsilon}{\sqrt{2}}\right)-\frac{\eta_ib_i}{\sqrt{d}}\phi\left(-\frac{\ell\epsilon}{\sqrt{2}}\right)
+\frac{\eta^2_i}{2d}\phi'(w_1),\nonumber\\
\Phi\left(-\frac{\ell\epsilon}{\sqrt{2}}+\frac{\eta_ib_i}{\sqrt{d}}\right)
&=&\Phi\left(-\frac{\ell\epsilon}{\sqrt{2}}\right)+\frac{\eta_ib_i}{\sqrt{d}}\phi\left(-\frac{\ell\epsilon}{\sqrt{2}}\right)
+\frac{\eta^2_i}{2d}\phi'(w_2),\nonumber\\
e^{-\epsilon\sqrt{\frac{2\ell^2}{d}}\eta_ib_i}
&=& 1-\epsilon\sqrt{\frac{2\ell^2}{d}}\eta_ib_i+\frac{\ell^2\epsilon^2\eta^2_i}{d}e^{-w_3},
\label{eq:taylor_series}
\end{eqnarray}
where $w_1$ lies between $-\frac{\ell\epsilon}{\sqrt{2}}$ and $-\frac{\ell\epsilon}{\sqrt{2}}-\frac{\eta_ib_i}{\sqrt{d}}$;
$w_2$ lies between $-\frac{\ell\epsilon}{\sqrt{2}}$ and $-\frac{\ell\epsilon}{\sqrt{2}}+\frac{\eta_ib_i}{\sqrt{d}}$, and
$w_3$ lies between $0$ and $\epsilon\sqrt{\frac{2\ell^2}{d}}\eta_ib_i$, and noting that 
$E_{b_i\epsilon}\left[b_i\epsilon\Phi\left(-\frac{\ell\epsilon}{\sqrt{2}}\right)\right]=0$,
(\ref{eq:drift1}) can be easily seen to be of the form
\begin{eqnarray}
\sqrt{d}E_0\left [ \min \left \{ 1, e^{\mathbb{Q}(x,\xi)} \right \} \xi_{i} \right]
&\approx& \sqrt{d}E_{b_i\epsilon}\left[b_i\epsilon\left\{\Phi\left(-\frac{\ell\epsilon}{\sqrt{2}}-\frac{\eta_{i}b_{i}}{\sqrt{d}}\right)
\right.\right.\nonumber\\
&&\quad\quad\quad\left.\left.+e^{-\epsilon\sqrt{\frac{2\ell^2}{d}}\eta_ib_i}
\Phi\left(-\frac{\ell\epsilon}{\sqrt{2}} + \frac{\eta_{i}b_{i}}{\sqrt{d}} \right)\right\}\right]\nonumber\\ 
&=& -\sqrt{2\ell^2}\eta_i\times 2\int_0^{\infty}u^2\Phi\left(-\frac{\ell u}{\sqrt{2}}\right)\phi(u)du
+O\left(d^{-\frac{1}{2}}\right)\nonumber\\
&\approx&-\sqrt{\frac{\ell^2}{2}}\eta_i\beta,
\label{eq:drift_2}
\end{eqnarray}
where
\begin{equation}
\beta=4\int_0^{\infty}u^2\Phi\left(-\frac{\ell u}{\sqrt{2}}\right)\phi(u)du.
\label{eq:beta_star}
\end{equation}

Hence, we can re-write (\ref{eq:drift0}) as
\begin{eqnarray}
d E_{0} \left ( {x_{i}}^{1} -  {x_{i}} \right ) &= & \frac{1}{\eta_{i}} 
\left (P^{d}x + \nabla\Psi^{d}(x) \right)_{i} \sqrt{2\ell^{2}d} E_{0} 
\left [ \min \left \{ 1, e^{\mathbb{Q}(x,\xi)} \right \} \xi_{i} \right]\nonumber \\
&=&-\ell^2\beta\left (P^{d}x + \nabla\Psi^{d}(x) \right)_{i}.\label{eq:drift3} 
\end{eqnarray}

\subsection{Expected diffusion coefficient}
\label{subsec:expected_diffusion_coefficient}

Now we evaluate the expected diffusion coefficients involving the cross product terms. For $ 1 \leq i \neq j \leq d$,  we have 

\begin{equation}
d E_{0} \left [ \left ({x_{i}}^{1} -  {x_{i}} \right) \left ({x_{j}}^{1} -  {x_{j}} \right) \right ] 
= d E_{0} \left[ \left \{ { {\gamma}^{1} \left ({y_{i}}^{1} - {x_{i}} \right ) }\right \} 
\left \{ { {\gamma}^{1} \left ({y_{j}}^{1} - {x_{j}} \right )} \right\} \right]\nonumber \\
\end{equation}

Check that if $i \neq j$, then the above expectation is $0$ using the fact that $b_{i}b_{j}\epsilon$ has 0 mean for $i \neq j$. 
However for $i=j$, using  (\ref{eq:lambda}) again, we can reduce the above expectation to 

\begin{eqnarray}\label{eq:cross}
d E_{0} \left [ \left ({x_{i}}^{1} -  {x_{i}} \right) \left ({x_{j}}^{1} -  {x_{j}} \right) \right ] 
& =& d E_{0} \left [ \left ({x_{i}}^{1} -  {x_{i}} \right)^{2} \right ] \nonumber \\
&=& d E_{0} \left [ {\alpha(x, \xi) \left ({y_{i}}^{1} - {x_{i}} \right )^{2}} \right ] \nonumber \\
&=& 2\ell^{2} { {\lambda_{i}}^{2} E_{0} \left [{\xi_{i}}^{2} \min \left \{ 1, e^{\mathbb{Q}(x,\xi)} \right \} \right ] }.\nonumber \\
\end{eqnarray}
Using the same Taylor's series expansions (\ref{eq:taylor_series}) it is easily seen that
\begin{eqnarray}
E_{0} \left [{\xi_{i}}^{2} \min \left \{ 1, e^{\mathbb{Q}(x,\xi)} \right \} \right ] 
&\approx& 4\int_0^{\infty}u^2\Phi\left(-\frac{\ell u}{\sqrt{2}}\right)\phi(u)du\nonumber\\
&=&\beta.
\label{eq:diffusion_coefficient2}
\end{eqnarray}
Hence,
\begin{eqnarray}\label{eq:cross2}
d E_{0} \left [ \left ({x_{i}}^{1} -  {x_{i}} \right) \left ({x_{j}}^{1} -  {x_{j}} \right) \right ] 
&=& 2\ell^{2} { {\lambda_{i}}^{2} E_{0} \left [{\xi_{i}}^{2} \min \left \{ 1, e^{\mathbb{Q}(x,\xi)} \right \} \right ] }\nonumber \\
&\approx & 2\ell^{2}{{\lambda_{i}}^{2} \beta}\nonumber\\  
&=& 2 \ell^{2}\beta \langle \phi_{i}, \Sigma \phi_{i} \rangle. \nonumber \\
\end{eqnarray}
It follows that
\begin{eqnarray}\label{eq:cross3}
d E_{0} \left [ \left ({x}^{1} -  {x} \right) \otimes\left ({x}^{1} -  {x} \right) \right ] 
&\approx& 2 \ell^{2}\beta \Sigma^d. 
\end{eqnarray}


Note that, by definition, 
\[
x^{k+1}=x^k+E_k(x^{k+1}-x^k)+\sqrt{\frac{2\ell^2\beta}{d}}\Gamma^{k+1,d},
\]
where, for $k\geq 0$,
\begin{equation*}
\Gamma^{k+1,d}=\sqrt{\frac{d}{2\ell^2\beta}}\left(x^{k+1}-x^k-E_k\left(x^{k+1}-x^k\right)\right).
\end{equation*}
From (\ref{eq:drift3}), we have, for $d$ large enough,
\begin{equation}
x^{k+1}\approx x^k-\frac{\ell^2\beta}{d}m^d(x^k)+\sqrt{\frac{2\ell^2\beta}{d}}\Gamma^{k+1,d},
\label{eq:approx1}
\end{equation}
where
\begin{equation}
m^d(x)=P^dx+\Sigma^d\nabla\Psi^d(x).
\label{eq:m1}
\end{equation}
From the definition of $\Gamma^{k,d}$ and (\ref{eq:cross3}) we have, as in \ctn{Pillai2011},
\begin{equation}
E_k\left(\Gamma^{k+1,d}\right)=0\quad\mbox{and}\quad 
E_k\left(\Gamma^{k+1,d}\otimes \Gamma^{k+1,d}\right)\approx\Sigma^d.
\label{eq:approx2}
\end{equation}
Thus, for large enough $d$, (\ref{eq:approx1}) can be viewed as the Euler scheme for simulating
the finite dimensional 
approximation
\begin{equation}
x^{k+1}\approx x^k-g(\ell)m^d(x^k)\Delta t+\sqrt{2g(\ell)\Delta t}~\Gamma^{k+1,d}\hspace{2mm}\mbox{where}
\hspace{2mm}\Delta t=\frac{1}{d},
\label{eq:euler}
\end{equation}
(with drift function $m^d$ and covariance operator $\Sigma^d$)
of the SDE 
\begin{equation}
\frac{dz}{dt} = - g(\ell) \left (z + \Sigma\nabla \Psi(z) \right ) + \sqrt{2g(\ell)}\frac{dW}{dt}, \ \ z(0)=z^0,
\label{eq:diffusion_equation_appendix}
\end{equation}
where $z^0\sim\pi$, $W$ is a Brownian motion in a relevant Hilbert space with covariance operator $\Sigma$, and
\begin{equation}
g(\ell)=\ell^2\beta,
\label{eq:general_hl_appendix}
\end{equation}
is the diffusion speed. 
